\documentclass[12pt]{iopart}
\usepackage[utf8]{inputenc}
\usepackage[T1]{fontenc}
\usepackage{footmisc}

\makeatletter
\@addtoreset{footnote}{section}
\makeatother

\usepackage{amssymb}
 \usepackage{amsthm}
 \usepackage{amsfonts}

\expandafter\let\csname equation*\endcsname=\relax
\expandafter\let\csname endequation*\endcsname=\relax
\usepackage{amsmath}
\usepackage{physics}
\usepackage{mathtools}
\usepackage [operators,sets] {cryptocode}

\usepackage{dsfont}
\usepackage{url}
\usepackage{hyperref}
\usepackage{algorithm}
\usepackage{algorithmic}
\usepackage{bm}
\usepackage{enumitem}
\setlist{noitemsep} \usepackage{qcircuit}
\usepackage{multicol}
\usepackage{microtype}
\usepackage{lmodern}
\usepackage{graphics}

\makeatletter
\AtBeginDocument{\def\doi#1{\url{https://doi.org/#1}}}
\makeatother

\usepackage{caption}
\usepackage{subcaption}

\floatstyle{ruled}

\newfloat{protocol}{htb!}{idf}
\floatname{protocol}{Protocol}

\newfloat{resource}{htb!}{idf}
\floatname{resource}{Resource}

\newfloat{simulator}{htb!}{idf}
\floatname{simulator}{Simulator}

\newfloat{problem}{htb!}{idf}
\floatname{problem}{Problem}

\newcommand{\cptp}[1]{\mathsf{#1}} \newcommand{\pd}[1]{\mathcal{#1}}  \newcommand{\sch}[1]{\bm{#1}}

\newcommand{\Id}{\cptp{I}}
\newcommand{\X}{\cptp{X}}
\newcommand{\Z}{\cptp{Z}}
\newcommand{\Y}{\cptp{Y}}
\newcommand{\Ha}{\cptp{H}}

\renewcommand{\1}{\mathds{1}}

\newcommand{\CZ}{\mathsf{CZ}}
\newcommand{\CR}{\mathsf{CR}}

\newcommand{\acc}{\mathsf{Acc}}
\newcommand{\rej}{\mathsf{Rej}}

\newcommand{\TP}{C\cup T}
\newcommand{\Deco}{\cptp D_{O, C}}
\newcommand{\Cdev}[1]{\tilde{\cptp C}_{T, #1}}

\newtheorem{theorem}{Theorem}
\newtheorem{lemma}{Lemma}

\newtheorem{definition}{Definition}
\newtheorem{example}{Example}

\newtheorem{remark}{Remark}

\let\ifcomment\iffalse 

\usepackage{marginnote}
\ifcomment
	\usepackage[paper=a4paper,outer=6cm,marginparwidth=5cm]{geometry}
	\newcommand{\mn}[2]{\textsuperscript{\textcolor{red}{\textsf{\textbf{#1}}}}\marginnote{\tiny \textcolor{red}{\textsf{{\textbf{#1:}} #2}}}}
	\newcommand{\pn}[2]{\textcolor{red}{\textsf{\textbf{#1:} #2}}}
\else
\newcommand{\mn}[2]{}
	\newcommand{\pn}[2]{}
\fi

\makeatletter
\renewcommand\paragraph{\@startsection{paragraph}{4}{\z@}{1ex \@plus.5ex \@minus.1ex}{-0.7em}{\normalfont\normalsize\itshape}}
\makeatother
\begin{document}

\pagestyle{plain}

\title{Unifying Quantum Verification and Error-Detection: Theory and Tools for Optimisations}

\author{Theodoros Kapourniotis$^1$, Elham Kashefi$^{2,3}$, 
Dominik Leichtle$^{2, 3}$, Luka Music$^4$ and Harold Ollivier$^5$}

\address{$^1$Department of Physics, University of Warwick, Coventry CV4 7AL, United Kingdom}
\address{$^2$Laboratoire d’Informatique de Sorbonne Université, CNRS, Sorbonne Université, 75005 Paris, France}
\address{$^3$School of Informatics, University of Edinburgh, Edinburgh EH8 9AB, United Kingdom}
\address{$^4$Quandela, 91300 Massy, France}
\address{$^5$DI-ENS, Ecole Normale Supérieure, PSL University, CNRS, INRIA, 75005 Paris, France}
\ead{luka.music@quandela.com}

\begin{abstract}
With the advent of cloud-based quantum computing, it has become vital to provide strong guarantees that computations delegated by clients to quantum service providers have been executed faithfully. Secure -- blind and verifiable -- Delegated Quantum Computing (SDQC) has emerged as one of the key approaches to address this challenge, yet current protocols lack at least one of the following three ingredients: composability, noise-robustness and modularity.

To tackle this question, our paper lays out the fundamental structure of SDQC protocols, namely mixing two components: the computation which the client would like the server to perform and tests that are designed to detect a server's malicious behaviour. Using this abstraction, our main technical result is a set of sufficient conditions on these components which imply the security and noise-robustness of generic SDQC protocols in the composable Abstract Cryptography framework. This is done by establishing a correspondence between these security properties and the error-detection capabilities of the test computations. Changing the types of tests and how they are mixed with the client's computation automatically yields new SDQC protocols with different security and noise-robustness capabilities.

This approach thereby provides the desired modularity as our sufficient conditions on test computations simplify the steps required to prove the security of the protocols and allows to focus on the design and optimisation of test rounds to specific situations. We showcase this by systematising the search for improved SDQC protocols for Bounded-error Quantum Polynomial-time ($\mathsf{BQP}$) computations. The resulting protocols do not require more hardware on the server's side than what is necessary to blindly delegate the computation without verification, and they outperform all previously known results.

\end{abstract}

\maketitle

\newpage

\section{Introduction}
\label{sec:intro}
\subsection{Context}
Secure delegation of computation is a long-standing topic of research where a client wants to perform a computation on a remote server, without necessarily trusting that it operates as the client desires~\cite{G09:fully,G17:verifiable}. This question was first considered in the quantum context to understand if one could trust the result provided by a quantum machine whose behaviour would be intractable to simulate by classical means~\cite{G04:conference}. It was formalised shortly after by~\cite{A07:scott} as: ``If a quantum computer can efficiently solve a problem, can it also efficiently convince [a classical user] that the solution is correct? More formally, does every language in the class $\mathsf{BQP}$ admit an interactive protocol where the prover is in $\mathsf{BQP}$ and the verifier is in [Bounded-Error Probabilistic Polynomial time] $\mathsf{BPP}$?''. A positive answer would both provide a definitive argument to convince quantum computation sceptics, and have fundamental implications regarding the falsifiability of quantum mechanics~\cite{V07:conference,AV12:is}. This lead to the development of two lines of research directed at finding such protocols.

They first considered a $\mathsf{BPP}$ verifier augmented with a constant number of qubits. These qubits are used either to encrypt the instructions delegated to the server or to perform the computation on a complex resource state provided by the server~\cite{ABE10:interactive,ABEM17:interactive,FK17:unconditionally,B15:how,HM15:verifiable}. More recent protocols map $\mathsf{BQP}$ computations onto the $2$-local Hamiltonian problem. \cite{FHM18:post,HKSE17:direct} remove the need for encryption, but the client must perform $\X$ and $\Z$ measurements. The second line of protocols, more recent yet, is based on the ground breaking work of~\cite{M18:classical}, in which the entirely classical client relies on various post-quantum secure computational assumptions for encryption.

More pragmatically, secure delegated quantum computing protocols provide a way to establish trust between a remote client and a quantum computing service provider. In this context, a delegated computation is deemed (i) blind if neither the data nor the algorithm are revealed to the server, and (ii) verified if a non-aborted computation is guaranteed to have been performed by the server as instructed even if no assumptions are made on the server's behaviour. Ever larger remotely accessible quantum computers are being developed, able to handle more and more complex computations. This has consequently increased the interest for this question: the absence of these cryptographic guarantees will eventually become a major pain point for this nascent industry, precisely when clients will want to use these machines to solve economically interesting problems using specialised algorithms on sensitive data.

Yet, any truly practical protocol must posses at least the following three vital properties. First, the protocol must be composably secure so that it can be safely used in larger applications. Second, it needs to have inherent noise-robustness and a low hardware overhead. Otherwise, clients of a (possibly noisy) quantum computer would have to sacrifice most of its computational power simply to achieve security. And third, the protocol needs to be modular so that its components may either be tailored to the client's and server's needs or optimised to specific use-cases, independently of one another. Such changes should not require an entirely new security proof.

None of the analysed approaches satisfy fully the criteria above. There is in fact a lack of theoretical understanding regarding the requirements for constructing robust and efficient secure delegation protocols, as well as a lack of tools to systematise their optimisation. While several independent protocols optimise either the qubit communication~\cite{KDK15:optimising,Z21:succinct}, the server's hardware overhead~\cite{KW17:optimised,XTH20:improved,LMKO21:verifying}, the set of operations that the client must wield in the protocol~\cite{FKD18:reducing}, or the amount of tolerable noise~\cite{LMKO21:verifying}, none of them provide general methods that could be readily reused and all require security to be proved from scratch.

In this paper, we lay out the foundations for building protocols that provide all three properties presented above. We start by abstracting the main ingredients of prepare-and-send protocols that provide secure delegated quantum computing in the Measurement-Based Quantum Computing (MBQC) framework. Intuitively, the traps -- single qubits in states known only to the client -- used in previous such protocols in fact implement an error-detection scheme. Building on that, we extend the type of traps to cover any computation that can be efficiently simulated by the client, and introduce the concept of trappified schemes -- computations mixed with randomised traps. We then reduce composable security and robustness of protocols using these trappified schemes to simple but powerful sufficient conditions: (i) \emph{detecting deviations from the client's instructions which are potentially harmful for the computation yields verification}, while (ii) \emph{being insensitive to those that are not harmful provides noise-robustness}. Modularity follows naturally from this approach as we show that (i) and (ii) relate to independent properties of trappified schemes. By formally connecting the verification capabilities of protocols to the error-detecting capabilities of their traps, we broaden considerably the sources of inspiration for designing new trappified schemes and therefore lower their overhead.

As a concrete application, we construct a generic compiler for verifying $\mathsf{BQP}$ computations without any overhead of physical resources compared to the blind delegation of the computation while accepting a constant level of global noise. Its efficiency is optimised thanks to the introduction of new tests inspired by syndrome measurements of error-correcting codes. The resulting protocol beats the current state-of-the-art robust SDQC protocol~\cite{LMKO21:verifying} in terms of detection efficiency, which is furthermore independent of the client's desired computation.

\subsection{Overview of Results}
\label{subsec:sum}

In this paper, we express our results in the prepare-and-send model trading generality for simplicity, whereas we rely on the equivalence with the receive-and-measure model to extend their applicability~\cite{WEP22:equivalence}. We chose to focus on protocols with weak quantum clients since protocols with fully classical clients impose a large overhead on the server to ensure the security of the post-quantum scheme they use. In addition, the various initiatives aiming at building metropolitan quantum networks together with the rapid development of photonic quantum computing hold the promise of cheap quantum communications.

In our model, the client prepares a small subset of quantum states, performs limited single-qubit operations and sends its prepared states to a server via a quantum communication channel.
The server then executes the client's instructions and possibly returns some quantum output via the same quantum channel.
As we seek not only verification but also blindness, we will use extensively the simple obfuscation technique put forth in the Universal Blind Quantum Computation (UBQC) protocol (see Section~\ref{sec:prelims} for basics about UBQC) and consisting in randomly rotating each individual qubit sent by the client to the server.

The main idea that has been put at work in previous verification protocols is that, in such case, the client can chose to insert some factitious computations alongside the one it really intends to delegate.
Because the client can choose factitious computations whose results are easy to compute classically and therefore to test, and because the server does not know whether the computation is genuine or factitious, these ensure that the server is non-malicious.

\paragraph{Analysing Deviations with Traps (Section \ref{sec:detect}).}
Here, we lay out a series of concepts that formally define theses factitious computations, or traps, as probabilistic error-detecting schemes.
More precisely, we define \emph{trappified canvases} as subcomputations on an MBQC graph with a fixed input state and classical outputs which follow a probability distribution that is efficiently computable classically.
This is paired to a decision function which, depending on the output of this subcomputation, returns whether the trap accepts or rejects.
The term canvas refers to the fact that there is still empty space on the graph alongside the factitious computation for the client's computation to be ``painted into''.
This task is left to an \emph{embedding algorithm}, which takes a computation and a trappified canvas and fills in the missing parts so that the output is a computation containing both the client's computation and a trap.

Because we aim at blindly delegating the execution of trappified canvases to a possibly fully malicious server that can deviate adaptively, a single trappified canvas will not be enough to constraint its behaviour significantly.
Instead we randomise the construction of trappified canvases, and in particular the physical location of the trap.
This gives rise to the concept of \emph{trappified schemes} (Definition~\ref{def:trap-scheme}) which are sets of trappified canvases from which the client can sample efficiently.

Additionally, for these constructs to be useful in blind protocols they need to satisfy two properties.
First, no information should leak to the server when it is using one trappified canvas over another.
This means that executing one trappified canvas or another must be indistinguishable to the server.
If this is the case, we say that they are \emph{blind-compatible}. In particular, it implies that all trappified canvases in a trappified scheme are supported by the same graph.
Second, no information should leak to the server about the computation in spite of being embedded into a larger computation that contains a trap.
This implies that the decision to accept or reject the computation should not be depending on the client's desired computation.
If this is the case, we call the embedding a \emph{proper embedding}.

Finally, we examine the effect of deviations on individual trappified canvases as well as on trappified schemes. More precisely, we categorise adversarial deviations with the help of trappified schemes as follows: (i) if the scheme rejects with probability $(1-\epsilon)$, then it $\epsilon$-detects the deviation; (ii) if the scheme accepts with probability $(1-\delta)$, it is $\delta$-insensitive to the deviation; and finally (iii) if the result of all possible computations of interest is correct with probability $(1-\nu)$, then the scheme is $\nu$-correct for this deviation.

\paragraph{Secure Verification from Trap Based Protocols (Section \ref{sec:verif}).}

Here, we prove a series of theorems that give sufficient conditions for constructing secure, efficient and robust verification protocols based on the detection, insensitivity and correctness properties of trappified schemes. These results may then serve as design guidelines for tailoring verification protocols to specific needs while removing the burden of proving anew the security for each concrete protocol obtained in this way.

We start by constructing a natural Prepare-and-send protocol from any trappified scheme (see informal Protocol~\ref{proto:informal}). 
\begin{protocol}[ht]
  \caption{Trappified Delegated Blind Computation Protocol (Informal)}
  \label{proto:informal}
    \begin{enumerate}
    	\item The Client samples a trappified canvas from the trappified scheme and embeds its computation, yielding a trappified pattern.
	    \item The Client blindly delegates this trappified pattern to the Server using the UBQC Protocol, after which the Client obtains the output of the trappified pattern.
	    \item The Client decides whether to abort or not based on the result of the decision function of the trappified canvas.
	    \item If it didn't abort, the Client performs some simple classical or quantum post-processing on the output.
    \end{enumerate}
\end{protocol}

We then address the following question: what are the conditions required for these error-detection mechanisms to provide security against arbitrarily malicious servers? Thanks to the blindness of the UBQC Protocol, we show that any strategy of the adversary can be expressed as: following the protocol correctly but applying a convex combination of Pauli operators right before any measurement or before sending back a qubit to the client. We therefore only need to analyse the effect of these strategies -- henceforth called \emph{Pauli deviations} -- on the trappified canvases.
The following theorem states that, if we want the protocol to be secure, the trappified scheme should detect with high probability at least all errors for which the computation is not correct. In other words, it is acceptable to not detect a deviation so long as it has only little effect on the result of the computation of interest.
\begin{theorem}[Detection Implies Verifiability, Informal]
Let $\mathcal E_\epsilon, \mathcal E_\nu \subset \mathcal{P}_V$, where $\mathcal{P}_V$ is the set of Pauli operators on the qubits indexed by the graph vertices (deviations), such that:
\begin{itemize}
\item $\mathcal{P}_V \setminus \mathcal E_\epsilon \subseteq \mathcal{E}_\nu$;
\item $\Id \in \mathcal E_\nu$.
\end{itemize}
If the Trappified Delegated Blind Computation Protocol uses a trappified scheme which:
\begin{itemize}
\item $\epsilon$-detects $\mathcal E_\epsilon$;
\item is $\delta$-insensitive to at least $\{\Id\}$;
\item is $\nu$-correct on $\mathcal E_\nu$;
\end{itemize}
then the protocol is $\max(\epsilon, \delta+\nu)$-secure against an arbitrarily malicious unbounded server.
\end{theorem}
This seemingly intuitive result is proved in the composable framework of Abstract Cryptography~\cite{MR11:abstract-cryptography}. A protocol is secure in this framework if it is a good approximation of an ideal protocol in which a trusted third party, called a resource, receives inputs privately, performs the client's desired computation and returns the outputs privately. The server can only choose whether to make the client abort or not -- independently of the client's computation or input -- and receives a controlled amount of information as a leak. In particular, the probability that the server makes the client accept an incorrect outcome should be low. We show the result above with a novel technique that allows us to derive the protocol's composable security directly -- i.e.~without resorting to local criteria as in~\cite{DFPR14:composable}. To do so, we construct a simulator that is able to correctly guess whether to accept or reject its interaction with the server without ever knowing what the client's computation is, thereby reproducing the behaviour of the concrete protocol even though it can only access the secure-by-design ideal delegated quantum computation resource. As a corollary of this, this theorem provides the first direct proof of composable security of the original Verifiable Blind Quantum Computation (VBQC) protocol~\cite{FK17:unconditionally}.

We next examine the conditions under which the protocol is robust against honest noise. We show that it is sufficient for the trappified scheme to be both insensitive to and correct on likely errors generated by the noise model.

\begin{theorem}[Robust Detection Implies Robust Verifiability, Informal]
  We assume now that the server in the Trappified Delegated Blind Computation Protocol is honest-but-noisy: the error applied is in $\mathcal E_\delta$ with probability $(1-p_\delta)$ and $\mathcal{P}_V \setminus \mathcal E_\delta$ with probability $p_\delta$. 
  Then, the client accepts with probability at least $(1-p_\delta)(1-\delta)$. If furthermore $\mathcal{E}_\delta \subseteq \mathcal{E}_\nu$, then the total correctness error is $p_\delta + \delta + \nu$.
\end{theorem}
Combined with the security above, if the client accepts then the result of its computation is also correct with high probability.

We conclude this abstraction of verification protocols by exploring the necessary conditions for obtaining a security error which is exponentially close to zero without blowing up the server's memory requirements.
We show that efficient trappified schemes must incorporate some error-correction mechanism.
\begin{theorem}[Error-Correction Prevents Resource Blow-up, Informal]
\label{thm:encoding-inf}
If the size of the output in the trappified pattern is the same as in an unprotected execution of the Client's computation for a non-negligible fraction of trappified canvases in the trappified scheme used in the protocol, then the security error of the Trappified Delegated Blind Computation Protocol will scale at least as an inverse polynomial of the size of the graph common to all trappified canvases in the trappified scheme. 
\end{theorem}
This in particular implies that it is impossible to attain negligible security errors without adding redundancy in the computation one way or another if we want to keep the protocol efficient to execute for a polynomial-time server.

These results reveal the strong interplay between the deviation detection properties of trappified schemes and the properties of the corresponding prepare-and-send verification protocol.  As a consequence, the optimisation of verification protocols translates into tailoring the deviation detection properties of trappified schemes to specific needs, for which the rich tools of error-correction can be used.  This is the focus of the rest of the paper.

\paragraph{Correctness and Security Amplification for Classical Input-Output Computations (Section \ref{sec:rvbqc}).}

Here, we construct a general compiler for obtaining trappified schemes.
It interleaves separate computations and test rounds in a way inspired by~\cite{LMKO21:verifying}.
As a consequence, the overhead for protocols based on such schemes is simply a repetition of operations of the same size as the client's original computation, meaning that verification comes for free so long as the client and server can run the blind protocol.
Using our correspondence between error-detection and verification, we then show that this compiler's parameters can be chosen to boost the constant detection and insensitivity rates of the individual test rounds to exponential levels after compilation.
\begin{theorem}[From Constant to Exponential Detection and Insensitivity Rates, Informal]
	Let $\sch P$ be a trappified scheme and $\sch P'$ be the compiled version described above for $n$ rounds combining a number of tests and computations which are both linear in $n$.
	If $\sch P$ $\epsilon$-detects error set $\mathcal{E}_\epsilon$ and is $\delta$-insensitive to $\mathcal{E}_\delta$, then there exists $k_\epsilon, k_\delta$ linear in $n$ and $\epsilon', \delta'$ exponentially-low in $n$ such that $\sch P'$ $\epsilon'$-detects errors with more than $k_\epsilon$ errors on all rounds from set $\mathcal{E}_\epsilon$ and is $\delta'$-insensitive to errors with less than $k_\delta$ errors from set $\mathcal{E}_\delta$.
\end{theorem}

This alone is not enough to obtain a negligible security error. To achieve this, as per Theorem \ref{thm:encoding-inf}, we must recombine the results of the computation rounds to correct for these low-weight errors which are not detected.
This is done by using a simple majority vote on the computation round outcomes, so that correctness can be independently amplified to an exponential level by using polynomially many computation rounds.
\begin{theorem}[Exponential Correctness from Majority Vote, Informal]
There exists $k_\nu$ linear in $n$ and $\nu'$ exponentially-low in $n$ such that $\sch P'$ is $\nu'$-correct so long as there are no more than $k_\nu$ errors.
\end{theorem}

Furthermore, it is possible to choose the values $k_\epsilon, k_\delta, k_\nu$ such that $k_\epsilon \leq k_\nu$ and $k_\delta \leq k_\nu$, such that the error is either too large and in that case detected (if beyond $k_\epsilon$), or it is corrected by the classical repetition code (if below $k_\nu$) and in that case accepted also with high probability (if below $k_\delta$).\footnote{The ideal case is taking $k_\epsilon = k_\delta = k_\nu$.}

In doing so, we have effectively untangled what drives correctness, security and robustness, thereby greatly simplifying the task of designing and optimising new protocols. More precisely, we can now focus only on the design of the test rounds as their  performance  drives the value of exponents in the exponentials from the two previous theorems.

\paragraph{New Optimised Trappified Schemes from Stabiliser Testing (Section \ref{sec:new_traps}).}
In this section, we design test rounds and characterise their error-detection and insensitivity properties. This allows to recover the standard traps used in several other protocols, while also uncovering new traps that correspond to syndrome measurements of stabiliser generators -- hence once again fruitfully exploiting the correspondence between error-detection and verification.

Finally, we combine all of the above into an optimisation of the deviation detection capability of the obtained trappified schemes that not only beats the current state-of-the-art, but more importantly provides an end-to-end application of our theoretical results.

\begin{figure}[ht]
\centering
\includegraphics[width=0.8\textwidth]{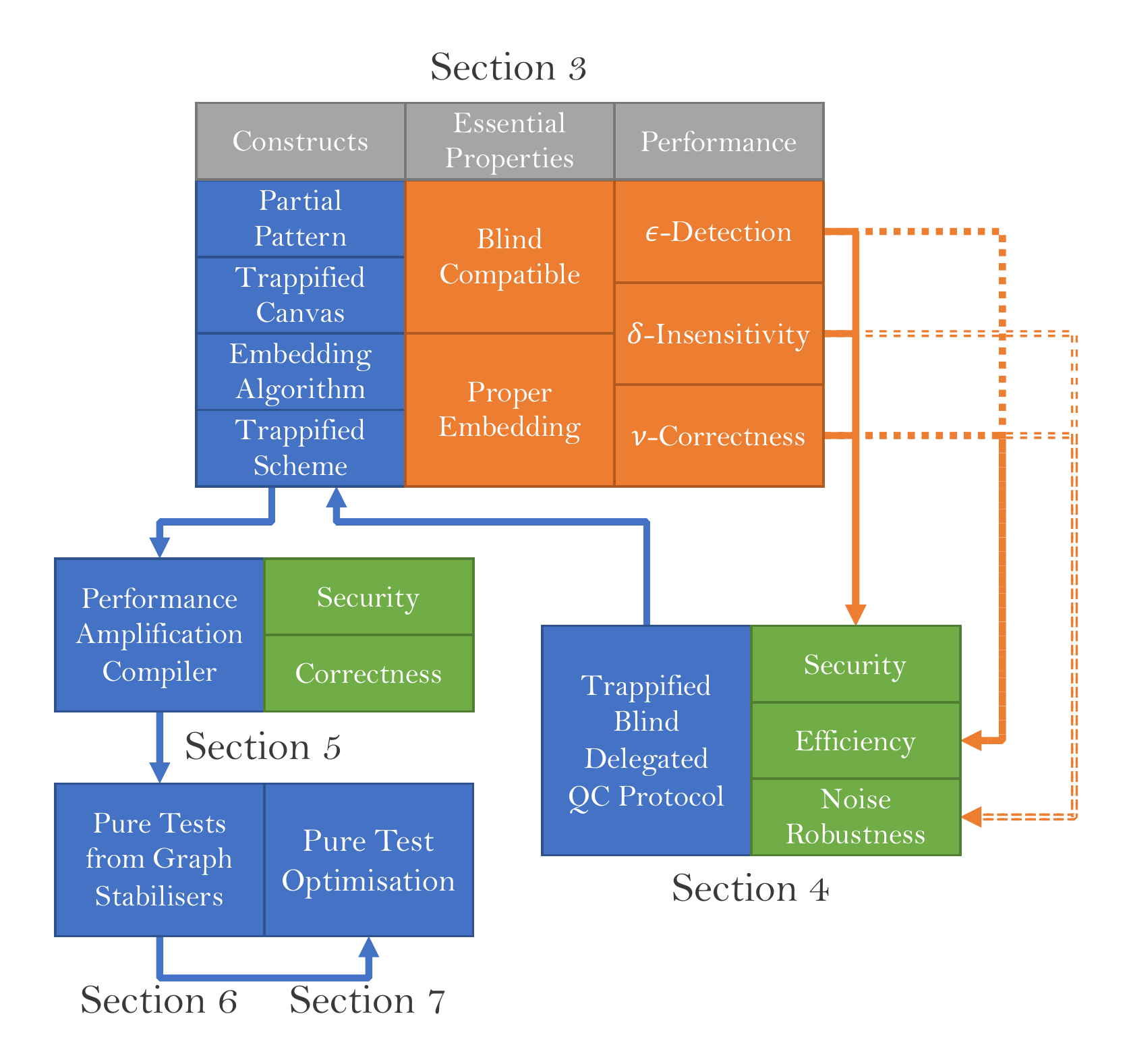}
\caption{Structure of the paper. The blue boxes represent the main objects which we construct, the orange ones are the main properties and the green the main theorems. The blue arrows go towards a higher level of granularity, meaning that an object can be simplified by using the next construction. The orange arrows indicate which property plays a role in the proof of each theorem.}
\end{figure}

\subsection{Future Work and Open Questions}
First, the uncovered connection between error-detection and verification raises further questions such as the extent to which it is possible to infer from the failed traps what the server has been performing.

Second, Theorem~\ref{thm:encoding-inf} implies that some form of error-correction is necessary to obtain exponential correctness. Yet, our protocol shows that sometimes classical error-correction is enough, thereby raising the question of understanding what are the optimal error-correction schemes for given classes of computation that are to be verified.

Third, the design of multi-party quantum computations can be greatly simplified by using the presented approach, making it a versatile tool that extends outside of the single client paradigm~\cite{KKLM+23:quantum}.
   
\section{Preliminaries}
\label{sec:prelims}
\subsection{Notation}
Throughout this work we will use the following notations:
\begin{itemize}
  \item For a set $V$, $\wp(V)$ is the powerset of $V$, the set of all subsets of $V$.
  \item For a set $B \subseteq A$, we denote by $B^c$ the complement of $B$ in $A$, where $A$ will often be the vertex set of a graph and $B$ a subset of vertices, usually input or output locations.
  \item For $n\in \mathbb N$, the set of all integers from $0$ to $n$ included is denoted $[n]$.
  \item For a real function $\epsilon(\eta)$, we say that $\epsilon(\eta)$  is \emph{negligible in $\eta$} if, for all polynomials $p(\eta)$ and $\eta$ sufficiently large, we have $\epsilon(\eta) \leq \frac{1}{p(\eta)}$.
  \item For a real function $\mu(\eta)$, we say that $\mu(\eta)$ is \emph{overwhelming in $\eta$} if there exists a negligible $\epsilon(\eta)$ such that $\mu(\eta) = 1 - \epsilon(\eta)$.  
  \item We denote by $\Theta$ the set of angles $\qty{\frac{k\pi}{4}}_{k \in \qty{0, \ldots, 7}}$.
  \item We note $\X, \Y, \Z$ the Pauli operators. Then $\mathcal{P}_1 = \langle \X, \Y, \Z \rangle$ is the single-qubit Pauli group and $\mathcal{P}_n = \{\cptp P_1 \otimes \ldots \otimes \cptp P_n \mid \cptp P_j \in \mathcal{P}_1\}$ is the $n$-qubit Pauli group.
  \item The rotation operator around the $\Z$-axis of the Bloch sphere by an angle $\theta$ is noted $\Z({\theta}) = \begin{pmatrix} 1 & 0 \\ 0 & e^{i\theta} \end{pmatrix}$.
  \item The states in the $\X - \Y$ plane of the Bloch sphere are noted $\ket{+_{\theta}} = \Z(\theta)\ket{+} = \frac{1}{\sqrt{2}}(\ket{0} + e^{i\theta}\ket{1})$.
  \item Given a set of qubits indexed by elements in set $V$, for all $i \in V$ and any single-qubit unitary $\cptp U$, we denote $\cptp U_i$ the unitary applying $\cptp U$ to qubit $i$ and identity to the rest of the qubits in $V$. If multiple sets of qubits are indexed by $V$, this is easily extended to multi-qubit gates.
  \item For a measurement in the basis $\ket{\pm_\theta}$, we associated the value $0$ to outcome $\ket{+_{\theta}}$ and $1$ to outcome $\ket{-_{\theta}}$.
  \item For an $n$-qubit operator $\cptp U$ and $n$-qubit mixed state $\rho$, we write $\cptp U[\rho]$ for $\cptp U\rho\cptp U^\dagger$.
  \item For two operators $\cptp U$ and $\cptp V$ acting on the same number of qubits, we write $\cptp U \circ \cptp V$ for the composition of the two operators if there may be an ambiguity.
\end{itemize}

\subsection{Measurement-Based Quantum Computation}
\label{sec:mbqc}
The MBQC model of computation emerged from the gate teleportation principle.
It was introduced in~\cite{RB01:one} where it was shown that universal quantum computing can be implemented using graphs-states as resources and adaptive single-qubit measurements.
Therefore MBQC and gate-based quantum computations have the same power.
The measurement calculus expresses the correspondence between the two models~\cite{DKP07:measurement-calculus}. 

MBQC works by choosing an appropriate graph state, performing single-qubit measurements on a subset of this state and, depending on the outcomes, apply correction operators to the rest.
Quantum computations can be easily delegated in this model by having the client supply the quantum input to the server and instruct it by providing measurement instructions, while the server is tasked with the creation of a large entangled state which is suitable for the client's desired computation.

While the discussions below hold for angles in $[0,2\pi)$, if we settle for approximate universality it is sufficient to restrict ourselves to the set of angles $\Theta$ \cite{FK17:unconditionally}. 
The client's computation is defined by a \emph{measurement pattern} as follows.
\begin{definition}[Measurement Pattern]
  \label{def:pattern}
  A \emph{pattern} in the Measurement-Based Quantum Computation model is given by a graph $G = (V,E)$, input and output vertex sets $I$ and $O$, a set of measurement angles for non-output qubits $\{\phi(i)\}_{i\in O^c}$ in the $\X - \Y$ plane of the Bloch sphere, and a flow function $f$ which induces a partial ordering of the qubits $V$.
\end{definition}

To make this more concrete, we will now describe an example of an MBQC pattern and corrections on the three-vertex linear graph. In that case we have $V = \{1, 2, 3\}$ and $E = \{(1, 2), (2, 3)\}$. The first qubit in the line will be the only one in the input set $I$ and the last qubit the only one in the output set $O$. We note $\phi(1)$ and $\phi(2)$ the measurement angles of the first two (non-output) qubits. We start with a single qubit in state $\ket{\psi}$ as input, the qubits associated to the other two vertices are initialised in the $\ket{+}$ state. We apply one $\CZ$ gate for each pair of qubits whose associated vertices are linked by an edge in $E$. If $\ket{\psi} = \alpha \ket{0} + \beta\ket{1}$, the resulting state is
\begin{equation}
\CZ_{1, 2}\CZ_{2, 3}\ket{\psi}\ket{+}\ket{+} = \frac{\alpha}{2}(\ket{000} + \ket{001} + \ket{010} - \ket{011}) + \frac{\beta}{2}(\ket{100} + \ket{101} - \ket{110} + \ket{111}).
\end{equation}
In order to perform the measurement on the qubit in vertex $1$, we apply the rotation $\Z(-\phi(1))$ and project either onto state $\ket{+}$ -- measurement outcome $0$ -- or $\ket{-}$ -- outcome $1$. The states of the unmeasured qubits (vertices $2$ and $3$) are then
\begin{align}
\ket{\psi_0} &= \frac{1}{\sqrt{2}}(\alpha + \beta e^{-i\theta})\ket{0}\ket{+} + \frac{1}{\sqrt{2}}(\alpha - \beta e^{-i\theta})\ket{1}\ket{-},\\
\ket{\psi_1} &= \frac{1}{\sqrt{2}}(\alpha - \beta e^{-i\theta})\ket{0}\ket{+} + \frac{1}{\sqrt{2}}(\alpha + \beta e^{-i\theta})\ket{1}\ket{-}.
\end{align}
There are two things that we can notice from this. In the first case, the state is the same as if we had started with the state $\ket{\psi'} = \Ha\Z(-\phi(1))\ket{\psi}$ and entangled it to a single $\ket{+}$ state using a single $\CZ$ operation -- a two-qubit linear graph. This fact allows us to perform the translation between the MBQC model and the circuit model. We see also that, in order to recover $\ket{\psi_0}$ from the state $\ket{\psi_1}$, we need to apply an $\X$ operation on the qubit associate to vertex $2$ and a $\Z$ operation on the qubit associated to vertex $3$. After applying these operations, the state will be independent of the outcome of the measurement. This induces an ordering on the vertices since $1$ must be measured before $2$ and $3$ if we want to use this correction strategy.

Instead of directly applying these operations, we can absorb them into the angle of future measurements since $\bra{\pm}\Z(\phi)\X = \bra{\pm}\Z(-\phi)$ and $\bra{\pm}\Z(\phi)\Z = \bra{\pm}\Z(\phi + \pi)$. This can be done for the qubit in vertex $2$ but the output qubits still need to be corrected since they are not measured.

The flow function $f$ guarantees that the computation is independent of the intermediary measurement outcomes by specifying how these corrections must be performed depending on previous outcomes. It is an injective function from non-output vertices $O^c$ -- the outputs are not measured and therefore do not generate corrections -- to non-input vertices $I^c$ -- these are measured first and therefore do not need to be corrected. To each vertex $i$ are associated the sets $S_X(i) = f^{-1}(i)$ and $S_Z(i) = \{j \mid i \in N_G(f(j))\}$ which are respectively called the $X$ and $Z$ dependency sets for vertex $i$. The measurement angle of vertex $i$ depends on the measurement outcomes of qubits in $S_X(i)$ and $S_Z(i)$. A measurement outcome of $1$ in a qubit from $S_X(i)$ will multiply the angle of $i$ by $-1$, while the $Z$ dependencies add $\pi$ to the angle.In the example above, we can specify the flow function by $f(1) = 2$ and $f(2) = 3$, which induces the measurement order $1 \preceq 2 \preceq 3$. Then $S_X(2) = 1$, $S_X(3) = 2$, $S_Z(3) = 2$ and the rest are empty. More generally, finding the flow relies on the stabilisers of the graph state associated to $G$.
Further details regarding the definition of the flow and its generalisation g-flow can be found in \ref{app:flow} and references~\cite{HEB04:multiparty,DK06:determinism}.

The execution of MBQC patterns can then be delegated to servers, alleviating the need for the client to own a quantum machine using Protocol~\ref{proto:mbqc}. This first protocol is not blind, nor can the client test that the server is performing the computation correctly.

\begin{protocol}[ht]
  \caption{Delegated MBQC Protocol}\label{proto:mbqc}
  \begin{algorithmic}[0]
    \STATE \textbf{Client's Inputs:} A measurement pattern $(G, I, O, \{\phi(i)\}_{i\in O^c}, f)$ and a quantum register containing the input qubits $i\in I$.
    \STATE \textbf{Protocol:}
    \begin{enumerate}
	\item The Client sends the graph's description $(G, I, O)$ to the Server.
	\item The Client sends its input qubits for positions $I$ to the Server.
    \item The Server prepares $\ket +$ states for qubits $i\in I^c$.
    \item The Server applies a $\CZ$ gate between qubits $i$ and $j$ if $(i,j)$ is an edge of $G$.
    \item The Client sends the measurement angles $\{\phi(i)\}_{i\in O^c}$ along with the description of $f$ to the Server.
    \item The Server measures the qubits $i \in O^c$ in the order $\preceq$ induced by $f$ in the basis $\ket{\pm_{\phi'(i)}}$ where
      \begin{align}
        s_X(i) & = \bigoplus_{j \in S_X(i)} b(j), \ s_Z(i) = \bigoplus_{j \in S_Z(i)} b(j),\\
        \phi'(i) & = (-1)^{s_X(i)}\phi(i) + s_Z(i) \pi, \label{eq:updt-mbqc}
      \end{align}
      where $b(j)\in \{0,1\}$ is the measurement outcome for qubit $j$.
    \item The Server applies the correction $\Z^{s_Z(i)}_i \X^{s_X(i)}_i$ for each output qubits $i \in O$, which it sends back to the Client.
    \end{enumerate}
  \end{algorithmic}
\end{protocol}

If the client is able to perform single-qubit preparations and use quantum communication, it can delegate an MBQC pattern blindly~\cite{BFK09:universal}, meaning that the Server does not learn anything about the computation besides the prepared graph $G$, the set of outputs $O$ and the order of measurements.
The goal of the Universal Blind Quantum Computation Protocol is to hide the computation, the inputs and the outputs up to a controlled leakage which consists of the graph and order of measurements. The client needs to be able to generate state in the $\X - \Y$ plane $\ket{+_\theta}$ for values $\theta \in \Theta$, and in the case of quantum inputs it must also be able to apply $\Z(\theta)$ to its inputs and $\X$.

We will use the example MBQC computation described above to demonstrate an execution of this protocol. The client would like to hide the input state $\ket{\psi}$, the measurement angles $\phi(1)$ and $\phi(2)$ and the output. For the input, the client uses a variant of the Quantum One-Time-Pad, sampling a random bit $a(1) \in \bin$ and a random angle $\theta(1)$ and applying the operation $\Z_i(\theta(i))\X_i^{a(i)}$ to $\ket{\psi}$. For the non-input qubits it sets $a(2) = a(3) = 0$. For the non-output qubit it samples at random $\theta(2) \in \Theta$ and for the ouptut vertex it samples at random $\theta(3) \in \{0, \pi\}$. It creates the states $\ket{+_{\theta(2)}}$ and $\ket{+_{\theta(3)}}$ and sends these two states and its encrypted input to the server.

The server receives these three qubits and performs the entangling operations $\CZ_{1, 2}$ and $\CZ_{2, 3}$ as above. For now the security is guaranteed since the input is perfectly encrypted and the other qubits are in random states uncorrelated to the computation. However, the client still desires to run its computation and must do so through the encryption. To do so, it will instruct the server to measure the qubits $1$ and $2$ with the angle $\delta(i) = (-1)^{a(i)}\phi'(i) + \theta(i) + (r(i) + a_N(i)) \pi$, where $a_N(i)$ is the sum of values of $a(j)$ for $j$ neighbours of $i$ and $r(i)$ is a random bit.

We can see that this performs the same computation as the MBQC example described above. The $\Z$ rotation encryption commutes with the $\CZ$ operations and cancels out the encryption of the measurement angle. On the other hand, when the $\X$ encryption of the input commutes with the $\CZ$ gates, it creates an additional $\Z$ on the neighbour which is then taken care of by $a_N(2) = a(1)$. Commuting the $\X$ to the end also flips the sign of the measurement angle, which is why $(-1)^{a(i)}$ appears in from of $\phi'(i)$ in the expression of $\delta(i)$. In the end the computation is the same as the unencrypted one above so long as the client corrects the measurement outcomes returned by the server to account for the additional $r(i)$ by flipping the outcome $b(i)$ if $r(i) = 1$. If $\theta(3) = \pi$, the client must also apply $\Z$ to the output to compensate. This process is summarised in Figure~\ref{fig:ubqc-cor}.

\begin{figure}[h]\centering
\subfloat[UBQC Protocol with the explicit values of $\delta(i)$.]{
\label{fig:ubqc-cor1}
\makebox[\textwidth]{
$
\Qcircuit @C=1.0em @R=.7em {
\lstick{\ket{\psi}} & \gate{\X^{a(1)}} & \gate{\Z(\theta(1))} & \ctrl{1}  & \qw       & \gate{\Z( - (-1)^{a(1)}\phi'(1) - \theta(1) + r(1) \pi)} & \qw & \measureD{\pm} \\
\lstick{\ket{+}}    & \qw              & \gate{\Z(\theta(2))} & \ctrl{-1} & \ctrl{1}  & \gate{\Z( - \phi'(2) - \theta(2) + (r(2) + a(1)) \pi)}   & \qw & \measureD{\pm} \\
\lstick{\ket{+}}    & \qw              & \gate{\Z(\theta(3))} & \qw       & \ctrl{-1} & \qw                                                      & \qw & \qw
}
$
}
}\\ \vspace{5mm}
\subfloat[The $\Z$ encryption commutes through the $\CZ$ gates and is cancelled out by the later $\Z$ rotation.]{
\label{fig:ubqc-cor2}
\makebox[\textwidth]{
$
\Qcircuit @C=1.0em @R=.7em {
\lstick{\ket{\psi}} & \gate{\X^{a(1)}} & \ctrl{1}  & \qw       & \gate{\Z( - (-1)^{a(1)}\phi'(1) + r(1) \pi)} & \qw                  & \measureD{\pm} \\
\lstick{\ket{+}}    & \qw              & \ctrl{-1} & \ctrl{1}  & \gate{\Z( - \phi'(2) + (r(2) + a(1)) \pi)}   & \qw                  & \measureD{\pm} \\
\lstick{\ket{+}}    & \qw              & \qw       & \ctrl{-1} & \qw                                          & \gate{\Z(\theta(3))} & \qw
}
$
}
}\\ \vspace{5mm}
\subfloat[The input $\X$ encryption commutes through the $\CZ$ gates but adds a $\Z$ on qubit $2$, which is cancelled out by the $a(1)\pi$ inside the rotation. The $\X$ on qubit $1$ is commuted through the rotation and absorbed by the measurement. The result is the client's desired MBQC computation up to Pauli corrections and bit-flips.]{
\label{fig:ubqc-cor3}
\makebox[\textwidth]{
$
\Qcircuit @C=1.0em @R=.7em {
\lstick{\ket{\psi}} & \ctrl{1}  & \qw       & \gate{\Z( - \phi'(1) + r(1) \pi)} & \qw                  & \measureD{\pm} \\
\lstick{\ket{+}}    & \ctrl{-1} & \ctrl{1}  & \gate{\Z( - \phi'(2) + r(2) \pi)} & \qw                  & \measureD{\pm} \\
\lstick{\ket{+}}    & \qw       & \ctrl{-1} & \qw                               & \gate{\Z(\theta(3))} & \qw
}
$
}
}\caption{Correctness of UBQC for three-vertex linear graph.}
\label{fig:ubqc-cor}
\end{figure}
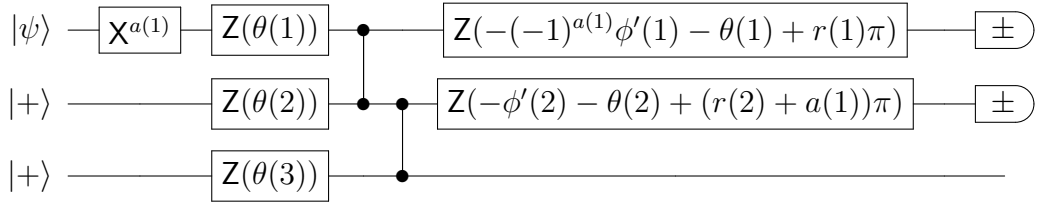
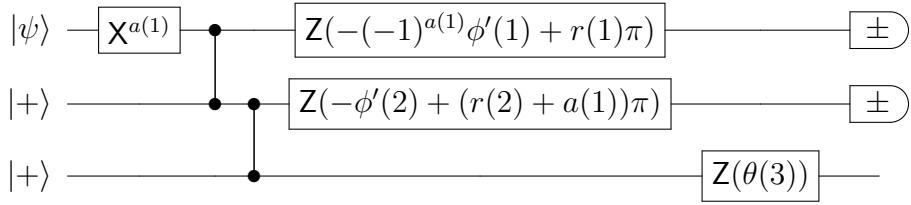
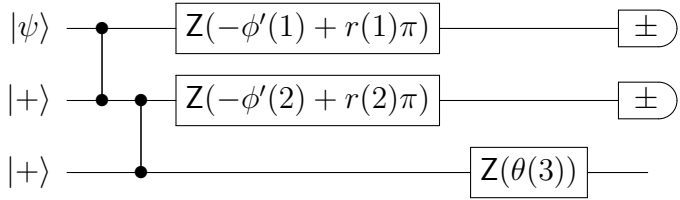

As for the security, intuitively, the fact that $\theta(i)$ is perfectly random hides the value of $\phi'(i)$ in $\delta(i)$, while the value $r(i)$ hides the measurement outcome but also the output of the computation since it is propagated by the flow and results in a Quantum One-Time-Pad of the output. The process described above is formalised in Protocol~\ref{proto:ubqc}.

\begin{protocol}[ht]
  \caption{UBQC Protocol}\label{proto:ubqc}
  \begin{algorithmic}[0]
    \STATE \textbf{Client's Inputs:} A measurement pattern $(G, I, O, \{\phi(i)\}_{i\in O^c}, f)$ and a quantum register containing the input state $\rho_C$ on qubits $i\in I$.
    \STATE \textbf{Protocol:}
    \begin{enumerate}
    \item The Client sends the graph's description $(G, I, O)$ and the measurement order to the Server.
    \item The Client prepares and sends all the qubits in $V$ to the Server:\footnotemark
      \begin{enumerate}
      \item For $i \in I$, it chooses a random bit $a(i)$. For $i \in I^c$, it sets $a(i) = 0$.
      \item For $i \in O$, it chooses a random bit $r(i)$ and sets $\theta(i) = (r(v) + a_N(v))\pi$ where $a_N(i) = \sum_{j \in N_G(i)} a(j)$. For $i \in O^c$, it samples a random $\theta(i) \in \Theta$.
      \item For $i \in I$, it sends $\prod_{i \in I}\Z_i(\theta(i))\X_i^{a(i)}[\rho_C]$. For $i \in I^c$ it sends $\ket{+_{\theta(i)}}$.
      \end{enumerate}
    \item The Server applies a $\CZ$ gate between qubits $i$ and $j$ if $(i,j)$ is an edge of $G$.
    \item For all $i \in O^c$, in the order specified by the flow $f$, the Client computes the measurement angle $\delta(i)$ and sends it to the Server, receiving in return the corresponding measurement outcome $b(i)$:
      \begin{align}
        s_X(i) & = \bigoplus_{j \in S_X(i)} b(i) \oplus r(i), \
        s_Z(i) = \bigoplus_{j \in S_Z(i)} b(i) \oplus r(i), \\
        \delta(i) & = (-1)^{a(i)}\phi'(i) + \theta(i) + (r(i) + a_N(i)) \pi,\label{eq:updt-ubqc}
      \end{align}
      where $\phi'(i)$ is computed using Equation \ref{eq:updt-mbqc} with the new values of $s_X(i)$ and $s_Z(i)$.
    \item The Server sends back the output qubits $i \in O$.
    \item The Client applies $\Z_i^{s_Z(i) + r(i)}\X_i^{s_X(i) + a(i)}$ to the received qubits $i \in O$.
    \end{enumerate}
  \end{algorithmic}
\end{protocol}

Note that if the output of the client's computation is classical, the set $O$ is empty and the client only receives measurement outcomes. The output measurement outcomes $b(i)$ sent by the Server need to be decrypted by the Client according to the equation $s(j) = b(j) \oplus r(j)$, thus preserving the confidentiality of the output of the computation.

The additional randomisation of the output qubit might seem superfluous since the server can simply measure the qubit to recover the value of the state and indeed it is not present in the original UBQC protocol.\footnotetext{In the original UBQC Protocol from \cite{BFK09:universal}, the outputs are prepared by the Server in the $\ket{+}$ state and are encrypted by the computation flow.} However, In the verification protocol in which we will use the UBQC Protocol later, some inputs to auxiliary trap computations may be included in the global output. This means that all output qubits must also be prepared by the Client and not adding this randomisation at this stage could allow the server to gain information from these qubits. This does not change the security properties of the original UBQC Protocol, which are presented in the next subsection.

To analyse the security of our SDQC protocol later, we will require the following Pauli Twirling Lemma as a way to decompose the actions of an Adversary in the blind protocol above. A Pauli twirl occurs when a random Pauli operator is applied (such as an encryption and decryption). The result from the point of view of someone who does not know which Pauli has been used is a state or channel that is averaged over all possible Pauli operators. This has the effect of removing all off-diagonal factors from the operation sandwiched between the two applications of the random Pauli, thus making it a convex combination of Pauli operators.

\begin{lemma}[Pauli Twirling]
\label{lem:twirl}
Let $\rho$ be an $n$-qubit mixed state and~$\cptp Q, \cptp Q' \in \mathcal{P}_n$ two $n$ qubit Pauli operators. Then, if $\cptp Q \neq \cptp Q'$, we have:

\begin{align}
\smashoperator[r]{\sum_{\cptp P \in \mathcal{P}_n}} \cptp P^{\dagger} \cptp Q \cptp P \rho \cptp P^{\dagger} \cptp Q'^{\dagger} \cptp P = 0.
\end{align}
\end{lemma}

\subsection{Abstract Cryptography}
\label{sec:ac}
Abstract Cryptography (AC) is a security framework for cryptographic protocols that was introduced in \cite{MR11:abstract-cryptography,M12:constructive-cryptography}.
The focus of the AC framework is to provide general composability. In this way, protocols that are separately shown to be secure within the framework can be composed in sequence or in parallel while keeping a similar degree of security. See \cite{DFPR14:composable} for further details.

On an abstract level, the AC framework considers resources and protocols. While a resource provides a specified functionality, protocols are essentially instructions how to construct resources from other resources. In this way, this framework allows the expansion of the set of available resources while ensuring general compatibility.

Technically, a quantum protocol $\pi$ with $N$ honest parties is described by $\pi = (\pi_1, \ldots, \pi_N)$, where the combined actions of party $i$, denoted $\pi_i$, are called the converter of party $i$ and consist in the quantum case of a sequence of efficiently implementable CPTP maps.
A resource has interfaces with the parties that are allowed to exchange states with it. During its execution, it waits for all input interfaces to be initialised, then applies a Completely Positive Trace-Preserving (CPTP) map to all interfaces and its internal state, and finally transmits the states in the output interfaces back to the appropriate parties. This process may be repeated multiple times.
Entirely classical resources can be enforced by immediate measurements of all input registers and the restriction of the output to computational basis states.

AC security is entirely based on the indistinguishability of resources. A protocol is considered to be secure if the resource which it constructs is indistinguishable from an ideal resource which encapsulates the desired security properties.
Two resources with the same number of interfaces are called indistinguishable if, given access to one of the resources, the guess of any algorithm trying to decide which one it is is close to random.
During this process, the algorithm, called the distinguisher, has access to all of the resource's interfaces.

\begin{definition}[Indistinguishability of Resources]
\label{indist-res}
Let $\epsilon>0$  and $\mathcal{R}_1$ and $\mathcal{R}_2$ be two resources with same input and output interfaces.  
Then, these resources are called $\epsilon$-statistically-indistinguishable, denoted $\mathcal{R}_1 \!\!\!\underset{\mathit{stat}, \epsilon}{\approx}\!\!\! \mathcal{R}_2$, if for all (unbounded) distinguishers $\mathcal{D}$ it holds that
\begin{align}
\Bigl\lvert\Pr[b = 1 \mid b \leftarrow \mathcal{D}\mathcal{R}_1] - \Pr[b = 1 \mid b \leftarrow \mathcal{D}\mathcal{R}_2]\Bigr\rvert \leq \epsilon.
\end{align}
Analogously, $\mathcal{R}_1$ and $\mathcal{R}_2$ are said to be computationally indistinguishable if this holds for all quantum polynomial-time distinguishers.
\end{definition}

With this definition in mind, the correctness of a protocol is captured by the indistinguishability of the resource constructed by the protocol from the ideal resource when all parties are honest, i.e.~they use their respective converters as specified by the protocol.
The security of the protocol against a set of malicious and collaborating parties is given by the indistinguishability of the constructed resource where the power of the distinguisher is extended to the transcripts of the corrupted parties.
This is formally captured by Definition~\ref{const-sec-def}.

\begin{definition}[Construction of Resources]
\label{const-sec-def}
Let $\epsilon > 0$.  
We say that an $N$-party protocol $\pi$ $\epsilon$-statistically-constructs resource $\mathcal{S}$ from resource $\mathcal{R}$ against adversarial patterns $\mathsf{P} \subseteq \wp([N])$ if:
\begin{enumerate}
\item It is correct: $\pi \mathcal{R} \!\!\!\underset{\mathit{stat}, \epsilon}{\approx}\!\!\! \mathcal{S}$.
\item It is secure for all subsets of corrupted parties in the pattern $M \in \mathsf{P}$: there exists a simulator (converter) $\sigma_M$ such that $\pi_{M^c}\mathcal{R} \!\!\!\underset{\mathit{stat}, \epsilon}{\approx}\!\!\! \mathcal{S} \sigma_M$.
\end{enumerate}
Analogously, computational correctness and security is given for computationally bounded distinguishers as in Definition~\ref{indist-res}, and with a quantum polynomial-time simulator $\sigma_M$.
\end{definition}

This finally allows us to formulate the General Composition Theorem at the core of the Abstract Cryptography framework.

\begin{theorem}[General Composability of Resources {\cite[Theorem~1]{MR11:abstract-cryptography}}]
\label{thm:ac-compos}

Let $\mathcal{R}$, $\mathcal{S}$ and $\mathcal{T}$ be resources, $\alpha$, $\beta$ and $\mathsf{id}$ be protocols (where protocol $\mathsf{id}$ does not modify the resource it is applied to). Let $\cdot$ and $\parallel$ denote respectively the sequential and parallel composition of protocols and resources. Then the following implications hold:

\begin{itemize}
\item The protocols are \emph{sequentially composable}: if $\alpha \mathcal{R} \!\!\!\underset{\mathit{stat}, \epsilon_{\alpha}}{\approx}\!\!\! \mathcal{S}$ and $\beta \mathcal{S} \!\!\!\underset{\mathit{stat}, \epsilon_{\beta}}{\approx}\!\!\! \mathcal{T}$ then $(\beta \cdot \alpha) \mathcal{R} \!\!\!\underset{\mathit{stat}, \epsilon_{\alpha} + \epsilon_{\beta}}{\approx}\!\!\! \mathcal{T}$.
\item The protocols are \emph{context-insensitive}: if $\alpha \mathcal{R} \!\!\!\underset{\mathit{stat}, \epsilon_{\alpha}}{\approx}\!\!\! \mathcal{S}$ then $(\alpha \parallel \mathsf{id}) (\mathcal{R} \parallel \mathcal{T}) \!\!\!\underset{\mathit{stat}, \epsilon_{\alpha}}{\approx}\!\!\! (\mathcal{S} \parallel \mathcal{T})$.
\end{itemize}
\end{theorem}

A combination of these two properties yields concurrent composability, where the distinguishing advantage accumulates additively as well.

Resource \ref{def:bvqc_if} captures the security properties of a blind and verifiable delegated protocol for a given class of computations. It allows a single Client to run a quantum computation on a Server so that the Server cannot corrupt the computation and does not learn anything besides a given leakage $l_{\rho}$. We recall the original definition from \cite[Definition 4.2]{DFPR14:composable}.

\begin{resource}[ht]
\caption{Secure Delegated Quantum Computation}
\label{def:bvqc_if}
\begin{algorithmic}[0]

\STATE \textbf{Public Information:} Nature of the leakage $l_{\rho_C}$.

\STATE \textbf{Inputs:} 
\begin{itemize}
\item The Client inputs the classical description of a computation $\cptp C$ from subspace $\Pi_{I,C}$ to subspace $\Pi_{O,C}$ and a quantum state $\rho_C$ in $\Pi_{I,C}$.
\item The Server chooses whether or not to deviate. This interface is filtered by two control bits $(e, c)$ (set to $0$ by default for honest behaviour).
\end{itemize}

\STATE \textbf{Computation by the Resource:}
\begin{enumerate}
\item If $e = 1$, the Resource sends the leakage $l_{\rho}$ to the Server's interface; if it receives $c = 1$, the Resource outputs $\dyad{\bot}\otimes\dyad{\rej}$ at the Client's output interface.
\item Otherwise it outputs $\cptp{C}[\rho_C] \otimes \dyad{\acc}$ at the Client's output interface.
\end{enumerate}

\end{algorithmic}
\end{resource}

On the other hand, the following resource models the security of the UBQC Protocol \ref{proto:ubqc}. It leaks no information to the Server beyond a controlled leak, but allows the Server to modify the output by deviating from the Client's desired computation.

\begin{resource}[ht]
\caption{Blind Delegated Quantum Computation}
\label{def:bqc_if}
\begin{algorithmic}[0]

\STATE \textbf{Public Information:} Nature of the leakage $l_{\rho_C}$.

\STATE \textbf{Inputs:} 
\begin{itemize}
\item The Client inputs the classical description of a computation $\cptp C$ from subspace $\Pi_{I,C}$ to subspace $\Pi_{O,C}$ and a quantum state $\rho_C$ in $\Pi_{I,C}$.
\item The Server chooses whether or not to deviate. This interface is filtered by two control bits $(e, c)$ (set to $0$ by default for honest behaviour). If $c = 1$, the Server has an additional input CPTP map $\cptp F$ and state $\rho_S$.
\end{itemize}

\STATE \textbf{Computation by the Resource:}
\begin{enumerate}
\item If $e = 1$, the Resource sends the leakage $l_{\rho_C}$ to the Server's interface.
\item If $c = 0$, it outputs $\cptp{C}[\rho_C]$ at the Client's output interface. Otherwise, it waits for the additional input and outputs $\Tr_S(\cptp F[\rho_{CS}])$ at the Client's interface.
\end{enumerate}

\end{algorithmic}
\end{resource}

The following theorem captures the security guarantees of the UBQC Protocol \ref{proto:ubqc} in the Abstract Cryptography Framework, as expressed in \cite{DFPR14:composable}.

\begin{theorem}[Security of Universal Blind Quantum Computation]
\label{thm:sec-ubqc}
The UBQC Protocol \ref{proto:ubqc} perfectly constructs the Blind Delegated Quantum Computation Resource \ref{def:bqc_if} for leak $l_{\rho_C} = (G, O, \preceq_G)$, where $\preceq_G$ is the ordering induced by the flow of the computation.
\end{theorem}

\section{Analysing Deviations with Traps}
\label{sec:detect}
The goal of this section is to introduce the concepts and tools for detecting deviations from a given computation.
Later, in Section~\ref{sec:verif}, we combine these techniques with blindness in order to detect malicious deviations, i.e.~perform verification.

\subsection{Abstract Definitions of Traps}
\label{subsec:detect-def}
We start by defining partial MBQC patterns in Definition~\ref{def:pp}, which fix only a subset of the measurement angles and flow conditions on a given graph. We constrain the flow such that the determinism of the computation is preserved on the partial pattern independently of the how the rest of the flow is specified.

\begin{definition}[Partial MBQC Pattern]
\label{def:pp}

Given a graph $G=(V,E)$, a partial pattern $P$ on $G$ is defined by:

\begin{itemize}
\item $G_P = (V_P, E_P = E \cap V_P \times V_P)$, a subgraph of $G$;
\item $I_P$ and $O_P$, the partial input and output vertices, with subspaces $\Pi_{I, P}$ and $\Pi_{O, P}$ defined on vertices $I_P$ and $O_P$ through bases $\mathcal{B}_{I, P}$ and $\mathcal{B}_{O, P}$ respectively;
\item $\{\phi(i)\}_{i\in V_{P} \setminus O_{P}}$, a set of measurement angles;
\item $f_p : V_{P} \setminus O_{P} \rightarrow V_{P} \setminus I_{P}$, a flow inducing a partial order $\preceq_P$ on $V_P$.
\end{itemize}

\end{definition}

\begin{example}[Partial Pattern for Computing]
\label{ex:partial}
Let $G$ be the $n\times m$ 2D-cluster graph -- i.e.~$n$-qubit high and $m$-qubit wide -- and the ordering of the qubits starting in the upper-left corner, going down first then right. Such graph state is universal for MBQC~\cite{RB01:one}.  There are many possible partial patterns that can be defined on such graph. For instance, consider a pattern $Q$ that runs on a smaller $n'\times m'$ 2D-cluster graph. Then, one can define a partial pattern $P$ on $G$ as the top-left $(n'+1)\times (m'+1)$ subgraph. The set $I_P$ is defined as the set $I$ of $Q$ together with all the qubits on the bottom row and right column. The input space corresponds to the Hilbert space of the input qubits of $Q$ tensored with $\ket 0$ for the qubits of the bottom row and right column. The output set $O_P$ is the same set as in $Q$ and $\Pi_{O,P}$ is the full Hilbert space of the output qubits. The measurement angles are the same as in $Q$ for the corresponding qubits and set to be random for the bottom row and right column. The flow is the same as in $Q$, provided that the added $\ket 0$ qubits have no dependent qubits.  Because the added qubits are forced to be in the $\ket 0$ state, this isolates a $n' \times m'$ 2D-cluster graph that can then be used to perform the same operations as in $Q$, thereby allowing to compute the same unitary, albeit using a larger graph, see Figure~\ref{fig:ex1}.  Note that one can change the location of the $n' \times m'$ 2D-cluster graph used for the computation, as long as it is properly surrounded by qubits in the $\ket 0$ state. This is done by defining the input subspace of the partial pattern to take that constraint into account.
\end{example}

\begin{figure}
  \centering
  \includegraphics[width=8cm]{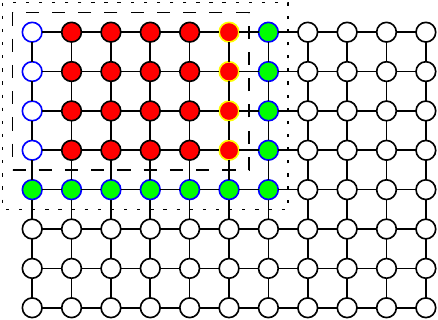}
  \caption{Partial pattern for computing. The partial pattern is in the dashed box. Input qubits are surrounded in blue, output qubits in yellow. Red filled qubits are prepared in $\ket +$ while the green ones are prepared in $\ket 0$. The green qubits define a subspace of the Hilbert space of the input qubits that guarantees that a $4 \times 6$ cluster state computation can be run inside the long-dashed box.}
  \label{fig:ex1}
\end{figure}

We now use this notion to define trappified canvases. These contain a partial pattern whose input state is fixed such that it produces a sample from an easy to compute probability distribution when its ouput qubits are measured in the $\X$ basis. These partial patterns are called \emph{traps} and will be used to detect deviations in the following way. Whenever a trap computation is executed, it should provide outcomes that are compatible with the trap's probability distribution. Failure to do so is a sign that the server deviated from the instructions given by the client.

\begin{definition}[Trappified Canvas]
\label{def:trap-c}

A \emph{trappified canvas} $(T, \sigma, \pd{T}, \tau)$ on a graph $G = (V, E)$ consists of:

\begin{itemize}
\item $T$, a partial pattern on a subset of vertices $V_T$ of $G$ with input and output sets $I_T$ and $O_T$;
\item $\sigma$, a tensor product of single-qubit states on $\Pi_{I, T}$;
\item $\pd{T}$, an efficiently classically computable probability distribution over binary strings;
\item and $\tau$, an efficient classical algorithm that takes as input a sample from $\pd{T}$ and outputs a single bit;
\end{itemize}

\noindent such that the $X$-measurement outcomes of qubits in $O_T$ are drawn from probability distribution $\pd{T}$. Let $t$ be such a sample, the outcome of the trappified canvas is given by $\tau(t)$. By convention we say that it accepts whenever $\tau(t)=0$ and rejects for $\tau(t)=1$.

\end{definition}

We will often abuse the notation and refer to the trappified canvas $(T,\sigma, \pd T, \tau)$ as $T$. 

\begin{example}[Canvas with a Single Standard Trap]
\label{ex:canvas}
Consider the $n\times m$ 2D-cluster graph and consider the partial pattern of Example~\ref{ex:partial} where the subgraph is a $3\times 3$ square -- i.e.~a single computation qubit surrounded by 8 $\ket 0$ states. The input state is fixed to be $\sigma = \ket{+}\otimes\ket{0}^{\otimes 8}$ where $\ket +$ is the state of the central qubit, the others being the aforementioned peripheral ones. Because the central qubit is measured along the $\X$-axis $\pd T$ is deterministic -- the measurement outcome $0$ corresponding to the projector $\ketbra +$ has probability $1$. The accept function is defined by $\tau(t) = t$ so that the trappified canvas accepts whenever the measurement outcome of the central qubit corresponds to the expected $0$ outcome. Here, the $3\times 3$ partial pattern defines a trap( see Figure~\ref{fig:ex2}).
\end{example}

\begin{figure}
  \centering
  \includegraphics[width=8cm]{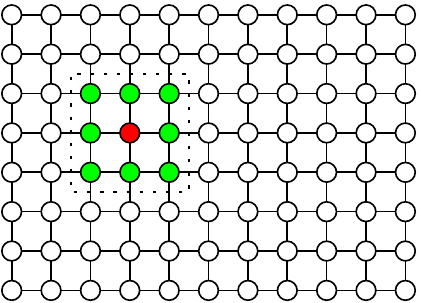}
  \caption{Trappified canvas. The partial pattern inside the dashed box is the trap. The central qubit (red) is prepared in $\ket +$ and is surrounded by $\ket 0$'s (green) that effectively ensure that irrespective of the measurement angles on the remaining qubits the central qubit will remain in $\ket +$. Failure to obtain the 0-outcome when measuring $\X$ will be a proof that the server deviated from the given instructions. The preparation and measurement angles of the remaining qubits is left unspecified.}
  \label{fig:ex2}
\end{figure}

Note that the input and output qubits of a partial pattern may not be included in the input and output qubits of the larger MBQC graph. This gives us more flexibility in defining trappified canvases: during the protocol presented in the next section, the server will measure all qubits in $O^c$ -- which may include some of the trap outputs --, while any measurement of qubits in $O$ will be performed by the client. This allows the trap to catch deviations on the output qubits as well.

In order to be useful, trappified canvases must contain enough empty space -- vertices which have been left unspecified -- to accommodate the client's desired computation. Inserting this computation is done via an \emph{embedding algorithm} as described in the following Definition.

\begin{definition}[Embedding Algorithm]
\label{def:embed-alg}

Let $\mathfrak{C}$ be a class of quantum computations. An \emph{embedding algorithm} $E_{\mathfrak{C}}$ for $\mathfrak{C}$ is an efficient classical probabilistic algorithm that takes as input:

\begin{itemize}
\item $\cptp{C} \in \mathfrak{C}$, the computation to be embedded;
\item $G = (V, E)$, a graph, and an output set $O$;
\item $T$, a trappified canvas on graph $G$;
\item $\preceq_G$, a partial order on $V$ which is compatible with the partial order defined by $T$;
\end{itemize}

\noindent and outputs:
\begin{itemize}
\item a partial pattern $C$ on $V \setminus V_T$, with
  \begin{itemize}
  \item input and output vertices $I_C \subset V \setminus V_T$ and $O_C = O \setminus O_T$;
  \item two subspaces (resp.) $\Pi_{I,C}$ and $\Pi_{O,C}$ of (resp.) $I_C$ and $O_C$ with bases (resp.) $\mathcal{B}_{I,C}$ and $\mathcal{B}_{O,C}$;
  \end{itemize}
\item a decoding algorithm $\Deco$;
\end{itemize}
such that the flow $f_C$ of partial pattern $C$ induces a partial order which is compatible with $\preceq_G$. If $E_{\mathfrak C}$ is incapable of performing the embedding, it outputs $\bot$.

\end{definition}

As will be come apparent in later definitions, a good embedding algorithm will yield patterns which apply a desired computation $\cptp C$ to any input state in subspace $\Pi_{I,C}$, with the output being in subspace $\Pi_{O,C}$ after the decoding algorithm has been run. The decoding algorithm can be quantum or classical depending on the nature of the output. Its purpose is to allow the client to recover its desired output from what the server returns at the end of the computation. It could for example be a classical or quantum error-correction decoder. We will furthermore require all embedding algorithms in the paper to have the following property.

\begin{definition}[Proper Embedding]
\label{def:prop-embed}

We say that an embedding algorithm $E_{\mathfrak{C}}$ is \emph{proper} if, for any computation $\cptp{C} \in \mathfrak{C}$ and trappified canvas $T$ that do not result in a $\bot$ output, we have that:
\begin{itemize}
\item $f_C$ does not induce dependencies on vertices $V_T$ of partial pattern $T$;
\item the input and output subspaces $\Pi_{I,C}$ and $\Pi_{O,C}$ do not depend on the trappified canvas $T$.
\end{itemize}

\end{definition}

\begin{example}[Embedding Algorithm on a 2D-Cluster Graph Canvas with a Single Trap]
  \label{ex:embedding}
  Define $\mathfrak C$ as the class of computations that can be implemented using a $(n-3) \times m$ 2D-cluster state. An embedding algorithm for $\mathfrak C$ on $T$ can be defined in the following way.  Consider the trappified canvas $T$ of Example~\ref{ex:canvas} with a $n\times m$ 2D-cluster graph and a single $3\times 3$ trap in the upper left corner.  The output of the embedding algorithm would be the pattern $P$ defined in the following way.  For $C \in \mathfrak C$, by assumption, one can define a pattern $Q$ on a $(n-3) \times m$ 2D-cluster graph that implements $C$.  The angles and flow of the partial pattern $P$ is identical to that of $Q$ albeit applied on the lower $n-3$ rows of $T$. On the $3 \times (m-3)$ upper right rectangular subgraph, all angles are set randomly. $I_{C}$ is such that it comprises all inputs defined in $Q$ and the last $m-3$ qubits of the third row. Choose $\Pi_{I,P}$ so that these $m-3$ qubits are set to $\ket 0$.  Then, by construction, this together with the trap isolates a $(n-3)\times m$ rectangular subgraph on which $P$ will be defining MBQC instructions identical to those of $Q$, thereby implementing $C$. In addition, one can see that there are no dependency between measurements of $P$ and that of the trap in $T$ so that the embedding algorithm is proper. Note that one can change the location of the trap to any column. If in addition the 2D-cluster graph if cylindrical instead of rectangular, the trap can be moved to any location within the cylinder.
\end{example}

\begin{figure}
  \centering
  \includegraphics[width=8cm]{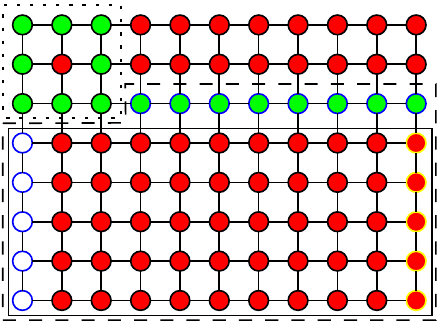}
  \caption{Trappified canvas. Input qubits are surrounded in blue, output qubits in yellow. Red filled qubits are prepared in $\ket +$ while the green ones are prepared in $\ket 0$, white ones are left unspecified. The trap is located in the upper left corner. The actual computation takes place in the $5 \times 11$ rectangular cluster state surrounded by a solid-line while the computation pattern comprises the qubits surrounded by long-dashed line. This allows to include some dummy $\ket 0$ qubits in the inputs so as to disentangle the lower 5 rows from the rest of the graph and perform the computation. Output qubits of the computation are surrounded by a yellow line. The partial pattern inside the dashed box is the trap. The central qubit (red) is prepared in $\ket +$ and is surrounded by $\ket 0$'s (green) that effectively ensure that irrespective of the measurement angles on the remaining qubits the central qubit will remain in $\ket +$. Failure to obtain the 0-outcome when measuring $\X$ will be a proof that the server deviated from the given instructions. The preparation and measurement angles of the remaining qubits is left unspecified.}
  \label{fig:ex3}
\end{figure}

\begin{definition}[Trappified Pattern]
\label{def:trap-pat}

Let $E_{\mathfrak{C}}$ be an embedding algorithm for $\mathfrak{C}$. Given a computation $\cptp C \in \mathfrak{C}$ and a trappified pattern $T$ on graph $G$ with order $\preceq_G$, we call the completed pattern $\TP$ which is the first output of $E_{\mathfrak{C}}(\cptp C, G, T, \preceq_G)$ a \emph{trappified pattern}.

\end{definition}

While embedding a computation in a graph that has enough space for it might seem simple, the hard part is to ensure that the embedding is \emph{proper}.  This property implies that no information is carried via the flow of the global pattern from the computation to the trap and it is essential for the security of the verification protocol built using trappified canvases. In Example~\ref{ex:embedding} above, this is done by breaking the graph using the states initialised in $\ket{0}$. The only other known way is to separate runs for tests and computations and satisfying this condition using other methods is left as an open question.

Note that the input and output qubits of the computation $\cptp{C}$ might be constrained to be in (potentially strict) subspaces $\Pi_{I, C}$ and $\Pi_{O, C}$ of $I_C$ and $O_C$ respectively. This allows for error-protected inputs and outputs, without having to specify any implementation for the error-correction scheme.  In particular, it encompasses encoding classical output data as several, possibly noisy, repetitions which will be decoded by the client through a majority vote as introduced in \cite{LMKO21:verifying}.  It also allows to take into account the case where the trappified pattern comprises a fully fault-tolerant MBQC computation scheme for computing $\cptp C$ using topological codes as described in \cite{RHG07:topological}.

For verification, our scheme must be able to cope with malicious behaviour: detecting deviations is useful for verification only so long as the server cannot adapt its behaviour to the traps that it executes. Otherwise, it could simply decide to deviate exclusively on non-trap qubits. 
This is achieved by executing the patterns in a blind way so that the server has provably no information about the location of the traps and cannot avoid them with high probability. 
To this end, we define \emph{blind-compatible} patterns as those which share the same graph, output vertices and measurement order of their qubits. The UBQC Protocol described in Appendix \ref{sec:prelims} leaks exactly this information to the server, meaning that it cannot distinguish the executions of two different blind-compatible patterns.

\begin{definition}[Blind-Compatibility]
\label{def:compat}
A set of patterns $\sch P$ is \emph{blind-compatible} if all patterns $P \in \sch P$ share the same graph $G$, the same output set $O$ and there exists a partial ordering $\preceq_{\sch P}$ of the vertices of $G$ which is an extension of the partial ordering defined by the flow of any $P \in \sch P$. This definition can be extended to a set of trappified canvases $\sch P = \{(T, \sigma, \pd{T}, \tau)\}$. The partial order $\preceq_{\sch P}$ is required to be an extension of the orderings $\preceq_T$ of partial patterns $T$.
\end{definition}

A single trap is usually not sufficient to catch deviations on more than a subset of positions of the graph. In order to catch all deviations, it is then necessary to randomise the blind delegated execution over multiple patterns.
We therefore define a trappified scheme as a set of blind-compatible trappified canvases which can be efficiently sampled according to a given distribution, along with an algorithm for embedding computations from a given class into all the canvases.

\begin{definition}[Trappified Scheme]
\label{def:trap-scheme}
A \emph{trappified scheme} $(\sch P, \preceq_G, \pd P, E_{\mathfrak{C}})$ over a graph $G$ for computation class $\mathfrak{C}$ consists of:
\begin{itemize}
\item $\sch P$, a set of \emph{blind-compatible} trappified canvases over graph $G$ with common partial order $\preceq_{\sch P}$;
\item $\preceq_G$, a partial ordering of vertices $V$ of $G$ that is compatible with $\preceq_{\sch P}$;
\item $\pd P$, a probability distribution over the set $\sch P$ which can be sampled efficiently;
\item $E_{\mathfrak{C}}$, a \emph{proper} embedding algorithm for $\mathfrak{C}$.
\end{itemize}
\end{definition}

Without loss of generality, in the following, the probability distribution used to sample the trappified canvases will generally be $u(\sch P)$, the uniform distribution over $\sch P$.
The general case can be approximated from the uniform one with arbitrary fixed precision by having several copies of the same canvas in $\sch P$.
We take $T \sim \sch P$ to mean that the trappified canvas is sampled according to the distribution $\pd P$ of trappified scheme $\sch P$.

Note that in Definition~\ref{def:trap-scheme} above, while the blindness condition ensures that a completed pattern obtained after running the embedding algorithm hides the location of the traps, the existence of a partial order $\preceq_{G}$ compatible with that of the trappified canvases ensures that this remains true when considering the scheme as a whole, i.e the order in which the qubits are measured does not reveal information about the chosen trappified canvas itself, which would otherwise break the blindness of the scheme. 

\begin{example}[Trappified Scheme for a Cylindrical-Cluster Graph]
  Consider the set of trappified canvases together with the embedding algorithm $E_{\mathfrak C}$ on the cylindrical cluster-graph with a single randomly placed $3 \times 3$ trap as defined in Example~\ref{ex:embedding}. This defines a trapification scheme for $\mathfrak{C}$ consisting of computations that can be implemented using a $(n-3)\times m$ 2D-cluster graph (See Figure~\ref{fig:ex4}).
\end{example}

\begin{figure}
  \centering
  \includegraphics[width=8cm]{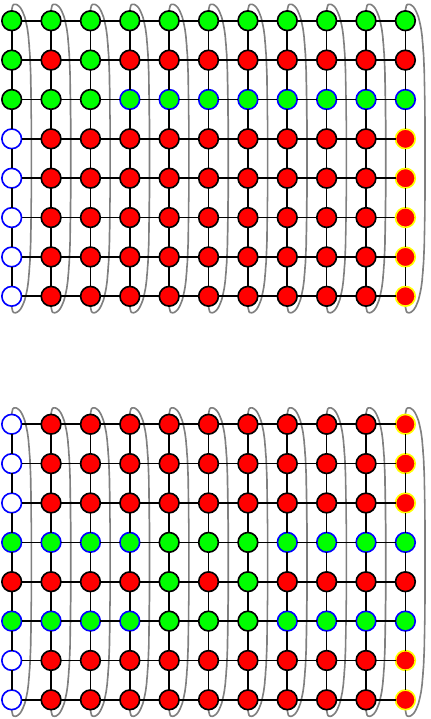}
  \caption{Trappified scheme. Two possible canvas extracted from a trappified scheme with a single $3\times 3$ trap on a toric $8\times 11$ toric cluster state. Input qubits are surrounded in blue, output qubits in yellow. Red filled qubits are prepared in $\ket +$ while the green ones are prepared in $\ket 0$, white ones are left unspecified. The actual computation takes place in the $5 \times 11$ rectangular cluster state, while the trap is located at a different positions in each picture allowing to detect all possible deviations performed by a server, albeit with a low probability of success.}
  \label{fig:ex4}
\end{figure}

\subsection{Effect of Deviations on Traps}
\label{subsec:detect-prop}

We can now describe the purpose of the objects described in the previous subsection, namely detecting the server's deviations from their prescribed operations during a given delegated computation. We start by recalling that the blindness of UBQC Protocol is obtained by Pauli-twirling the operations delegated to the server. This implies that any deviation can be reduced to a convex combination of Pauli operators. Then, we formally define Pauli deviation detection and insensitivity for trappified canvases and schemes. We show in the next section that these key properties are sufficient for obtaining a verifiable delegated computation by formalising the steps sketched here.

When a client delegates the execution of a pattern $P$ to a server using Protocol~\ref{proto:ubqc}, the server can potentially deviate in an arbitrary way from the instructions it receives. 
By converting into quantum states both the classical instructions sent by the client -- i.e.~the measurement angles -- and the measurement outcomes sent back by the server, all operations on the server's side can be modelled as a unitary $\cptp F$ acting on all the qubits sent by the client and some ancillary states $\ket 0_S$, before performing measurements in the computational basis to send back the outcomes $\ket b$ that the client expects from the server.

The instructions of the server in an honest execution of the UBQC Protocol \ref{proto:ubqc} correspond to:
\begin{enumerate}
\item entangling the received qubits corresponding to the vertices of the computation graph with operation $\cptp G = \prod_{(i,j) \in E} \CZ_{i, j}$;
\item performing rotations on non-output vertices around the $\Z$-axis, controlled by the qubits which encode the measurement angles instructed by the client;
\item applying a Hadamard gate $\Ha$ on all non-output vertices;
\item measuring non-output vertices in the $\{\ket 0, \ket 1\}$ basis.
\end{enumerate}
The steps (i-iii) correspond to a unitary transformation $\cptp U_P$ that depends only on the public information that the server has about the pattern $P$ -- essentially the computation graph $G$ and an order of its vertices compatible with the flow of $P$.

Hence, the unitary part $\cptp U_P$ of the honestly executed protocol for delegating $P$ can always be extracted from $\cptp F$, so that $\cptp F = \cptp F' \circ \cptp U_P$. 
Here, $\cptp F'$ is called a \emph{pure deviation} and is applied right before performing the computational basis measurements for non-output qubits and right before returning the output qubits to the Client.

When the pattern is executed blindly using Protocol~\ref{proto:ubqc}, the state in the server's registers during the execution is a mixed state over all possible secret parameters chosen by the client. It is shown in \cite{K16:efficient} that the resulting summation over the secret parameters which hide the inputs, measurement angles and measurement outcomes is equivalent to applying a Pauli twirl to the pure deviation $\cptp F'$.
This effectively transforms it into a convex combination of Pauli operations applied after $\cptp U_P$.

Hence, any deviation by the server can be represented without loss of generality by choosing with probability $\Pr[\cptp E]$ an operator $\cptp E$ in the Pauli group $\mathcal P_V$ over the vertices $V$ of the graph used to define $P$, and executing $\cptp E \circ \cptp U_P$ instead of $\cptp U_P$ for the unitary part of the protocol.
By a slight abuse of notation, such transformation will be denoted $\cptp E\circ P$.
Furthermore, if $\TP$ is a trappified pattern obtained from a trappified canvas $T$ that samples $t = (t_1, \ldots, t_N)$ from the distribution $\pd T$, then in the presence of deviation $\cptp E$, it will sample from a different distribution.
For instance, whenever $\cptp E$ applies a $\Z$ operator on a vertex, it can be viewed as an execution of a pattern where the angle $\delta$ for this vertex is changed into $\delta + \pi$. Whenever $\cptp E$ applies a $\X$ operator on a vertex, $\delta$ is transformed into $-\delta$.
We now give a lemma which will be useful throughout the rest of the paper.
\begin{lemma}[Independence of Trap and Computation]
\label{lem:indep-t}
Let $\TP$ be a trappified pattern obtained from the trappified canvas $T$ which samples from distribution $\pd T$ through a proper embedding algorithm of computation $\cptp C$. Then, for all Pauli errors $\cptp E$, the distribution of trap measurement outcomes is independent of the computation $\cptp C$ and of the input state in the subspace $\Pi_{I, C}$.
\end{lemma}
\begin{proof}
Let $f_C$ be the flow of computation of the embedded computation $\cptp C$. Because the embedding is proper according to Definition \ref{def:embed-alg}, the dependencies induced by $f_C$ do not affect trap qubits $V_T$. Furthermore, the input of the trap is fixed along with its partial pattern, independently of the computation. Therefore, the distribution of the trap measurement outcomes is also independent of the embedded computation being performed on the rest of the graph as well as the input state of such computation.
\end{proof}
Indeed, for a completed trappified pattern $\TP$ obtained by embedding a computation $\cptp C$ onto a trappified canvas $T$, the action of $\cptp E$ on the vertices outside $V_T$ does not have an impact on the measurement outcomes of the vertices in $V_T$. This allows to define the trap outcome distribution under the influence of error $\cptp E$ solely as a function of $\cptp E$ and $\pd{T}$. Such modified distribution is denoted $\pd T_{\cptp E}$.

As an additional consequence, it is possible to define what it means for a given trappified canvas to detect and to be insensitive to Pauli errors:

\begin{definition}[Pauli Detection]
  \label{def:detect_tp}
  Let $T$ be a trappified canvas sampling from distribution $\pd T$.
  Let $\mathcal{E}$ be a subset of the Pauli group $\mathcal P_V$ over the graph vertex qubits.
  For $\epsilon > 0$, we say that $T$ \emph{$\epsilon$-detects} $\mathcal{E}$ if: 
  \begin{align}
    \forall \cptp E \in \mathcal{E},\ \Pr_{t\sim \pd T_{\cptp E}}[\tau(t) = 1] \geq 1-\epsilon.
  \end{align}
  
  We say that a trappified scheme $\sch P$ \emph{$\epsilon$-detects} $\mathcal{E}$ if:
  \begin{align}
    \forall \cptp E \in \mathcal{E}, \ \sum_{T\in \sch P} \Pr_{\substack{T \sim \pd P \\ t\sim \pd T_{\cptp E}}}[\tau(t) = 1, T] \geq 1-\epsilon.
  \end{align}
\end{definition}

\begin{definition}[Pauli Insensitivity]
  \label{def:insens_tp}
  Let $T$ be a trappified canvas sampling from distribution $\pd T$.
  Let $\mathcal{E}$ be a subset of $\mathcal P_V$.
  For $\delta > 0$, we say that $T$ is \emph{$\delta$-insensitive to} $\mathcal{E}$ if:   
  \begin{align}
     \forall \cptp E \in \mathcal{E}, \ \Pr_{t\sim \pd T_{\cptp E}}[\tau(t) = 0] \geq 1-\delta.
  \end{align}
  
  We say that a trappified scheme $\sch P$ is \emph{$\delta$-insensitive to} $\mathcal{E}$ if:
  \begin{align}
    \forall \cptp E \in \mathcal{E}, \ \sum_{T\in \sch P} \Pr_{\substack{T\sim \pd P \\ t\sim \pd T_{\cptp E}}}[\tau(t) = 0, T] \geq 1-\delta.
  \end{align}
\end{definition}

Above, the probability distribution stems both from the randomness of quantum measurements of the trap output qubits yielding the bit string $t$, and the potentially probabilistic nature of the decision function $\tau$. 
In the case of trappified schemes, the probability distribution for obtaining a given result for $\tau$ also depends on the choice of canvas $T \in \sch P$, sampled according to the probability distribution $\pd P$.

In the same spirit, there are deviations that nonetheless produce little effect on the computations embedded into trappified canvases and trappified schemes. When they occur, the computation is still almost correct.\footnote{Most generally, this can depend on the embedding used to insert the computation into the canvas but in most existing cases it is not. Hence why we phrase it as a property of the trappified canvas itself.}

\begin{definition}[Pauli Correctness]
\label{def:trap-cor}
  Let $(T, \sigma, \pd T, \tau)$ be a trappified canvas on graph $G$, $\preceq_G$ an order on the vertices of $G$ and $E_{\mathfrak C}$ an embedding algorithm. Let $\TP$ be the trappified pattern obtained by embedding a computation $\cptp C \in \mathfrak{C}$ on $T$ using $E_{\mathfrak C}$ and order $\preceq_G$. Let $I_C$ be the set of input vertices for the computation $\cptp C$ in $\TP$ and let $\ket{\psi}$ be a state on $\abs{I_C}+\abs{R}$ qubits, for sufficiently large auxiliary system $R$, such that $\Tr_R(\ket{\psi}) \in \Pi_{I,C}$, where $\Pi_{I,C}$ is the client's input subspace.
  Let $\mathcal E$ be a subset of  $\mathcal P_V$. For $\cptp E \in \mathcal{E}$, we define $\Cdev{\cptp E} = \Deco\circ\Tr_{O_C^c}\circ\cptp E\circ (\TP)$ to be the CPTP map resulting from applying the trappified pattern $\TP$ followed by the decoding algorithm $\Deco$ on the output of the computation.
  For $\nu \geq 0$, we say that $T$ is \emph{$\nu$-correct} on $\mathcal E$ if:\footnote{Equation \ref{eq:pauli-cor} corresponds to the diamond norm between the correct and deviated CPTP maps, but with a fixed input subspace and a fixed input for the trap qubits.}
  \begin{align}
  \label{eq:pauli-cor}
     \forall \cptp E \in \mathcal E, \ \forall \cptp C \in \mathfrak{C}, \ \max_\psi\| (\Cdev{\cptp E}  - \cptp C\otimes \Id_T)\otimes\Id_R [\dyad{\psi} \otimes \sigma]\|_{\Tr} \leq \nu.
  \end{align}

\noindent This is extended to a trappified scheme $\sch P$ by requiring the bound to hold on average over $T \in \sch P$:
\begin{equation}
\forall \cptp E \in \mathcal E, \ \forall \cptp C \in \mathfrak{C}, \max_\psi \left(\sum_{T\in \sch P} \Pr_{T\sim \pd P}[T] \| (\Cdev{\cptp E}  - \cptp C\otimes \Id_T)\otimes\Id_R [\dyad{\psi} \otimes \sigma]\|_{\Tr}  \right) \leq \nu.
\end{equation}
\end{definition}

In the following, the sets of deviations that have little effect on the result of the computation according to diamond distance will be called \emph{harmless}, while their complement are \emph{possibly harmful}.

We conclude this section with some remarks regarding basic properties of trappified schemes and a simple but powerful result allowing to construct trappification schemes from simpler ones.

\begin{remark}[Consequence of Embedding Failure]
Given a trappified canvas $T$ on a graph $G$, $\mathfrak{C}$, embedding algorithm $E_{\mathfrak{C}}$ and order $\preceq_G$ such that there exists $\cptp C \in \mathfrak{C}$ with $E_{\mathfrak{C}}(\cptp{C}, G, T, \preceq_G) = \bot$, then $T$ is $1$-correct for $\mathcal{P}_V$.
\end{remark}

\begin{remark}[Existence of Harmless Deviations]
\label{rem:dev-struct}
Why not just detect all possible deviations rather than count on the possibility that some have little impact on the actual computation? The reason is that these are plentiful in MBQC. Following our convention to view all measurements as computational basis measurements preceded by an appropriate rotation, any deviation $\cptp E$ that acts as $\Id$ and $\Z$ on measured qubits does not change the measurement outcomes and have no effect on the final outcome. Consequently, for classical output computations, only $\X$ and $\Y$ deviations need to be analysed. These are equivalent to flipping the measurement outcome, which propagate to the output via the flow corrections.
\end{remark}

\begin{remark}[A Trappified Canvas is a Trappified Scheme]
\label{rem:tr-can-sch}
	Any trappified canvas $T$ -- together with an embedding algorithm -- can be seen as a trappified scheme $\sch P = \{T\}$ and the trivial distribution. If the trappified pattern $\epsilon$-detects $\mathcal{E}_\epsilon$, is $\delta$-insensitive to $\mathcal{E}_\delta$ and $\nu$-correct on $\mathcal{E}_\nu$ for computations in $\mathfrak{C}$, so is the corresponding trappified scheme.
\end{remark}

\begin{remark}[Pure Traps]
\label{rem:pure_traps}
A trappified scheme $\sch P$ may only consist of trappified canvases that cover the whole graph $G = (V, E)$ if $V_T = V$ for all $T \in \sch P$. This corresponds to the special case where the trappified scheme cannot embed any computation and the embedding algorithm applied to a canvas $T \in \sch P$ always return $T$. The detection, insensitivity and correctness properties also apply to this special case, although $\nu$ is trivially equal to $0$ for $\mathfrak{C} = \emptyset$ and equal to $1$ for $\mathfrak{C} \neq \emptyset$.

\end{remark}

\begin{lemma}[Simple Composition of Trappified Schemes]
\label{lem:compos-sch}
	Let $(\sch P_i)_i$ be a sequence of trappified schemes over the same graph $G$ which are mutually blind compatible, with corresponding distributions $\pd{P}_i$, such that $\sch P_i$ $\epsilon_i$-detects $\mathcal{E}_{\epsilon, i}$, is $\delta_i$-insensitive to $\mathcal{E}_{\delta, i}$ and $\nu_i$-correct on $\mathcal{E}_{\nu, i}$ for computations in $\mathfrak{C}$.
	
	Let $(p_i)_i$ be a probability distribution and $\sch P = \bigcup_i \sch P_i$ be the trappified scheme with the following distribution $\pd{P}$:
	\begin{enumerate}
		\item Sample a trappified scheme $\sch P_j$ from $(\sch P_i)_i$ according to $(p_i)_i$;
		\item Sample a trappified canvas from $\sch P_j$ according to $\mathcal{P}_j$.
	\end{enumerate}
	The embedding function of $\sch P$ simply uses the embedding function of the sampled scheme $\sch P_j$.
	
	Let $\mathcal{E}_\epsilon \subseteq \bigcup_i \mathcal{E}_{\epsilon, i}$, $\mathcal{E}_\delta \subseteq \bigcup_i \mathcal{E}_{\delta, i}$ and $\mathcal{E}_\nu \subseteq \bigcup_i \mathcal{E}_{\nu, i}$.
	Then, $\sch P$ $\epsilon$-detects $\mathcal{E}_\epsilon$, is $\delta$-insensitive to $\mathcal{E}_\delta$ and $\nu$-correct on $\mathcal{E}_\nu$ for computations in $\mathfrak{C}$ with
	\begin{align}
		&1-\epsilon = \min_{\cptp E\in\mathcal{E}_\epsilon} \sum_{\substack{i \\ \cptp E\in\mathcal{E}_{\epsilon, i}}} p_i (1-\epsilon_i), \\
		&1-\delta = \min_{\cptp E\in\mathcal{E}_\delta} \sum_{\substack{i \\ \cptp E\in\mathcal{E}_{\delta, i}}} p_i (1-\delta_i), \mbox{ and }\\
		&1-\nu = \min_{\cptp E\in\mathcal{E}_\nu} \sum_{\substack{i \\ \cptp E\in\mathcal{E}_{\nu, i}}} p_i (1-\nu_i).
	\end{align}
\end{lemma}

\section{Secure Verification from Trap Based Protocols}
\label{sec:verif}
In this section we use the properties defined above to derive various results which help break down the tasks of designing and proving the security of verification protocols into small and intuitive pieces. We start by giving a description of a general protocol using trappified schemes which encompasses all prepare-and-send MBQC-based protocol aiming to implement the SDQC functionality (Definition~\ref{def:bvqc_if}). We then relate the security of this protocol in the Abstract Cryptography framework to the $\epsilon$-detection, $\delta$-insensitivity and $\nu$-correctness of the trappified scheme used in the protocol. Consequently, we can from then on only focus on these three properties instead of looking at the full protocol, which already removes a lot of steps in future proofs.

Then we demonstrate how increasing the insensitivity set yields a protocol which is robust to situations where the server is honest-but-noisy with a contained noise parameter. These results further simplify the design of future protocols since many complex proofs can be avoided, allowing us to concentrate on designing more efficient trappified schemes and directly plugging them into the generic protocol and compiler to yield exponentially-secure and noise-robust protocols implementing SDQC.
We finally describe a consequence of these results in the case where the security error of the protocol is exponentially close to $0$. We show that this automatically implies that the computation must be protected against low-weight errors if we restrict the server's resources to be polynomial in the security parameter.

\subsection{General Verification Protocol from Trappified Schemes}
Given a computation $\cptp C$, it is possible to delegate its trappified execution in a blind way. To do so, the Client simply chooses one trappified canvas from a scheme at random, inserts into it the computation $\cptp C$ using an embedding algorithm and blindly delegates the execution of the resulting trappified pattern to the Server. The steps are formally described in Protocol~\ref{proto:dev_detect}.

\begin{protocol}[ht]
  \caption{Trappified Delegated Blind Computation}
  \label{proto:dev_detect}
  \begin{algorithmic}[0]
  
    \STATE \textbf{Public Information:} 
	\begin{itemize}
		\item $\mathfrak{C}$, a class of quantum computations;
		\item $G = (V, E)$, a graph with output set $O$;
		\item $\sch P$, a trappified scheme on graph $G$;
		\item $\preceq_G$, a partial order on $V$ compatible with $\sch P$.
	\end{itemize}
    
    \STATE \textbf{Client's Inputs:} A computation $\cptp C \in \mathfrak{C}$ and a quantum state $\rho_C$ compatible with $\cptp C$.
    
    \STATE \textbf{Protocol:}
    \begin{enumerate}
    	\item The Client samples a trappified canvas $T$ from the trappified scheme $\sch P$.
	    \item The Client runs the embedding algorithm $E_{\mathfrak{C}}$ from $\sch P$ on its computation $\cptp C$, the graph $G$ with output space $O$, the trappified canvas $T$, and the partial order $\preceq_G$. It obtains as output the trappified pattern $\TP$.
	    \item The Client and Server blindly execute the trappified pattern $\TP$ on input state $\rho_C$ using the UBQC Protocol~\ref{proto:ubqc}.
	    \item If the output set is non-empty (if there are quantum outputs), the Server returns the qubits in positions $O$ to the Client.
	    \item The Client measures the qubits in positions $O \cap V_T$ in the $\X$ basis. It obtains the trap sample $t$.
	    \item The Client checks the trap by computing $\tau(t)$:
	    \begin{itemize}
	    	\item If $\tau(t) = 1$, it rejects and outputs $(\bot, \rej)$.
	    	\item Otherwise, the Client accepts the computation. It applies the decoding algorithm $\Deco$ to the output of Protocol~\ref{proto:ubqc} on vertices $O \setminus V_T$ and set the result as its output along with $\acc$.
		\end{itemize}	    
    \end{enumerate}
    
  \end{algorithmic}
\end{protocol}

Note that this protocol offers blindness not only at the level of the chosen trappified pattern, but also at the level of the trappified scheme itself.
More precisely, by delegating the chosen pattern, the client reveals at most the graph of the pattern, a partial order of its vertices and the location of the output qubits of the pattern, if there are any, comprising computation and trap outputs.
However, trappified patterns of a trappified scheme are blind-compatible, that is they share the same graph and same set of output qubits.
Therefore, the above protocol also hides which trappified pattern has been executed among all possible ones, hence concealing the location of traps.

We now formalise the following intuitive link between deviation detection and verification in the context of delegated computations.
On one hand, if a delegated computation protocol is correct,\footnote{Here we use correctness in a cryptographic setting, meaning that all parties execute as specified their part of the protocol.} not detecting any deviation by the server from its prescribed sequence of operations should be enough to guarantee that the final result is correct.
Conversely, detecting that some operations have not been performed as specified should be enough for the client to reject potentially incorrect results.
Combining those two cases should therefore yield a verified delegated computation.

To this end, we show how the deviation detection capability of trappified schemes is used to perform verification.
This is done by proving that Protocol \ref{proto:dev_detect} above constructs the Secure Delegated Quantum Computation Resource \ref{def:bvqc_if} in the Abstract Cryptography framework. 
This resource allows a Client to input a computation and a quantum state and to either receive the correct outcome or an abort state depending on the Server's choice, whereas the Server only learns at most some well defined information contained in a leak $l_{\rho}$. 
More precisely, we show that any distinguisher has a bounded distinguishing advantage between the real and ideal scenarios so long as the trappified scheme $\sch P$ detects a large fraction of deviations that are possibly harmful to the computation.

\begin{theorem}[Detection Implies Verifiability]
  \label{thm:verif}
  Let $\sch P$ be a trappified scheme with a proper embedding. Let $\mathcal E_\epsilon$ and $\mathcal E_\nu$ be two sets of Pauli deviations such that:
\begin{itemize}
\item $\mathcal P_V \setminus \mathcal E_\epsilon \subset \mathcal E_\nu$;
\item $\Id \in \mathcal{E}_\nu$.
\end{itemize}
If $\sch P$:
\begin{itemize}
\item $\epsilon$-detects $\mathcal E_\epsilon$;
\item is $\delta$-insensitive to at least $\{\Id\}$;
\item is $\nu$-correct on $\mathcal E_\nu$;
\end{itemize}
for $\epsilon, \delta, \nu > 0$, then the Trappified Delegated Blind Computation Protocol~\ref{proto:dev_detect} for computing CPTP maps $\cptp C$ in $\mathfrak C$ using $\sch P$ is $\delta + \nu$-correct and $\max(\hat{\epsilon}, \nu)$-secure in the Abstract Cryptography framework, for
\begin{equation}
\hat{\epsilon} = \epsilon \times \max_{\substack{\cptp E \in \mathcal E_\epsilon\\ \cptp C \in \mathfrak{C} \\ \psi_C}}\left(\sum_{T\in \sch P} \Pr_{T \sim \pd P}[T]\|(\Cdev{\cptp E}  - \cptp C\otimes \Id_T)\otimes\Id_D [\dyad{\psi_C} \otimes \sigma]\|_{\Tr}\right) \leq \epsilon,
\end{equation}
where $\ket{\psi_C}$ is a purification of the client's input using register $D$ and $\Cdev{\cptp E} = \Deco\circ\Tr_{O_C^c}\circ\cptp E\circ (\TP)$ corresponds to running the trappified canvas $\TP$ with a deviation $\cptp E$ and decoding the output of the computation using the decoding algorithm $\Deco$ provided by the embedding algorithm.

Overall, Protocol~\ref{proto:dev_detect} $\max(\hat{\epsilon}, \delta + \nu)$-constructs the Secure Delegated Quantum Computation Resource~\ref{def:bvqc_if} where the leak is defined as $l_{\rho} = (\mathfrak{C}, G, \sch P, \preceq_G)$.
\end{theorem}

We give the finer characterisation using $\hat{\epsilon}$ instead of bounding it with $\epsilon$ since it may be useful in some scenarii to more finely bound the trace distance for these deviations instead of bluntly bounding it by $1$.

Figure~\ref{fig:deviation_sets} depicts the configuration of the various sets of Pauli deviations used in Theorems~\ref{thm:verif} and \ref{thm:robust_verif}.
\begin{figure}[ht]
\centering
\begin{picture}(0,0)\includegraphics{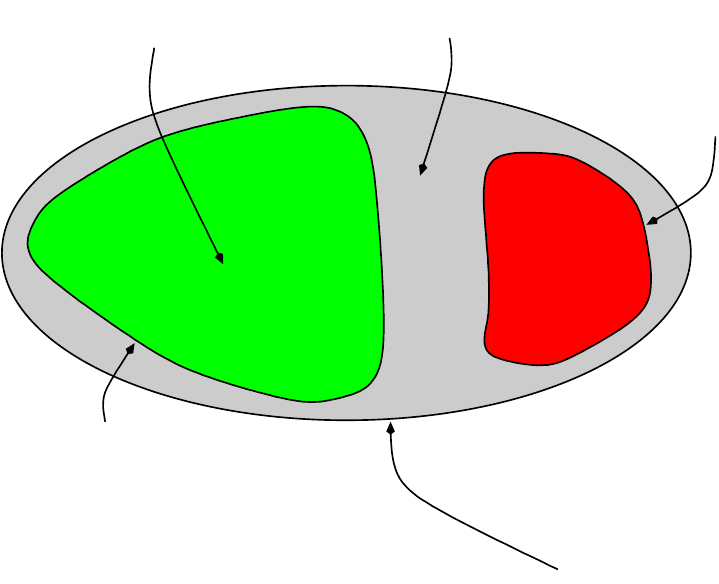}\end{picture}\setlength{\unitlength}{4144sp}\begingroup\makeatletter\ifx\SetFigFont\undefined \gdef\SetFigFont#1#2#3#4#5{\reset@font\fontsize{#1}{#2pt}\fontfamily{#3}\fontseries{#4}\fontshape{#5}\selectfont}\fi\endgroup \begin{picture}(3282,2659)(16,-1430)
\put(2611,-1366){\makebox(0,0)[lb]{\smash{{\SetFigFont{9}{10.8}{\rmdefault}{\mddefault}{\updefault}{\color[rgb]{0,0,0}$\mathcal{P}_V$}}}}}
\put(631,1064){\makebox(0,0)[lb]{\smash{{\SetFigFont{9}{10.8}{\rmdefault}{\mddefault}{\updefault}{\color[rgb]{0,0,0}$\Id$}}}}}
\put(2881,659){\makebox(0,0)[lb]{\smash{{\SetFigFont{9}{10.8}{\rmdefault}{\mddefault}{\updefault}{\color[rgb]{0,0,0}$\mathcal{E}_\epsilon$: $\epsilon$-detected by $\sch P$}}}}}
\put( 46,-871){\makebox(0,0)[lb]{\smash{{\SetFigFont{9}{10.8}{\rmdefault}{\mddefault}{\updefault}{\color[rgb]{0,0,0}$\mathcal{E}_\delta$: $\sch P$ is $\delta$-insensitive}}}}}
\put(1756,1109){\makebox(0,0)[lb]{\smash{{\SetFigFont{9}{10.8}{\rmdefault}{\mddefault}{\updefault}{\color[rgb]{0,0,0}$\mathcal{E}_\nu \supset \mathcal P_V \setminus \mathcal E_\epsilon$: $\sch P$ is $\nu$-correct}}}}}
\end{picture} \caption{Setup of the Pauli deviations sets used in Theorems~\ref{thm:verif} and \ref{thm:robust_verif}.
Security requires that $\mathcal E_\epsilon$ is $\epsilon$-detected by $\sch P$.
Correctness requires that $\sch P$ is $\nu$-correct on at least the complement of $\mathcal E_\epsilon$.
Robustness comes from $\sch P$ being $\delta$-insensitive to $\mathcal E_\delta$, on which it must also be $\nu$-correct.}
\label{fig:deviation_sets}
\end{figure}
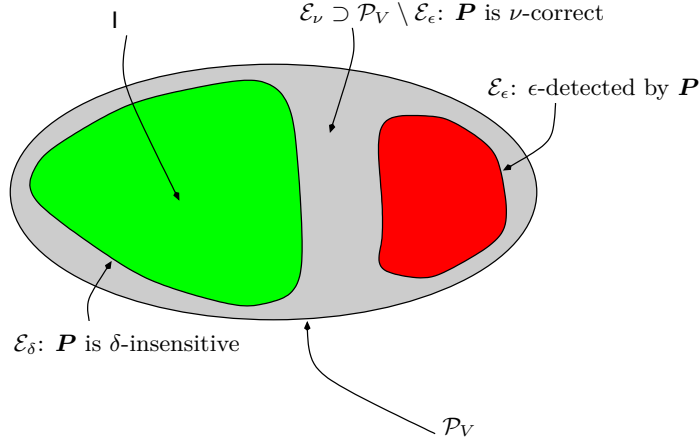

\begin{proof}[\bf Proof of Correctness]
We start by analysing the correctness of Protocol \ref{proto:dev_detect}, i.e.~ the distance between the real and ideal input/output relation if both parties follow their prescribed operations. Let $\cptp C \in \mathfrak C$ be the Client's desired computation. Let $\rho_C$ be the Client's input state and $\ket{\psi_C}$ a purification of $\rho_C$ using the distinguisher's register $D$. Let $\TP$ be a trappified pattern obtained from sampling a trappified canvas $T$ from the trappified scheme $\sch P$ using probability distribution $\pd P$ and embedding computation $\cptp C$ into it using the embedding algorithm.

The output in the ideal case can be written as $\cptp C \otimes \Id_D[\dyad{\psi_C}]$, while in the real case the state after the execution of the trappified pattern is $\Tr_{O_C^c}(\TP[\rho_C\otimes\sigma])\otimes\dyad{\tau(t)}$, where the trace is over all registers not containing the output of the Client's computation.\footnote{We use here the notation $P[\rho]$ to mean the honest application of the trappified pattern $P$ to the input state $\rho$. Also we consider here that the decision function $\tau$ outputs either $\acc$ for acceptance or $\rej$ for rejection instead of a binary value.} If we denote $\Cdev{\Id} = \Deco\circ\Tr_{O_C^c}\circ \TP$ and take the average over the choice of trappified canvas by the Client which is unknown to the distinguisher, then the distinguishing advantage is given by
\begin{align}
	\epsilon_{\mathit{cor}} = \| \cptp C \otimes \Id_D[\dyad{\psi_C}]\otimes\dyad{\acc} - \sum_{T\in \sch P} \Pr_{T\sim \pd P}[T] \Cdev{\Id} \otimes \Id_D[\dyad{\psi_C}\otimes\sigma]\otimes\dyad{\tau(t)}\|_{\Tr},
\end{align}

In the honest case, the concrete and ideal settings will output different states only in the case where the protocol wrongly rejects the computation or outputs a wrong result despite the absence of errors.

Since the trappified scheme is $\delta$-insensitive to $\Id$, the probability that the decision function outputs $\rej$ is bounded by $\delta$ as per Definition \ref{def:insens_tp}. Furthermore, using Lemma \ref{lem:indep-t}, the output of the test is independent of the computation being performed. Combining these two properties yields
\begin{align}
\label{eq:cor}
\begin{split}
	\epsilon_{\mathit{cor}} &\leq \| \cptp C \otimes \Id_D[\dyad{\psi_C}]\otimes\dyad{\acc} - \\
	&\quad \sum_{T\in \sch P} \Pr_{T\sim \pd P}[T]\Cdev{\Id} \otimes \Id_D[\dyad{\psi_C}\otimes\sigma]\otimes(\delta\dyad{\rej} + (1-\delta)\dyad{\acc})\|_{\Tr}.
\end{split}
\end{align}
Using the convexity of the trace distance, we get
\begin{align}
\epsilon_{\mathit{cor}}	\leq (1-\delta)\sum_{T\in \sch P} \Pr_{T\sim \pd P}[T]\| \cptp C \otimes \Id_D[\dyad{\psi_C}] - \Cdev{\Id} \otimes \Id_D[\dyad{\psi_C}\otimes\sigma]\|_{\Tr} + \delta.
\end{align}

Finally, the trappified scheme is $\nu$-correct on $\Id \in \mathcal E_\nu$. Therefore we have that
\begin{equation}
\sum_{T\in \sch P} \Pr_{T\sim \pd P}[T]\| \cptp C \otimes \Id_D[\dyad{\psi_C}] - \Cdev{\Id} \otimes \Id_D[\dyad{\psi_C}\otimes\sigma]\|_{\Tr} \leq \nu,
\end{equation}
meaning that $\epsilon_{\mathit{cor}} \leq (1-\delta)\nu + \delta$. Hence, the protocol is $(\delta+\nu)$-correct.
\end{proof}

\begin{proof}[\bf Proof of Security against Malicious Server]
To prove the security of the protocol, as per Definition \ref{const-sec-def}, we first describe our choice of Simulator that has access to the Server's interface of the Secure Delegated Quantum Computation Resource. Then, we analyse the interaction involving either the Simulator or the real honest Client and show that they are indistinguishable. For this latter part, we will decompose the analysis into the following steps. First, we describe the state sent by the Client or Simulator to the Server. Second, we derive the state after the interaction between the Client or Simulator with the Server, encompassing all possible deviations from the honest protocol. Third, we analyse the resulting state as seen from the Distinguisher's point of view, i.e.~knowing the chosen input, the chosen computation and the Server-side deviation, but ignorant of the secret parameters set by the Client's protocol or the Simulator. Fourth, using the composable security of UBQC, we show that no information about the real or ideal setup is leaked on the Server's side to help the Distinguisher tell them apart. Fifth, we analyse the remaining output and abort probabilities. Sixth, we bound the distinguishing probability by analysing all possible deviation choices, which then concludes the proof.

\textit{Defining the Server's Simulator.}
To do so, we use again the fact that when the protocol is run and a deviation is applied by the Server, the probability of accepting or rejecting the computation is dependent only on the deviation and not on the computation performed on the non-trap part of the pattern. This is a crucial property as this allows to simulate the behaviour of the concrete protocol even when the computation performed is unknown. More precisely, we define the Simulator in the following way:
\begin{simulator}[h!]
\caption{}
  \begin{enumerate}
  \item The Simulator request a leak from the Secure Delegated Quantum Computation Resource and receives in return $(\mathfrak{C}, G, \sch P, \preceq_G)$.
  \item It chooses at random any computation $\cptp C_\emptyset \in \mathfrak C$ and an input which is compatible with $\cptp C$.
  \item It performs the same tasks as those described by the Client's side of the Trappified Delegated Blind Computation Protocol~\ref{proto:dev_detect}.
  \item Whenever $\tau$ accepts, the Simulator sends $c=0$ to the Secure Delegated Quantum Computation Resource, indicating that the honest Client should receive its output. If it rejects, the Simulator sends $c=1$ Secure Delegated Quantum Computation Resource, indicating an abort.
  \end{enumerate}
\end{simulator}

\textit{State Sent to the Server.}
Here we only describe the state representing the interaction of the Client or of the Simulator with the Server. Since the Simulator defined above performs the same tasks as the Client when the Protocol is run, we only need to derive the expression for the Client's interaction. The expression for the Simulator is obtained by replacing the Client's secrets with that of the Simulator, and the computation $\cptp C$ by $\cptp C_\emptyset$. These steps are similar to the ones in \cite[Proof of Theorem 3]{FKD18:reducing} and work as can be seen here for the basic UBQC protocol and any protocol based on it.

Let $\cptp C$ and $\rho_C$ be the Client's computation and input, let $T$ and $\sigma$ be the trappified canvas chosen from the trappified scheme $\sch P$ and the associated input. Finally, let $\TP$ be the trappified pattern resulting from embedding $\cptp C$ into $T$, with base angles $\{\phi(i)\}_{i \in O^c}$. 

We start by expressing the state in the simulation and the real protocol. The Server first receives quantum states which are encrypted with $\Z_i(\theta(i))\X_i^{a(i)}$ for all vertices $i \in V$. This is explicitly the case for the inputs to the computation and trap patterns, but also for the other qubits of the graph, since we have that $\ket{+_\theta} = \Z(\theta)\ket{+} = \Z(\theta)\X^a\ket{+}$.\footnote{In the real protocol, the value $a(i)$ is always $0$ for $i \in I^c$. This is perfectly indistinguishable since the distribution of the values of $\delta$ are identical regardless of this choice of parameter for non-input qubits and correctness is unaffected.} Recall that $a_N(i) = \sum_{j \in N_G(i)} a(j)$ and the outputs qubits are only Quantum One-Time-Padded, i.e.~$\theta(i) = (r(i) + a_N(v)) \pi$ for $i \in O$. Then, omitting the Client's classical registers containing the secret values $\bm{\theta},\bm{a}, \bm{r}$, the state from the point of view of the Client is noted $\rho_{in, \cptp C, \bm{b} + \bm{r}}^{\bm{\theta},\bm{a}, \bm{r}, T}$, defined as
\begin{align}
\rho_{in, \cptp C, \bm{b} + \bm{r}}^{\bm{\theta},\bm{a}, \bm{r}, T} = \left(\prod_{i \in V} \Z_i(\theta(i))\X^{a(i)}_i\right) [\rho_C \otimes \sigma \otimes \dyad{+}^{\otimes |V| - |I|}]\bigotimes_{i \in O^c} \dyad{\delta_{\bm{b} + \bm{r}}(i)},
\end{align}
where $\bm{b}$ corresponds to the perceived branch of computation based on the outcomes returned by the Server to the Client. The values $\delta_{\bm{b} + \bm{r}}(i) = (-1)^{a(i)}\phi'_{\bm{b} + \bm{r}}(i) + \theta(i) + (r(i) + a_N(i)) \pi$ are each encoded as computational basis states on three qubits from a register $R$ with $3n$ qubits. The angle $\phi'_{\bm{b} + \bm{r}}(i)$ is obtained through the formula for $\phi'(i)$ from the UBQC Protocol \ref{proto:ubqc}, Equation \ref{eq:updt-ubqc} and includes the corrections stemming from $\bm{b}$ and $\bm{r}$, while $a_N(i)$ compensates the effect of the $\X$ encryption from a qubit on its neighbours. While this seems that the Client is sending the values of $\delta_{\bm{b} + \bm{r}}(i)$ at the beginning breaks the causal structure of the protocol, these states will indeed not be affected by any operations before they can actually be correctly computed by the Client. This will be made formal below. Finally, note that for simplicity, the qubits in the state above are not grouped in the order in which the Client sends.

\textit{State after the Interaction with the Server.}
To describe all possible deviations, we consider a purification $\dyad{\psi_S}$ of $\rho_S$, the Server's private work state. Let $\cptp F^{\mathtt{in}}$ be a unitary operator such that $\ket{\psi_S} = \cptp F^{\mathtt{in}} \ket{0}^{\otimes w}$ for the appropriate work register size $w$. Then the operations which the Server applies before any measurement can be written as unitaries acting on all qubits which have not yet been measured and the available values of $\delta_{\bm{b} + \bm{r}}(i)$. These can be then decomposed into the correct unitary operation followed by a unitary attack of the Server's choice. The Server receives all qubits, applies the entanglement operation corresponding to the Client's desired graph, then a unitary attack $\cptp F^{\mathtt{G}}$, then the correct rotation on the first measured qubits, followed by another attack $\cptp F^{\mathtt{1}}$.\footnote{We use the upper index to avoid confusion with operators such as $\Z_i$, which indicates a $\Z$ applied on qubit $i$ in the graph and  identity on all other qubits} These last two steps -- rotation of a qubit in the graph followed by a deviation on all qubits -- are repeated once per measured qubit.

Recall that $\cptp G = \prod_{(i, j) \in E} \CZ_{i, j}$ is the unitary operation which entangles the qubits according to the graph $G = (V, E)$. The $\Z$-axis rotations required for performing the measurement in the basis defined by $\delta_{\bm{b} + \bm{r}}$ are represented by unitaries $\CR$, controlled rotations around the $\Z$-axis with the control being performed by the registers containing the corresponding value of $\delta_{\bm{b} + \bm{r}}$. Figure~\ref{fig:cr_gate} shows one possible implementation of this controlled operation.

\begin{figure}[ht]
	\centering
	\begin{subfigure}[t]{0.45\textwidth}
		\centering
$
\Qcircuit @C=1em @R=.7em @!R {
\lstick{\rho} & \qw & \gate{\Z(\cdot)} & \qw & \qw \\
\lstick{\ket{\delta}} & {/} \qw & \ctrl{-1} & {/} \qw & \qw
}
$
	\end{subfigure}
	\quad
	\begin{subfigure}[t]{0.45\textwidth}
		\centering
$\Qcircuit @C=1em @R=.7em @!R {
\lstick{\rho} & \qw & \gate{\Z} & \gate{\Z(\frac{\pi}{2})} & \gate{\Z(\frac{\pi}{4})} & \qw \\
\lstick{\ket{\delta^1}} & \qw & \ctrl{-1} & \qw & \qw & \qw \\
\lstick{\ket{\delta^2}} & \qw & \qw & \ctrl{-2} & \qw & \qw \\
\lstick{\ket{\delta^3}} & \qw & \qw & \qw & \ctrl{-3} & \qw
}$
	\end{subfigure}
  \caption{Controlled rotation used to unitarise Protocol~\ref{proto:ubqc}. The right hand side is a possible implementation of the rotation on the left, where $\delta^j$ are the bits composing the value $\delta$. The $3$ controlling qubits are sent by the client to the server in the computational basis as they correspond to classical values.}
	\label{fig:cr_gate}
\end{figure}
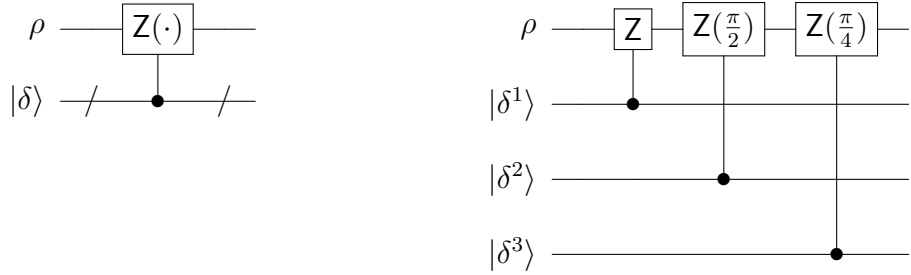

The quantum state representing the interaction between the Client and Server implementing the protocol just before the measurements are performed, noted $\rho_{pre, \cptp C, \bm{b} + \bm{r}}^{\bm{\theta},\bm{a}, \bm{r}, T}$, is thus
\begin{align}
\rho_{pre, \cptp C, \bm{b} + \bm{r}}^{\bm{\theta},\bm{a}, \bm{r}, T} = \cptp F^{\mathtt{n}}\circ\CR_n^\dagger\circ\ldots\circ\cptp F^{\mathtt{1}}\circ\CR_1^\dagger \circ\cptp F^{\mathtt{G}}\circ(\cptp G\otimes \cptp F^{\mathtt{in}})[\rho_{in, \cptp C, \bm{b} + \bm{r}}^{\bm{\theta},\bm{a}, \bm{r}, T}\otimes\dyad{0}^{\otimes w}],
\end{align}
with $|O^c| = n$.
We can move all deviations through the controlled rotations and regroup them as $\cptp F'$.\footnote{Formally, we have $\cptp F' = \cptp F^{\mathtt{n}}\circ\CR_n^\dagger\circ\ldots\circ\cptp F^{\mathtt{1}}\circ\CR_1^\dagger \circ\cptp F^{\mathtt{G}}\circ (\Id_V \otimes \cptp F^{\mathtt{in}}) \circ \prod_{i \in O^c} \CR_i$.} Then, it is possible to replace the (classically) controlled rotations corresponding to the honest execution of the protocol by ordinary rotations $\Z(\delta_{\bm{b} + \bm{r}}(i))^\dagger$, thus yielding\footnote{If these operations were replaced before, the deviations would pick up a dependency on $\delta_{\bm{b} + \bm{r}}(i)$ during the commutation.}
\begin{align}
\rho_{pre, \cptp C, \bm{b} + \bm{r}}^{\bm{\theta},\bm{a}, \bm{r}, T} = \cptp F'\circ \left(\prod_{i \in O^c}\Z_i(\delta_{\bm{b} + \bm{r}}(i))^\dagger\right) \circ (\cptp G \otimes \Id_w)[\rho_{in, \cptp C, \bm{b} + \bm{r}}^{\bm{\theta},\bm{a}, \bm{r}, T}\otimes\dyad{0}^{\otimes w}],
\end{align}
where $\Id_w$ is the identity on the Server's auxiliary register.
We now apply the decryption operations performed by the Client on the output layer qubits after the Server has returned these qubits at the end of the protocol. The resulting state, noted $\rho_{dec, \cptp C, \bm{b} + \bm{r}}^{\bm{\theta},\bm{a}, \bm{r}, T}$, can be written as
\begin{align}
\rho_{dec, \cptp C, \bm{b} + \bm{r}}^{\bm{\theta},\bm{a}, \bm{r}, T} = \left(\prod_{i \in O}\Z_i^{s_Z(i) + r(i)}\X_i^{s_X(i) + a(i)}\right) [\rho_{pre, \cptp C, \bm{b} + \bm{r}}^{\bm{\theta},\bm{a}, \bm{r}, T}],
\end{align}
where $s_X(i)$ and $s_Z(i)$ stem from the flow of the trappified pattern $\TP$. To finish, we enforce that the computation branch is effectively $\bm{b}$ by projecting all non-output qubits $i \in O^c$ using the projection operators $\Z_i^{b(i) + r(i)}\dyad{+}\Z_i^{b(i)}$.\footnote{These qubits can be assumed to be measured without loss of generality since (i) the Server needs to produce the values $\bm{b}$ using its internal state and the values received from the Client and (ii) the operation $\cptp F'$ is fully general, meaning that the Server can use it to reorder the qubits before the measurement if it so desires.}\textsuperscript{,}\footnote{The difference in coefficients takes into account the corrections which the Client applies to the outputs of the measurements to account for $r(i)$. The left-hand side operator represents the view of the Client while the right-hand side is the Server's view.} Since $\ket{+}$ is a $+1$ eigenstate of $\X$, this is equivalent to projecting with $\Z_i^{b(i) + r(i)}\dyad{+}_i\X_i^{a(i)}\Z^{b(i)}$. We define $\cptp P_{\bm{b}}$ to be the projection operator $\prod_{i \in O^c}\Z^{b(i)}_i[\dyad{+}^{\otimes \abs{O^c}}]$, which is applied on the non-output qubits of the graph. The final state, noted $\rho_{out, \cptp C, \bm{b} + \bm{r}}^{\bm{\theta},\bm{a}, \bm{r}, T}$, is therefore
\begin{align}
\rho_{out, \cptp C, \bm{b} + \bm{r}}^{\bm{\theta},\bm{a}, \bm{r}, T} = (\cptp P_{\bm{b} + \bm{r}}\otimes \Id_O\otimes \Id_w) \circ \left(\prod_{i \in O^c} \X^{a(i)}_i\Z^{r(i)}_i\right)[\rho_{dec, \cptp C, \bm{b} + \bm{r}}^{\bm{\theta},\bm{a}, \bm{r}, T}].
\end{align}
We then apply the change of variable $b'(i) = b(i) + r(i)$ and relabel $b'(i)$ into $b(i)$. This has the effect of removing the influence of $r(i)$ in the corrected measurement angles, transforming $\phi'_{\bm{b} +\bm{r}}(i)$ into $\phi'_{\bm{b}}(i)$:\footnote{This value uses the formula for $\phi'(i)$ from the Delegated MBQC Protocol \ref{proto:mbqc}, Equation \ref{eq:updt-mbqc}.}
\begin{align}
\rho_{out, \cptp C, \bm{b}}^{\bm{\theta},\bm{a}, \bm{r}, T} &= (\cptp P_{\bm{b}}\otimes \Id_O\otimes\Id_w) \circ \left(\prod_{i \in O^c} \X^{a(i)}_i\Z^{r(i)}_i\right)[\rho_{dec, \cptp C, \bm{b}}^{\bm{\theta},\bm{a}, \bm{r}, T}]\\
&= (\cptp P_{\bm{b}}\otimes \Id_O\otimes\Id_w)\circ \tilde{\cptp U}_{\TP} \left[\rho_C \otimes \sigma \otimes \dyad{+}^{\otimes |V| - |I|}\bigotimes_{i \in O^c} \dyad{\delta_{\bm{b}}(i)} \otimes\dyad{0}^{\otimes w}\right].
\end{align}
where we defined $\tilde{\cptp U}_{\TP}$ as the unitary part of the deviated execution of the pattern $\TP$:
\begin{align}
\begin{split}
\tilde{\cptp U}_{\TP} = &\left(\prod_{i \in O^c} \X^{a(i)}_i\Z^{r(i)}_i\right) \circ \left(\prod_{i \in O}\Z_i^{s_Z(i) + r(i)}\X_i^{s_X(i) + a(i)}\right) \circ \cptp F' \circ \\
&\left(\prod_{i \in O^c}\Z_i(\delta_{\bm{b}}(i))^\dagger\right) \circ (\cptp G \otimes \Id_w) \circ \left(\prod_{i \in V} \Z_i(\theta(i))\X^{a(i)}_i\right).
\end{split}
\end{align}

\textit{Resulting State from the Distinguisher's Point of View.}
Above, $\rho_{out, \cptp C, \bm{b}}^{\bm{\theta},\bm{a}, \bm{r}, T}$ describes the quantum state when both the secrets chosen by the Client's protocol and the deviation performed by the Server are known. However, for the Distinguisher, the secrets chosen by the Client's protocol are unknown. Therefore, from its point of view, the resulting state is $\rho_{out, \bm{b}}$, obtained by taking the sum over the secret parameters $\bm{\theta}, \bm{a}, \bm{r}$ is given by
\begin{align}
\rho_{out, \cptp C, \bm{b}}^T = \frac{1}{8^{|O^c|} \cdot 4^{|V|}} \sum_{\bm{\theta},\bm{a}, \bm{r}} \rho_{out, \cptp C, \bm{b}}^{\bm{\theta},\bm{a}, \bm{r}, T}.
\end{align}
To work out an explicit form for $\rho_{out, \cptp C, \bm{b}}^T$, we focus on the state before the projection $\cptp P_{\bm{b}}$ is applied.
The goal is to remove dependencies on $r(i), a(i)$ which appear outside the encryption and decryption procedures in order to be able to use the twirling lemma, using the fact that these parameters are chosen at random.\footnote{These parameters must be perfectly random as using them multiple times might introduce correlations which the Server can exploit to derandomise the Pauli twirl.} To this end we cancel out the values of $\theta(i)$ coming from the initial encryption with those which appear in the rotations by $\delta_{\bm{b}}(i) = (-1)^{a(i)}\phi'_{\bm{b}}(i) + \theta(i) + (r(i) + a_N(i)) \pi$ for $i \in O^c$:
\begin{align}
\Z_i(\delta_{\bm{b}}(i))^\dagger \circ\cptp G\circ \Z_i(\theta(i))\X^{a(i)}_i = \Z_i((-1)^{a(i)}\phi'_{\bm{b}}(i) + (r(i) + a_N(i)) \pi)^\dagger\circ \cptp G\circ \X^{a(i)}_i,
\end{align}
due to the fact that the entanglement operation consists of $\CZ$ operations through which the $\Z$ rotations commute. Now, the values $\theta(i)$ appear only in the definition of the classical states $\ket{\delta_{\bm{b}}(i)}$. Hence, they perform a One-Time-Pad encryption to these angles and summing over $\theta(i)$ yields the perfectly mixed state in the register $R$. Formally
\begin{align}
\rho_{out, \cptp C, \bm{b}}^T = \frac{1}{4^{|V|}} \sum_{\bm{a}, \bm{r}} (\cptp P_{\bm{b}}\otimes\Id_O\otimes\Id_w)\circ \tilde{\cptp U}_{\TP} \left[\rho_C \otimes \sigma \otimes \dyad{+}^{\otimes |V| - |I|} \otimes \1_{3n} \otimes\dyad{0}^{\otimes w}\right],
\end{align}
where $\1_{3n}$ is the perfectly mixed state over the $3n$ qubits of $R$. This register has thus no effect on either the computation or the traps and is in tensor product with the rest of the state, it can therefore be traced out by assuming without loss of generality that the Server's deviation has no effect on it. We will keep the previous labels for the states after this trace operation is performed to reduce complexity.

We can now commute the encryption on both sides of the deviation so that the deviation is exactly sandwiched between two identical random Pauli operations. We start on the right side of $\cptp F'$ in the expression of $\tilde{\cptp U}_{\TP}$. For all qubits in the graph, we need to commute all $\X^{a(i)}_i$ through the entanglement operation first. Since $\CZ_{i, j}\X_i = \Z_j\X_i\CZ_{i, j}$ (and similarly for $\X_j$), using $a_N(i) = \sum_{j \in N_G(i)} a(j)$ we get that
\begin{align}
\cptp G \circ \left(\prod_{i \in V} \X_i^{a(i)}\right) = \left(\prod_{i \in V} \Z_i^{a_N(i)}\X_i^{a(i)}\right) \circ\cptp G.
\end{align}
The additional $\Z_i^{r(i) + a_N(i)}$ encryption of the output qubits commute unchanged through the entanglement operation $\cptp G$. These encryptions now need to be commuted through the $\Z$ rotations for measured qubits:\footnote{$\X_i$ and $\Z_i$ commute trivially through the rotations on qubits $j \neq i$.}
\begin{align}
\Z_i((-1)^{a(i)}\phi'_{\bm{b}}(i) + (r(i) + a_N(i)) \pi)^\dagger\Z_i^{a_N(i)}\X_i^{a(i)} = \Z_i^{r(i)}\X_i^{a(i)}\Z_i(\phi'_{\bm{b}}(i))^\dagger.
\end{align}
On the other hand, on the output qubits, the operation applied is $\Z_i^{a_N(i)}\Z_i^{r(i) + a_N(i)}\X_i^{a(i)} = \Z_i^{r(i)}\X_i^{a(i)}$. In total, we have that
\begin{align}
\left(\prod_{i \in O^c}\Z_i(\delta_{\bm{b}}(i))^\dagger\right)\circ \cptp G\circ \left(\prod_{i \in V} \Z_i(\theta(i))\X^{a(i)}_i\right) = \cptp Q_{\bm{a}, \bm{r}}\circ\left(\prod_{i \in O^c}\Z_i(\phi'_{\bm{b}}(i))^\dagger\right)\circ\cptp G,
\end{align}
where $\cptp Q_{\bm{a}, \bm{r}} = \prod_{i \in V} \Z_i^{r(i)}\X^{a(i)}_i$. On the other side of $\cptp F'$ in the expression of $\tilde{\cptp U}_{\TP}$, we simply have that
\begin{align}
\left(\prod_{i \in O^c} \X^{a(i)}_i\Z^{r(i)}_i\right)\circ \left(\prod_{i \in O}\Z_i^{s_Z(i) + r(i)}\X_i^{s_X(i) + a(i)}\right) = \left(\prod_{i \in O}\Z_i^{s_Z(i)}\X_i^{s_X(i)}\right)\circ \cptp Q_{\bm{a}, \bm{r}}^\dagger,
\end{align}
up to a global phase.

We note that $\rho_{cor, \cptp C, \bm{b}}^T = \left(\prod_{i \in O^c}\Z_i(\phi'_{\bm{b}}(i))^\dagger\right)\circ\cptp G [\rho_C \otimes \sigma \otimes \dyad{+}^{\otimes |V| - |I|}]$ is the correct state during an execution of $\TP$ before the encryption-deviation-decryption, and define $\cptp D_{\bm{b}} = \prod_{i \in O}\Z_i^{s_Z(i)}\X_i^{s_X(i)}$ as the final plain MBQC correction operator. We then obtain
\begin{align}
\rho_{out, \cptp C, \bm{b}}^T = \frac{1}{4^{|V|}} (\cptp P_{\bm{b}}\otimes \cptp D_{\bm{b}}\otimes\Id_w)\left[ \sum_{\cptp Q_{\bm{a}, \bm{r}} \in \mathcal P_V} (\cptp Q_{\bm{a}, \bm{r}}^\dagger \otimes \Id_w)\circ \cptp F'\circ (\cptp Q_{\bm{a}, \bm{r}}  \otimes \Id_w) [\rho_{cor, \cptp C, \bm{b}}^T\otimes\dyad{0}^{\otimes w}]\right].
\end{align}

Without loss of generality we can decompose the Server's deviation in the Pauli operator basis over the graph's vertices as $\cptp F' = \sum_{\cptp E \in \mathcal P_V} \alpha_E \cptp E \otimes \cptp U_{\cptp E}$. Applying the notation $\cptp U[\rho] = \cptp U\rho\cptp U^\dagger$, we get
\begin{align}
\rho_{out, \cptp C, \bm{b}}^T = \frac{1}{4^{|V|}} (\cptp P_{\bm{b}}\otimes \cptp D_{\bm{b}}\otimes\Id_w)\left[ \smashoperator[r]{\sum_{\cptp E, \cptp E' \in \mathcal P_V}} \alpha_{\cptp E}\alpha_{\cptp E'}^* \smashoperator[r]{\sum_{\cptp Q_{\bm{a}, \bm{r}} \in \mathcal P_V}}\cptp Q_{\bm{a}, \bm{r}}^\dagger \cptp E \cptp Q_{\bm{a}, \bm{r}} \rho_{cor, \cptp C, \bm{b}}^T \cptp Q_{\bm{a}, \bm{r}}^\dagger \cptp E'^\dagger \cptp Q_{\bm{a}, \bm{r}} \otimes \cptp U_{\cptp E}\dyad{0}^{\otimes w}\cptp U_{\cptp E'}^\dagger \right],
\end{align}
where $\alpha_{\cptp E'}^*$ is the complex conjugate of $\alpha_{\cptp E'}$. We now apply the Twirling Lemma \ref{lem:twirl}, leading to
\begin{equation}
\sum_{\cptp Q_{\bm{a}, \bm{r}} \in \mathcal P_V}\cptp Q_{\bm{a}, \bm{r}}^\dagger \cptp E \cptp Q_{\bm{a}, \bm{r}} \rho_{cor, \cptp C, \bm{b}}^T \cptp Q_{\bm{a}, \bm{r}}^\dagger \cptp E'^\dagger \cptp Q_{\bm{a}, \bm{r}} = 0,
\end{equation}
for $\cptp E \neq \cptp E'$. Therefore
\begin{align}
\rho_{out, \cptp C, \bm{b}}^T = \frac{1}{4^{|V|}} (\cptp P_{\bm{b}}\otimes \cptp D_{\bm{b}}\otimes\Id_w)\left[ \sum_{\cptp Q_{\bm{a}, \bm{r}}, \cptp E \in \mathcal P_V} |\alpha_{\cptp E}|^2 \cptp Q_{\bm{a}, \bm{r}}^\dagger \cptp E \cptp Q_{\bm{a}, \bm{r}} \rho_{cor, \cptp C, \bm{b}}^T \cptp Q_{\bm{a}, \bm{r}}^\dagger \cptp E^\dagger \cptp Q_{\bm{a}, \bm{r}} \otimes \cptp U_{\cptp E}\dyad{0}^{\otimes w}\cptp U_{\cptp E}^\dagger \right],
\end{align}
Consequently, the deviation is a CPTP map defined by $\{\cptp E \otimes \cptp U_{\cptp E}, p_{\cptp E} = |\alpha_{\cptp E}|^2\}_{\cptp E \in \mathcal P_V}$, a convex combination of Pauli operators on the graph's vertices tensored with an operation on the Server's internal register.
Overall, this shows that the effect of the Server's deviation -- when averaged over the choice of secret parameters $\bm{\theta}, \bm{a}, \bm {r}$ -- is a probabilistic mixture of Pauli operators on the qubits of the graph.

The Pauli encryption and decryption $\cptp Q_{\bm{a}, \bm{r}}$ commutes up to a global phase with the Pauli deviation $\cptp E$. We can therefore rewrite the state as
\begin{align}
\rho_{out, \cptp C, \bm{b}}^T = \sum_{\cptp E \in \mathcal P_V} p_{\cptp E} (\cptp P_{\bm{b}}\otimes \cptp D_{\bm{b}})\circ\cptp E [\rho_{cor, \cptp C, \bm{b}}^T] \otimes \cptp U_{\cptp E}[\dyad{0}^{\otimes w}].
\end{align}

Since the distinguisher wishes to maximise its distinguishing probability, it is sufficient to consider that it applies a fixed Pauli deviation $\cptp E \in \mathcal P_V$ for which the distinguishing probability is maximal. Furthermore, the state in the Server's register is unentangled from the rest and therefore does not contribute to the attack of the Server on the Client's state. Once this is traced out, seeing as $\cptp D_b$ and $\cptp E$ are Paulis and therefore commute up to a global phase, the final state can be written as
\begin{align}
\rho_{out, \cptp C, \bm{b}}^T = (\cptp P_{\bm{b}}\otimes \Id_O)\circ \cptp E \circ (\Id_{O^c}\otimes\cptp D_{\bm{b}}) [\rho_{cor, \cptp C, \bm{b}}^T] = \cptp E \circ (\TP)[\rho_C \otimes \sigma].
\end{align}
The final equality stems from the definition of the notation $\cptp E \circ P$ for a pattern $P$ (Section \ref{subsec:detect-def}) and the fact that applying $\cptp D_{\bm{b}}$ to $\rho_{cor, \cptp C, \bm{b}}^T$ performs exactly the correct unitary portion of plain MBQC pattern $\TP$ -- up to the measurements which are handled by $\cptp P_{\bm{b}}$.\footnote{This is correct up to a relabelling of $\mathcal P_V$ since in the rest of the paper we assumed that the measurements are performed in the computational basis.}

\textit{Applying the Composable Security of UBQC.}
We next show that this deviation depends on the same classical parameters in the ideal and real scenarii. To this end, we apply the composition Theorem \ref{thm:ac-compos} of the AC framework to replace the execution of the UBQC Protocol by the Blind Delegated Quantum Computation Resource \ref{def:bqc_if} both in the simulation and the real protocol. As per the security of the UBQC Protocol as expressed in Theorem \ref{thm:sec-ubqc}, the distinguishing advantage is not modified by this substitution so long as the graph, order of measurements and output set of qubits are known to the Server. The results can be seen in Figures \ref{fig:real-ubqc} and \ref{fig:sim-ubqc}. The distinguisher has access to all outward interfaces.

\begin{figure}[ht]
\centering
\includegraphics[width=0.60\textwidth]{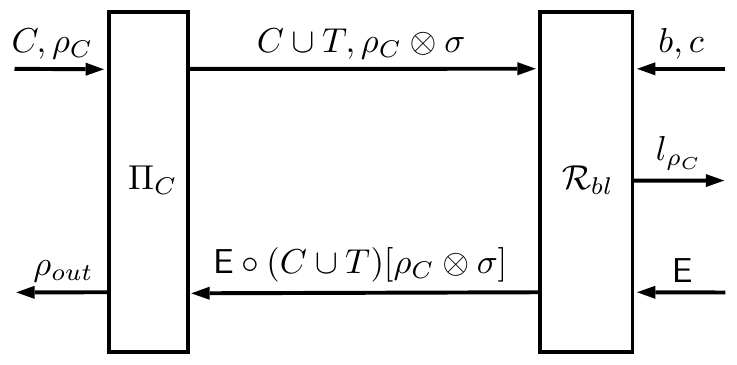}
\caption{Real world hybrid interaction between the Client's protocol CPTP map $\Pi_C$ and Blind Delegated QC Resource $\mathcal{R}_{bl}$.}
\label{fig:real-ubqc}
\end{figure}

\begin{figure}[ht]
\centering
\includegraphics[width=0.80\textwidth]{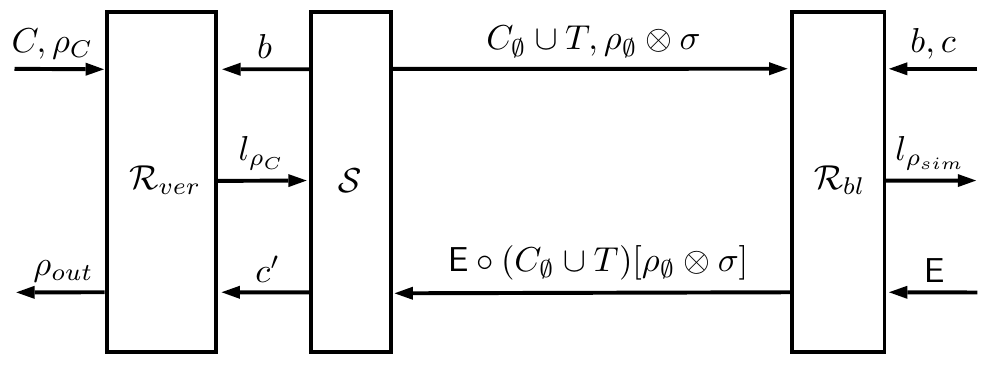}
\caption{Simulator $\mathcal{S}$ interacting Secure and Blind DQC Resources $\mathcal{R}_{ver}$ and $\mathcal{R}_{bl}$.}
\label{fig:sim-ubqc}
\end{figure}

The Simulator receives the leak $l_{\rho_C} = (\mathfrak{C}, G, \sch P, \preceq_G)$ from the Secure Delegated Quantum Computation Resource. In both cases, we assume that both the Client and Simulator send $(\mathfrak{C}, G, \sch P, \preceq_G)$ as a first message to the Server. All canvases in $\sch P$ are blind-compatible (Definition \ref{def:compat}) meaning that they all share the graph $G$ and the same output set $O$, and the order $\preceq_G$ is used for all patterns generated from $\sch P$ and the embedding algorithm. Since these parameters are the same in all executions of both the real and ideal case, the leak $l_{\rho_{ideal}}$ obtained by the Server in the simulation does not yield any more information. Overall, the classical information in both cases is identical and does not help the distinguisher on its own. The deviations in the real and ideal case follow therefore the same distribution as well, independently of the computation and trappified canvas chosen by the Client or Simulator.

\textit{Output and Abort Probability Analysis.}
The interactions are therefore indistinguishable before the output is sent back to the Client and we focus in the following on the output state and the abort probability in both cases, which are the only remaining elements which the distinguisher can use to decide which setup it is interacting with. 
At the moment when the distinguisher chooses the Pauli deviation $\cptp E \in \mathcal P_V$, it does not have access to the trappified canvas $T$ that the Client or Simulator have sampled. Therefore we can average the final output state over the possible choices of $T$ as follows
\begin{align}
\rho_{out, \cptp C, \bm{b}} = \sum_{T \in \sch P} \Pr_{T \sim \pd P}[T] \cptp E \circ (\TP)[\rho_C \otimes \sigma].
\end{align}
After this, the Client and Simulator decide whether the computation should be accepted or rejected and in the real case the Client applies the decoding algorithm $\Deco$ to the output of the computation. The final outputs $\rho_{\substack{out\\ real}}, \rho_{\substack{out\\ ideal}}$ of the Client in respectively the real and ideal settings can then be written as follows
\begin{align}
\rho_{\substack{out\\ real}} &= \sum_{T \in \sch P}\Pr_{T \sim \pd P}[T] \left(p_{\acc, \TP}\Cdev{\cptp E} [\rho_C \otimes \sigma]\otimes\dyad{\acc} + \left(1 - p_{\acc, \TP}\right)\dyad{\bot}\otimes\dyad{\rej}\right),\\
\rho_{\substack{out\\ ideal}} &= \sum_{T \in \sch P}\Pr_{T \sim \pd P}[T] \left(p_{\acc, C_\emptyset \cup T}\cptp C [\rho_C]\otimes\dyad{\acc} + \left(1 - p_{\acc, C_\emptyset \cup T}\right)\dyad{\bot}\otimes\dyad{\rej}\right),
\end{align}
where $\Cdev{\cptp E} = \Deco\circ\Tr_{O_C^c}\circ\cptp E\circ (\TP)$. For all $\cptp C \in \mathfrak{C}$ we have that
\begin{align}
p_{\acc, \TP} = \Pr_{t\sim \pd T_{\cptp E}}[\tau(t) = 0],
\end{align}
which uses Lemma \ref{lem:indep-t}, namely that the acceptance probability does not depend on the input or the computation. Therefore this probability is identical in the real and ideal setting $p_{\acc, \TP} = p_{\acc, C_\emptyset \cup T} = p_{\acc, T}$, regardless of the deviation chosen by the distinguisher. We see that whenever the computation is rejected, the output state is identical in both setups. On the other hand, whenever the computation is accepted, the ideal resource will always output the correct state, while the concrete protocol outputs a potentially erroneous state. By convexity of the trace distance, the distinguishing probability $p_d$ can therefore be written as:
\begin{align}
p_d = \max_{\substack{\cptp E \in \mathcal P_V\\ \cptp C \in \mathfrak{C} \\ \psi_C}}\left(\sum_{T\in \sch P} \Pr_{\substack{T \sim \pd P \\ t\sim \pd T_{\cptp E}}}[\tau(t) = 0, T] \times \|(\Cdev{\cptp E}  - \cptp C\otimes \Id_T)\otimes\Id_D [\dyad{\psi_C} \otimes \sigma]\|_{\Tr}\right),
\end{align}
where $\ket{\psi_C}$ is a purification of the Client's input $\rho_C$ using the distinguisher's register $D$. We now therefore analyse the output state in the case where the computation is accepted.

\textit{Deviation Influence on Distinguishing Probability.}
First consider the case where $\cptp E \in \mathcal E_\epsilon$. 
Since $\sch P$ $\epsilon$-detects such errors (Definition \ref{def:detect_tp}), the probability of accepting $\Pr_{t\sim \pd T_{\cptp E}}[\tau(t) = 0]$ is upper-bounded by $\epsilon$, which implies
\begin{align}
p_{d, \mathcal E_\epsilon} \leq \epsilon \times \max_{\substack{\cptp E \in \mathcal E_\epsilon\\ \cptp C \in \mathfrak{C} \\ \psi_C}}\left(\sum_{T\in \sch P} \Pr_{T \sim \pd P}[T]\|(\Cdev{\cptp E}  - \cptp C\otimes \Id_T)\otimes\Id_D [\dyad{\psi_C} \otimes \sigma]\|_{\Tr}\right) = \hat{\epsilon}.
\end{align}
The distinguisher can freely choose the Client's input state $\psi_C$ and the computation $\cptp C \in \mathfrak{C}$ and there is no constraint on the effect of this deviation on the computation part of the trappified pattern. In the worst case the incorrect real output state is orthogonal to the ideal output state, meaning that the distinguisher can tell apart both settings with certainty and the trace distance is upper-bounded by $1$. The distinguishing probability in this scenario therefore follows $p_{d, \mathcal E_\epsilon} \leq \epsilon$.

Second, we consider the alternate case, where $\cptp E \notin \mathcal E_\epsilon$. Here, we assumed that the trappified scheme $\sch P$ is $\nu$-correct on the set $\mathcal E_\nu$ of which $\mathcal P_V \setminus \mathcal E_\epsilon$ is a subset (Definition \ref{def:trap-cor}), therefore the trace distance between the correct result of the computation and the real output of the protocol is upper-bounded by $\nu$
\begin{align}
\|(\Cdev{\cptp E}  - \cptp C)\otimes\Id_D [\dyad{\psi_C} \otimes \sigma]\|_{\Tr} \leq \nu.
\end{align}
Therefore
\begin{align}
p_{d, \mathcal P_V \setminus \mathcal E_\epsilon} \leq \nu  \times \max_{\cptp E \in \mathcal P_V \setminus \mathcal E_\epsilon}(\sum_{T\in \sch P} \Pr_{\substack{T \sim \pd P \\ t\sim \pd T_{\cptp E}}}[\tau(t) = 0, T]),
\end{align}
where the maximisation is done only over the error since the acceptance probability is independent of the input and computation. In this case, the accepting probability $p_{\acc, T}$ is not constrained and hence only upper bounded by $1$, yielding $p_{d, \mathcal P_V \setminus \mathcal E_\epsilon} \leq \nu$.

Since the deviation chosen by the distinguisher falls in either of these two cases, we have $p_d = \max(p_{d, \mathcal E_\epsilon}, p_{d, \mathcal P_V \setminus \mathcal E_\epsilon})$ and the maximum distinguishing probability between the Resource together with the Simulator and the concrete Protocol is thus upper-bounded by $\max(\hat{\epsilon}, \nu) \leq \max(\epsilon, \nu)$.
\end{proof}

\begin{remark}[Using Other Blind Protocols.]
In this work we use the UBQC protocol to provide blindness. This protocol is based on the \emph{prepare-and-send} principle. 
The direct mirror situation, where the Server prepares states and sends them to the Client, is called the \emph{receive-and-measure} paradigm. These are also based on MBQC and were shown to be equivalent to prepare-and-send protocol by \cite{WEP22:equivalence} using the Abstract Cryptography framework. Our techniques are therefore directly applicable to this setting as well with the same security guarantees. These two setups together cover most protocols that have been designed and which may be implemented in the near future.

The work of~\cite{M18:classical} introduced an explicit protocol for verifying $\mathsf{BQP}$ computations by relying only on classical interactions and a computational hardness assumption.
Our techniques are fully applicable as well using a protocol which \emph{$\epsilon_{bl}$-computationally-constructs} the Blind Delegated Quantum Computation Resource \ref{def:bqc_if} in the AC framework and is capable of implementing MBQC computations natively. The resulting protocol is of course computationally-secure only. A simple hybrid argument can be used first to replace any such computationally-secure protocol with Resource \ref{def:bqc_if} first -- at a cost of $\epsilon_{bl}$ -- and then the UBQC protocol at no cost. The other steps of the proof remain unchanged.
\end{remark}

\subsection{Insensitivity Implies Noise-Robustness}
Then, we give conditions on protocols implementing SDQC so that they are able to run on noisy machines with a good acceptance probability. We show formally the following intuitive reasoning: if the errors to which the trappified scheme is insensitive do not disturb the computation too much, then a machine which mostly suffers from such errors will almost always lead to the client accepting the computation and the output will be close to perfect.

\begin{theorem}[Robust Detection Implies Robust Verifiability]
\label{thm:robust_verif}
  Let $\sch P$ be a trappified scheme which is $\delta$-insensitive to the set of Pauli deviations $\mathcal{E}_\delta$.
  We assume an execution of Protocol \ref{proto:dev_detect} with an honest-but-noisy Server whose noise is modelled by sampling an error $\cptp E \in \mathcal E_\delta$ with probability $(1-p_\delta)$ and $\cptp E \in \mathcal P_V \setminus \mathcal E_\delta$ with probability $p_\delta$. 
  Then, the Client in Protocol~\ref{proto:dev_detect} accepts with probability at least $(1-p_\delta)(1-\delta)$.
  If furthermore we have $\mathcal{E}_\delta \subset \mathcal{E}_\nu$, then the correctness error of Protocol~\ref{proto:dev_detect} is $p_\delta + \delta + \nu$.
\end{theorem}

\begin{proof}
By construction, $\sch P$ is $\delta$-insensitive to $\mathcal E_\delta$. Hence, it will accept deviations in $\mathcal E_\delta$ with probability at least $1-\delta$ which yields the overall lower bound on the acceptance probability of $(1 - p_\delta) (1 - \delta)$.

The proof of correctness from Theorem \ref{thm:verif} can be directly updated to account this honest-but-noisy Server if $\mathcal{E}_\delta \subset \mathcal{E}_\nu$ and $\Id \in \mathcal{E}_\delta$. Using the same reasoning as for $\Id$, we can upper-bound the rejection probability in Equation~\ref{eq:cor} by $p_\delta + \delta$ and therefore the overall correctness error by $p_\delta + \delta + \nu$.
\end{proof}

The theorem above shows the importance not only of the parameters of the scheme, but also the size of the sets $\mathcal{E}_\epsilon$, $\mathcal{E}_\delta$ and $\mathcal{E}_\nu$. By creating schemes which have more errors fall in set $\mathcal{E}_\delta$, it is possible to have a direct impact both in terms of acceptance probability and fidelity in the context of honest-but-noisy executions. The next section will show how all these parameters can indeed be amplified in such a way that (i) the noise process generates deviations that are within $\mathcal E_\delta$ with overwhelming probability, (ii) the embedding of the computation $\cptp C$ within $\sch P$ adds redundancy in such a way that $\nu$ is negligible, and (iii) $\sch P$ is $\delta$-insensitive to $\mathcal E_\delta$ for a negligible $\delta$ and detects $\mathcal E_\epsilon$ with negligible error $\epsilon$. In such situation, the protocol will accept the computation almost all the time, while its security error given by Theorem~\ref{thm:verif} will be negligible.

\subsection{Efficient Verifiability Requires Error-Correction.}
We now present an important consequence of Theorem~\ref{thm:verif} in the case where the correctness error $(\delta+\nu)$ and the security error $\max(\epsilon,\nu)$ are negligible with respect to a security parameter $\lambda$.
We show that this correctness and security regime can only be achieved with a polynomial qubit overhead if the computation is error-protected.

More precisely, we denote $\sch P(\lambda)$ a sequence of trappified schemes indexed by a security parameter $\lambda$, such that it $\epsilon(\lambda)$-detects a set $\mathcal E_\epsilon(\lambda) \subseteq \mathcal P_V(\lambda)$ of Pauli deviations, is $\delta(\lambda)$-insensitive to $\mathcal E_\delta(\lambda)$ and is $\nu(\lambda)$-correct on $\mathcal E_\nu(\lambda)$ which includes the complement of $\mathcal E_\epsilon(\lambda)$, for $\epsilon(\lambda)$, $\delta(\lambda)$ and $\nu(\lambda)$ negligible in $\lambda$.
Additionally, let $C$ be a computation pattern which implements the client's desired computation CPTP map $\cptp C \in \mathfrak{C}$ on some input state $\ket\psi$.

We are now interested in the server's memory overhead introduced by implementing $\cptp C$ using $\sch P(\lambda)$ for computation class $\mathfrak{C}$ instead of the unprotected pattern $C$.
This is expressed by the ratio $|G_{\sch P(\lambda)}|/|G_{C}|$ between the number of vertices in the graph $G_{\sch P(\lambda)}$ common to all canvases in $\sch P(\lambda)$ and the graph $G_{C}$ used by the pattern $C$.

For a trappified pattern $\TP$ obtained by using the embedding algorithm on a trappified canvas from $\sch P(\lambda)$ we denote by $|O_{\TP}|$ the number of computation output qubits in $\TP$.
Similarly, $|O_{C}|$ is the number of output qubits in $C$.
Without loss of generality, we impose that $|O_{C}|$ is minimal, in the sense that given the set of possible inputs and $\cptp C$, the space spanned by all possible outputs is the whole Hilbert space of dimension $2^{|O_{C}|}$.
This is always possible as one can add a compression phase at the end of any non-minimal pattern.

\begin{theorem}[Error-Correction Prevents Resource Blow-up]\label{thm:encoding}
Let $C$ be a minimal MBQC pattern implementing a CPTP map $\cptp C$.
Let $\TP$ denote a trappified pattern implementing $\cptp C$ obtained from $\sch P(\lambda)$.
Further assume that Protocol \ref{proto:dev_detect} using $\sch P(\lambda)$ has negligible security error $\max(\epsilon, \nu)$ with respect to $\lambda$.

If $|O_{\TP}|/|O_{C}| = 1$ for a non-negligible fraction of trappified canvases $T \in \sch P(\lambda)$, then the overhead $|G_{\sch P(\lambda)}|/|G_{C}|$ is super-polynomial in $\lambda$.
\end{theorem}

The usefulness of this theorem comes from the contra-positive statement.
Achieving exponential verifiability with a polynomial overhead imposes that $|O_{\TP}|/|O_{C}| > 1$ for an overwhelming fraction of the trappified patterns.
This means that the computation is at least partially encoded into a larger physical Hilbert space,
which then serves to actively perform some form of error-correction.

\begin{proof}
Consider a trappified pattern $\TP$ for computing $\cptp C$ obtained from $\sch P(\lambda)$ such that $|O_{\TP}| = |O_{C}|$.
Given $\preceq_{G_{\sch P(\lambda)}}$, let $o_{\TP} \in O_{\TP}$ be the first output position of the computation.
By definition, a bit-flip operation applied on position $o_{\TP}$ cannot be detected by the trap in $\TP$ since the outcome of the trap is independent of the computation.
Yet, because $C$ is minimal and $|O_{\TP}| = |O_{C}|$, we get that for some input states, the bit-flip deviation on $o_{\TP}$ si harmful.
As a consequence, there exists a $\lambda_0$ such that, for all $\lambda \geq \lambda_0$, the diamond distance between $\cptp C$ and the bit-flipped version is greater than $\nu(\lambda)$.
To obtain exponential verification it is therefore necessary for this bit flip to be in the set of $\epsilon$-detected deviations.
This means that deviating on this position without being detected can happen for at most a negligible fraction $\eta(\lambda)$ of the trappified canvases in $\sch P(\lambda)$.
In other words, the position $o_{\TP}$ can only be the first output computation qubit for a negligible fraction $\eta(\lambda)$ of trappified patterns in $\sch P(\lambda)$ that satisfy $|O_{\TP}| = |O_{C}|$.

Then define $\tilde {\sch P}(\lambda) = \{P = E_{\mathfrak C}(\cptp C,\sch P(\lambda)), \ |O_{\TP}| = |O_{C}|\}$ as the set of trappified patterns for $\mathcal C$ that have no overhead,
and $O = \{o_{\TP}, \ T \in \tilde{\sch P}(\lambda)\}$ the set of vertices corresponding to their first output location.
By construction, we have $\sum_{o \in O} |\{T \in \tilde{\sch P}, \ o_{\TP} = o\}| = |\tilde{\sch P}|$.
But, we just showed that $|\{T \in \tilde{\sch P}, \ o_{\TP} = o\}|/|\sch P(\lambda)|$ is upper-bounded by $\eta$, negligible in $\lambda$.
Thus, $|O|$ is lower-bounded by $|\tilde{\sch P}(\lambda)|/ (|\sch P(\lambda)|\eta)$ which is super-polynomial in $\lambda$ so long as $|\tilde{\sch P}(\lambda)|/|\sch P(\lambda)|)$ is not negligible in $\lambda$.
\end{proof}

Note that the situation where $|O_{\TP}| > |O_{C}|$ is interesting only if the bit-flip deviation on qubit $o_{\TP}$ does not alter the computation.
Otherwise, the same reasoning as above is still applicable.
This shows that enlarging the physical Hilbert space storing the output of the computation is useful only if it allows for some error-correction which transforms low-weight harmful errors into harmless ones.

\section{Correctness and Security Amplification for Classical Input-Output Computations}
\label{sec:rvbqc}
We now construct a generic compiler to boost the properties of trappified schemes in the case of classical inputs. This compiler is a direct application of the results from the previous section regarding the requirement of error-correction since it uses a classical repetition code to protect the computation from low-weight bit-flips. It works by decreasing the set of errors which are detected and increasing the set of errors to which the trappified scheme is insensitive. These errors then can be corrected via a recombination procedure, which in the classical case can be as simple as a majority vote.

\subsection{Classical Input Trappified Scheme Compiler}
Theorem~\ref{thm:verif} presents a clear objective for traps: they should (i) detect harmful deviations while being insensitive to harmless ones.
Yet, a trap in a trappified pattern cannot detect deviations happening on the computation part of the pattern itself.
To achieve exponential verifiability, one further needs to ensure that there are sufficiently many trappified patterns so that it is unlikely that a potentially harmful deviation hits only the computation part of the pattern, and that it is detected with high probability when it hits the rest.
This is best stated by Theorem~\ref{thm:encoding}, which imposes to (ii) error-protect the computation so that hard-to-detect deviations are harmless while remaining harmful errors are easy to detect.
Additionally, one further needs to (iii) find a systematic way to insert traps alongside computation patterns to generate these exponentially many trappified patterns.

Ideally, we would like to be able to design and analyse points (i), (ii) and (iii) independently from one another as much as possible. We show here a general way of performing this decomposition given slight constraints on the client's desired computation.

It is based on the realisation that if the client has $d$ copies of its inputs -- which is always possible whenever the inputs are classical -- it can run $d$ times its desired computation by repeating $d$ times the desired pattern $C$ on graph $G$ sequentially or in parallel.
If the output is classical, it is then naturally protected by a repetition code of length $d$ and the result of the computation can be obtained through a majority vote.
These $d$ executions are called \emph{computation rounds}.
To detect deviations, the client needs to run $s$ additional rounds which contain only traps.
More precisely, each of these \emph{test rounds} is a pattern run on the same graph $G$ so that it is blind-compatible with $C$ (see Definition~\ref{def:compat}).
The collections of these $s$ test rounds themselves constitute trappified canvases according to Definition~\ref{def:trap-c}, where acceptance is conditioned to less than $w$ test rounds failures.
Now, because computation rounds and test rounds are executed using blind-compatible patterns on the graph $G$, the trap insertion (iii) can be achieved by interleaving at random the $s$ test rounds with the $d$ computation rounds.

These steps, which are a generalisation of the technique from \cite{LMKO21:verifying}, are formalised in the following definition. We denote $\mathsf{MBQC}_{G, \preceq}$ the class of computations with classical inputs that can be evaluated by an MBQC pattern on graph $G$ using an order compatible with $\preceq$.

\begin{definition}[Amplified Trap Compiler]\label{def:rvbqc_compiler}
Let $\sch P$ be a trappified scheme on a graph $G = (V, E)$ with an order $\preceq_{\sch P}$, and let $d,s\in \mathbb{N}$, $n = d + s$ and $w \in [s]$. We define the Amplified Trap Compiler that turns $\sch P$ into a trappified scheme $\sch P'$ on $G^n$ for computation class $\mathsf{MBQC}_{G, \preceq_{\sch P}}$ as follows:

\begin{itemize}
	\item The trappified canvases $T' \in \sch P'$ and their distribution is given by the following sampling procedure:
	\begin{enumerate}
		\item Randomly choose a set $S \subset [n]$ of size $s$. These will be the test rounds.
		\item For each $j \in S$, independently sample a trappified canvas $T_j$ from the distribution of $\sch P$.
	\end{enumerate}
	\item For each trappified canvas $T'$ defined above and an output $t = (t_{j})_{j \in S}$, the output of the decision function $\tau'$ is obtained by thresholding over the outputs of the decision functions $\tau_j$ of individual trappified canvases. More precisely
	\begin{equation}
		\tau'(t) = 
0 \text{ if } \sum_{j \in S} \tau_{j}(t_j) < w, \ \text{and } 1 \text{ otherwise}.
\end{equation}
	\item The partial ordering of vertices of $G^n$ in $\sch P'$ is given by the ordering $\preceq_G$ on each copy of $G$.
	\item Let $\cptp C \in \mathfrak{C}$ and $C$ the pattern on $G$ which implements the computation $\cptp C$. Given a trappified canvas $T' \in \sch P'$, the embedding algorithm $E_\mathfrak{C}$ sets to $C$ the pattern of the $d$ graphs that are not in $S$.
\end{itemize}
\end{definition}

\subsection{Boosting Detection and Insensitivity}
The next theorem relates the parameters $d, s, w$ with the deviation detection capability of the test rounds, thus showing that not only (i), (ii) and (iii) can be designed separately, but also analysed separately with regards to the security achieved by the protocol.

Let $\mathcal F \subseteq \mathcal{P}_V$. For $\cptp E \in  \mathcal{P}_{V^n}$, we define $\operatorname{wt}_{\mathcal F}(\cptp E)$ as the number of copies of $G$ on which $\cptp E$ acts with an element of $\mathcal F$. 
For $\mathcal E \subset \mathcal{P}_V$, we define $\mathcal{M}_{\geq k}(\mathcal E) = \{ \cptp E \in \mathcal{P}_{V^n} \; | \; \operatorname{wt}_{\mathcal E} (\cptp E) \geq k \}$, 
and $\mathcal{H}_{\geq k}(\mathcal E) = \{ \cptp E \in (\mathcal E \cup \Id)^n \; | \; \operatorname{wt}_{\mathcal E} (\cptp E) \geq k\}$, and similarly for $\leq, <, >$.

\begin{theorem}[From Constant to Exponential Detection and Insensitivity Rates]\label{thm:boost}
  Let $\sch P$ be a trappified scheme on graph $G$ which $\epsilon$-detects the error set $\mathcal{E}_\epsilon$, is $\delta$-insensitive to $\mathcal{E}_\delta$ and perfectly insensitive to at least $\{\Id\}$.
  For $d, s \in \mathbb{N}$, $n = d + s$ and $w \in [s]$, let $\sch P'$ be the trappified scheme resulting from the compilation defined in Definition~\ref{def:rvbqc_compiler}.
  
   Let $k_\epsilon > nw/(s(1-\epsilon))$ and $k_\delta < nw/(s\delta)$. Then, $\sch P'$ $\epsilon'$-detects $\mathcal{E}_\epsilon' = \mathcal{M}_{\geq k_\epsilon}(\mathcal E_\epsilon)$ and is $\delta'$-insensitive to $\mathcal{E}_\delta' = \mathcal{H}_{< k_\delta}(\mathcal E_\delta)$ where
  \begin{align}
  	\epsilon' &= \min_{\chi \in \left[ 0, \frac{k_\epsilon}{n} - \frac{w}{s(1-\epsilon)} \right]} \exp \left( -2 \chi^2 s \right) + \exp \left( -2 \frac{\left(\left( \frac{k_\epsilon}{n} - \chi \right)(1-\epsilon) - \frac{w}{s}\right)^2}{\left( \frac{k_\epsilon}{n} - \chi \right)}s \right),\\
  	\delta' &= \min_{\chi \in \left[ 0, \frac{w}{s\delta} - \frac{k_\delta}{n} \right]} \exp \left( -2 \chi^2 s \right) + \exp \left( -2 \frac{\left(\left( \frac{k_\delta}{n} + \chi \right)\delta - \frac{w}{s}\right)^2}{\left( \frac{k_\delta}{n} + \chi \right)}s \right).
  \end{align}
\end{theorem}

\begin{proof}
  
	For a given deviation $\cptp E \in \mathcal P_{V^n}$, let $X$ be a random variable describing the number of test rounds on which the deviation's action is in $\mathcal E_\epsilon$, where the probability is taken over the choice of the trappified canvas in $\sch P'$. Let $Y$ be a random variable counting the number of test rounds for which the decision function rejects.
	
	We start by proving the first bound. We need to upper-bound the probability that a deviation in $\mathcal{M}_{\geq k_\epsilon}(\mathcal E_\epsilon)$ is not detected, which happens if and only if $Y < w$. Let $x \in [s]$, we can always decompose $\Pr \left[ Y < w \right]$ as
	\begin{align}
		\Pr \left[ Y < w \right] &= \Pr \left[ Y < w \mid X \leq x \right]\Pr \left[ X \leq x \right] + \Pr \left[ Y < w \mid X > x \right]\Pr \left[ X > x \right]\\
		&\leq \Pr \left[ X \leq x \right] + \Pr \left[ Y < w \mid X > x \right].
	\end{align}
	We now aim to bound both terms above --- an intuitive depiction of the bound derived below is presented in Figure~\ref{fig:detection_amplification}.
	
	Let $\cptp E \in \mathcal{M}_{\geq k_\epsilon}(\mathcal E_\epsilon)$. In this case, by definition of $\cptp E$ and construction of $\sch P'$, $X$ is lower-bounded in the usual stochastic order by a variable $\tilde{X}$ following a hypergeometric variable distribution of parameters $(n, k_\epsilon, s)$. We fix $x = \left( \frac{k_\epsilon}{n} - \chi \right) s$ for $\chi \geq 0$ and use tail bounds for the hypergeometric distribution to get
	\begin{align}
		\Pr \left[ X \leq \left( \frac{k_\epsilon}{n} - \chi \right) s \right] \leq \Pr \left[ \tilde{X} \leq \left( \frac{k_\epsilon}{n} - \chi \right) s \right] \leq \exp \left( -2 \chi^2 s \right).
	\end{align}
	
	For the other term, note that $Y$, conditioned on a lower bound $x$ for $X$, is lower-bounded in the usual stochastic order by an $(x, 1-\epsilon)$-binomially distributed random variable $\tilde{Y}$. Hoeffding's inequality for the binomial distribution then implies that
	\begin{align}
		\Pr \left[ \left. Y < w \; \right| \; X > x \right] \leq \Pr \left[ \tilde{Y} < w \right] \leq \exp \left( -2 \frac{(x(1-\epsilon) - w)^2}{x} \right).
	\end{align}
	
	All in all, replacing the value of $x$ above with $\left( \frac{k_\epsilon}{n} - \chi \right) s$ and combining it with the first bound, we have for $\chi \leq \frac{k_\epsilon}{n} - \frac{w}{s(1-\epsilon)}$ that
	\begin{align}
		\Pr \left[ Y < w \right] \leq \exp \left( -2 \chi^2 s \right) + \exp \left( -2 \frac{\left(\left( \frac{k_\epsilon}{n} - \chi \right) s(1-\epsilon) - w\right)^2}{\left( \frac{k_\epsilon}{n} - \chi \right) s} \right).
	\end{align}
	This concludes the first statement.
	
	For the second statement, we need to upper-bound the probability that a deviation in $\mathcal{H}_{< k_\delta}(\mathcal E_\delta)$ is detected. We can similarly decompose $\Pr \left[ Y \geq w \right]$ as
	\begin{align}
		\Pr \left[ Y \geq w \right] \leq \Pr \left[ Y \geq w \mid X < x \right] + \Pr \left[ X \geq x \right].
	\end{align}	
	
	Let $\cptp E \in \mathcal{H}_{< k_\delta}(\mathcal E_\delta)$. Here again, an intuitive depiction of the bound derived below is presented in Figure~\ref{fig:insensitivity_amplification}. Now $X$ is upper-bounded in the usual stochastic order by a variable $\tilde{X}$ following a hypergeometric distribution of parameters $(n, k_\delta, s)$, by definition of $\cptp E$. This holds here because the scheme is perfectly insensitive to $\Id$, and therefore the identity never triggers tests. It then holds for all $\chi \geq 0$ that
	\begin{align}
		\Pr \left[ X \geq \left( \frac{k_\delta}{n} + \chi \right) s \right] \leq \Pr \left[ \tilde{X} \geq \left( \frac{k_\delta}{n} + \chi \right) s \right] \leq \exp \left( -2 \chi^2 s \right),
	\end{align}
	using tail bounds for the hypergeometric distribution.
	
	Similarly, here, $Y$ (conditioned on an upper bound $x$ for $X$) is upper-bounded in the usual stochastic order by an $(x, \delta)$-binomially distributed random variable $\tilde{Y}$. This also holds because of the perfect insensitivity of tests to $\Id$. Hoeffding's inequality yields
	\begin{align}
		\Pr \left[ \left. Y \geq w \; \right| \; X \leq x \right] \leq \Pr \left[ \tilde{Y} \geq w\right] \leq \exp \left( -2 \frac{(x\delta - w)^2}{x} \right).
	\end{align}
	
	We then conclude for $\chi \leq \frac{w}{s\delta} - \frac{k_\delta}{n}$ that
	\begin{align}
		\Pr \left[ Y \geq w \right] \leq \exp \left( -2 \chi^2 s \right) + \exp \left( -2 \frac{\left(\left( \frac{k_\delta}{n} + \chi \right) s\delta - w\right)^2}{\left( \frac{k_\delta}{n} + \chi \right) s} \right). \tag*{\qedhere}
	\end{align}
\end{proof}

\begin{figure}[tb]
  \centering
  \begin{picture}(0,0)\includegraphics{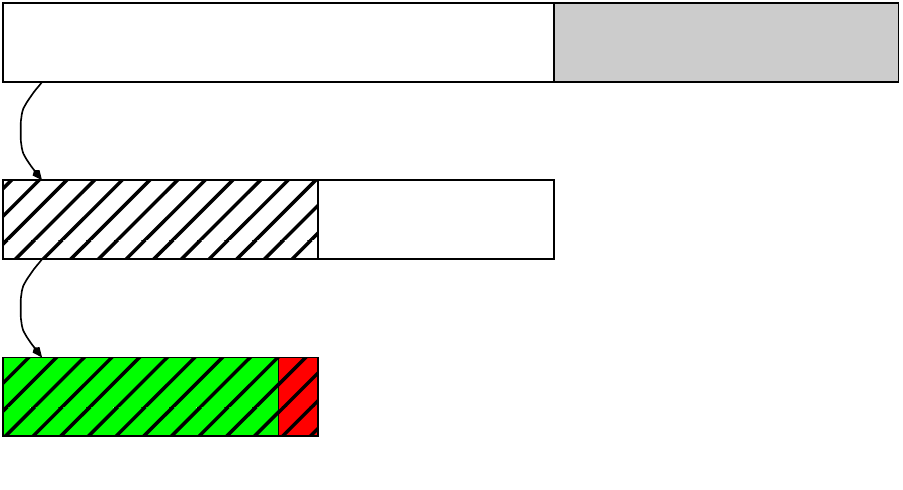}\end{picture}\setlength{\unitlength}{4144sp}\begingroup\makeatletter\ifx\SetFigFont\undefined \gdef\SetFigFont#1#2#3#4#5{\reset@font\fontsize{#1}{#2pt}\fontfamily{#3}\fontseries{#4}\fontshape{#5}\selectfont}\fi\endgroup \begin{picture}(4119,2236)(-11,-1385)
\put(181,299){\makebox(0,0)[lb]{\smash{{\SetFigFont{9}{10.8}{\rmdefault}{\mddefault}{\updefault}{\color[rgb]{0,0,0}$s$ test rounds}}}}}
\put(181,-511){\makebox(0,0)[lb]{\smash{{\SetFigFont{9}{10.8}{\rmdefault}{\mddefault}{\updefault}{\color[rgb]{0,0,0}More than $\frac{k_\epsilon s}{n}$  test rounds where $\cptp E \in \mathcal E_\epsilon$ w.h.p.}}}}}
\put(2701,299){\makebox(0,0)[lb]{\smash{{\SetFigFont{9}{10.8}{\rmdefault}{\mddefault}{\updefault}{\color[rgb]{0,0,0}$d$ computation rounds}}}}}
\put(181,-1321){\makebox(0,0)[lb]{\smash{{\SetFigFont{9}{10.8}{\rmdefault}{\mddefault}{\updefault}{\color[rgb]{0,0,0}More than $\frac{k_\epsilon s}{n}(1-\epsilon)$ failed test rounds w.h.p.}}}}}
\end{picture}   \caption{We consider $\cptp E \in \mathcal{M}_{\geq k_\epsilon}(\mathcal E_\epsilon)$ and its
    likely effect on the test rounds.  W.h.p, at least $k_\epsilon s / n$
    test rounds will be affected by an element of $\mathcal E_\epsilon$
    (hatched), among which a fraction $(1-\epsilon)$ will trigger a
    rejection by the decision function of individual test rounds
    (green). The undetected fraction is depicted in red. As a
    consequence, all $\cptp E \in \mathcal{M}_{\geq k_\epsilon}(\mathcal E_\epsilon)$
    will be detected w.h.p whenever $k_\epsilon > \frac{nw}
    {s(1-\epsilon)}$.}
    \label{fig:detection_amplification}
\end{figure}
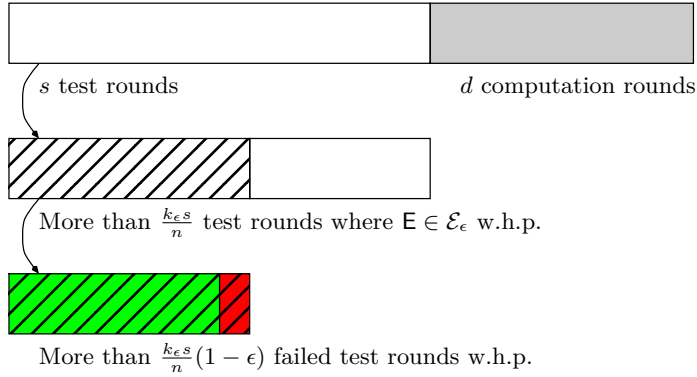

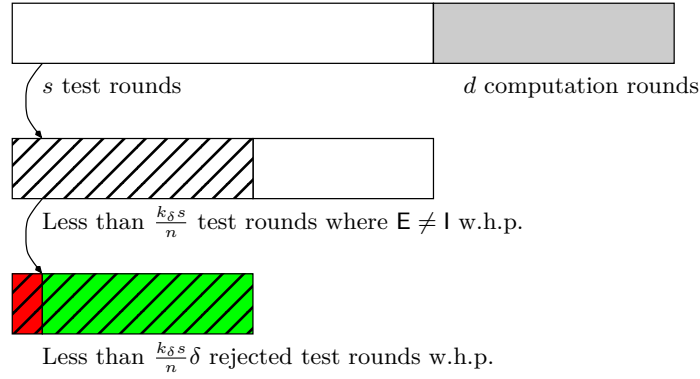
\begin{figure}[tb]
  \centering
  \begin{picture}(0,0)\includegraphics{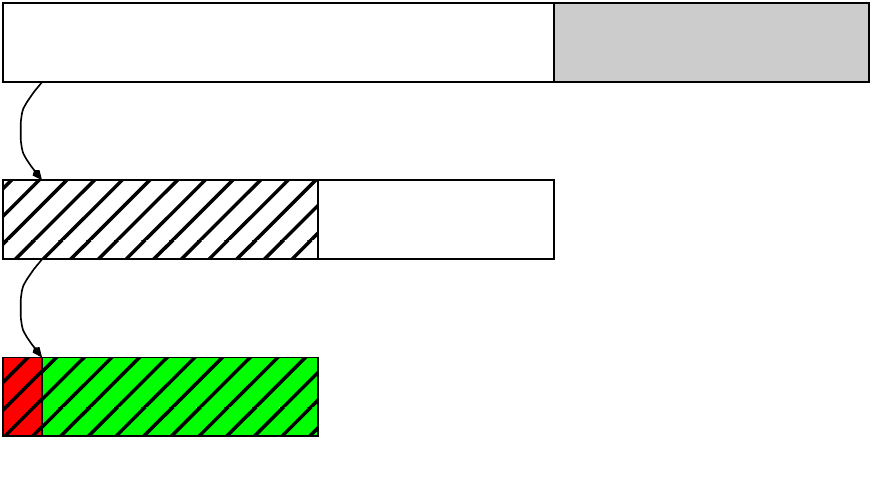}\end{picture}\setlength{\unitlength}{4144sp}\begingroup\makeatletter\ifx\SetFigFont\undefined \gdef\SetFigFont#1#2#3#4#5{\reset@font\fontsize{#1}{#2pt}\fontfamily{#3}\fontseries{#4}\fontshape{#5}\selectfont}\fi\endgroup \begin{picture}(3984,2236)(-11,-1385)
\put(181,299){\makebox(0,0)[lb]{\smash{{\SetFigFont{9}{10.8}{\rmdefault}{\mddefault}{\updefault}{\color[rgb]{0,0,0}$s$ test rounds}}}}}
\put(2701,299){\makebox(0,0)[lb]{\smash{{\SetFigFont{9}{10.8}{\rmdefault}{\mddefault}{\updefault}{\color[rgb]{0,0,0}$d$ computation rounds}}}}}
\put(181,-511){\makebox(0,0)[lb]{\smash{{\SetFigFont{9}{10.8}{\rmdefault}{\mddefault}{\updefault}{\color[rgb]{0,0,0}Less than $\frac{k_\delta s}{n}$  test rounds where $\cptp E\neq \Id$ w.h.p.}}}}}
\put(181,-1321){\makebox(0,0)[lb]{\smash{{\SetFigFont{9}{10.8}{\rmdefault}{\mddefault}{\updefault}{\color[rgb]{0,0,0}Less than $\frac{k_\delta s}{n}\delta$ rejected test rounds w.h.p.}}}}}
\end{picture}   \caption{We consider $\cptp E \in \mathcal{H}_{< k_\delta}(\mathcal E_\delta)$ and its
    effect on test rounds. W.h.p, at least $k_\delta s / n$ test rounds
    will be affected by a deviation in $\mathcal E_\delta$ (hatched) the
    rest being unaffected. Among the affected rounds at most a
    fraction $\delta$ will trigger the traps (red), while a fraction
    $1-\delta$ will be accepted (green). Hence, the probability of
    rejecting $\cptp E \in \mathcal{H}_{< k_\delta}(\mathcal E_\delta)$ is going
    to be exponentially small whenever $k_\delta < \frac{n w}{s\delta}$.}
  \label{fig:insensitivity_amplification}
\end{figure}

The consequence of the above theorem is that whenever the trappified schemes are constructed by interleaving computation rounds with test rounds chosen at random from a given set, the performance of the resulting protocol implementing SDQC crucially depends on the ability of these test rounds to detect harmful errors.
Therefore, when using the compiler, optimisation of the performance is achieved by focussing only on designing more efficient test rounds. This is addressed in Section \ref{sec:new_traps}.

\begin{remark}
Note that we do not make use in Definition~\ref{def:rvbqc_compiler} of the embedding function or computation class associated with the trappified scheme $\sch P$. In fact the initial scheme can even consist of pure traps as described in Remark~\ref{rem:pure_traps}. This is the case for the schemes described in the next sections.
If each trappified scheme used for tests can also embed the client's computation of interest, it is possible to use the alternative parallel repetition compiler presented in \ref{app:gen-comp} which has no separate computation rounds.
\end{remark}

\begin{remark}
The result presented here is close in spirit to the parallel repetition theorems from \cite{B21:secure,ACGH20:non}, which guarantee that if all tests of one kind pass, a malicious server cannot pass more than a constant fraction of another kind of tests with non negligible probability. In our case, the second kind of tests correspond to the computations of which the server cannot corrupt more than half with non-negligible probability. However, our result is stronger since we allow a constant fraction of tests to be wrong before aborting. This makes our construction robust to honest noise up to a certain global constant level, as shown in the next section.
\end{remark}

\subsection{Correctness Amplification for $\mathsf{BQP}$ via Majority Vote}
Theorem~\ref{thm:boost} has given detection and
insensitivity errors that are negligible $n$.
In order to recover exponential verifiability, we must now also make the
correctness error negligible in $n$. To this end, we 
recombine the multiple computation rounds into a single final result
so that error of weight lower than $k_\delta$ are corrected.

Whereas in the previous section we left the computation class undefined, 
here we restrict $\mathfrak C$ to be the class of $\mathsf{BQP}$
computations that can be implemented on $G$, and set $c$ to be the (constant)
probability of obtaining the incorrect result, with $c < 1/2$.
\footnote{Alternatively, if the size of the classical input $x$ to the computation is $|x|$,
then $c$ can chosen such that it is upper-bounded by $1/2 - 1/p(|x|)$ for any polynomial $p$.}

Then, we define $\sch V$ from the compiled $\sch P'$ by requiring
that the input subspace is symmetric with respect to exchanging
computation rounds -- i.e.~all computation rounds have the same
inputs -- and by defining the output subspace as the
bitwise majority vote of computation round outputs.

\begin{definition}[Amplified Trap Compiler for $\mathsf{BQP}$]
\label{def:rvbqc_bqp_compiler}
Let $\sch P$ be trappified scheme on a graph $G = (V, E)$, and let $d, s \in \mathbb{N}$, $n = d + s$ and $w \in [s]$. Let $\sch P'$ be the output of the Amplified Trap Compiler from Definition \ref{def:rvbqc_compiler} with parameters $n, d, s, w$.

Let $\mathfrak{C} = \mathsf{BQP} \cap \mathsf{MBQC}_{G, \preceq_{\sch P}}$ for the order $\preceq_{\sch P}$ induced by $\sch P$. The output of the Amplified Trap Compiler for $\mathsf{BQP}$ is a trappified pattern $\tilde{\sch P}$ for computations in $\mathfrak{C}$ which is equal to $\sch P'$ with the following additional constraints:
\begin{itemize}
\item The input subspace $\Pi_{I, C}$ is symmetric with respect to exchanging computation rounds.
\item The output subspace $\Pi_{O, C}$ is defined as the concatenation of the (classical) outputs of all computation rounds.
\item The decoding algorithm $\Deco$ is the bitwise majority vote of computation rounds outputs from the $d$ computations.
\end{itemize}
\end{definition}

Intuitively, if it is guaranteed that the fraction of all rounds
affected by a possibly harmful deviation is less than $(1 - 2c)/(2 - 2c)$
then the output of $\tilde{\sch P}$ will yield the correct result of the
computation. This is because, in the large $n$ limit, out of the $d$
computation rounds a fraction $c$ will be incorrect due to the
probabilistic nature of the computation itself. Consequently, to
maintain that more than $1/2$ the computation rounds yield the correct
result so that the majority vote is able to eliminate the spurious
results, the fraction $f$ of computation rounds that the deviation can
affect must satisfy $(1 - c)(1 - f) > 1/2$, that is 
$f < (1 - 2c)/(2 - 2c)$.
Due to the blindness of the scheme, it is enough to impose that no
more than a fraction $(1 - 2c)/(2 - 2c)$ of the $n$ rounds is affected by
the deviation to obtain the desired guarantee on the computation
rounds with high probability.

Recall that $\mathcal{M}_{\geq k}(\mathcal E) = \{\cptp E \in \mathcal P_{V^n} | \operatorname{wt}_{\mathcal E}(\cptp E) \geq k\}$
and $\Cdev{\cptp E} = \Deco\circ\Tr_{O_C^c}\circ\cptp E\circ (\TP)$.

\begin{theorem}[Exponential Correctness from Majority Vote]
\label{thm:exp-cor}
Let $\sch P$ be a trappified scheme on graph $G$. For $d, s \in \mathbb{N}$, $n = d + s$ and $w \in [s]$, let $\tilde{\sch P}$ be the trappified scheme obtained through the compiler of Definition~\ref{def:rvbqc_bqp_compiler} applied to $\sch P$ with parameters $n, d, s, w$.

Let $\mathcal{E}_\nu \subset \mathcal{P}_V$ be a set of Pauli deviations such that, for all computations $\cptp C \in \mathfrak{C} = \mathsf{BQP} \cap \mathsf{MBQC}_{G, \preceq_{\sch P}}$ the errors in $\mathcal{E}_\nu$ do not affect the computation, i.e.~$\Cdev{\cptp E} = \Cdev{\cptp \Id}$ for all $\cptp E \in \mathcal{E}_\nu$.
Let $c$ be the bounded error of $\mathsf{BQP}$ computations and $k_\nu < \frac{1 - 2c}{2 - 2c}n$.

Then $\tilde{\sch P}$ is $\nu'$-correct for computations in $\mathfrak{C}$ on $\mathcal{E}_\nu' = \mathcal{M}_{\geq n - k_\nu}(\mathcal E_\nu)$ with
\begin{equation}
\nu' = \min_{\chi \in \left[ 0, \frac{1 - 2c}{2 - 2c} - \frac{k_\nu}{n} \right]} \exp(-2\chi^2 d) + \exp(- 2\frac{\left(\left(1 - \frac{k_\nu}{n} - \chi\right)(1 - c) - \frac{1}{2}\right)^2}{\left(1 - \frac{k_\nu}{n} - \chi\right)}d).
\end{equation}
\end{theorem}

\begin{proof}

We will compute the bound on the correctness for finite $n$. First, define
two random variables $Z_1$ and $Z_2$ that account for possible sources
of erroneous results for individual computation rounds. More
precisely, $Z_1$ is the number of computation rounds that are affected
by a deviation containing an $\Y$ or $\Z$ for one of the qubits in the
round. $Z_2$ is the number of computation rounds which give the wrong
outcome due to the probabilistic nature of the computation
itself -- i.e.~inherent failures for the computation in the honest and
noise free case. Given that $\sch V$ uses a majority vote to recombine
the results of each computation rounds, as long a $Z_1 + Z_2 < d/2$,
then the output result will be correct.

Our goal now is to show that the probability that $Z_1 + Z_2$ is
greater than $d/2$ can be made negligible. An intuitive depiction of the proof is presented in Figure~\ref{fig:correctness_amplification}. For any $z_1$ one has the
following
\begin{align}
\Pr[Z_1+Z_2 \geq \frac{d}{2}]
& = \Pr[Z_1+Z_2 \geq \frac{d}{2} | Z_1 \leq z_1] \Pr[Z_1 \leq z_1] \\
& \quad + \Pr[Z_1+Z_2 \geq \frac{d}{2} | Z_1 > z_1] \Pr[Z_1 > z_1].
\end{align}
Then
\begin{align}
\Pr[Z_1+Z_2 \geq \frac{d}{2}]& \leq \Pr[Z_1+Z_2 \geq \frac{d}{2} | Z_1 \leq z_1] + \Pr[Z_1 > z_1] \\
& \leq \Pr[Z_2 \geq \frac{d}{2} - z_1 | Z_1 \leq z_1] + \Pr[Z_1 > z_1] \\
& \leq \Pr[Z_2 \geq \frac{d}{2} - z_1 | Z_1 = z_1] + \Pr[Z_1 > z_1].
\end{align}

Now, consider a deviation in $\mathcal{E}_\nu' = \mathcal{M}_{\geq n - k_\nu}(\mathcal E_\nu)$.
If we note $\mathcal{F} = \mathcal{P}_V \setminus \mathcal E_\nu$ the set of deviation for
which we have no correctness guarantees, 
we can rewrite this set as $\mathcal{E}_\nu' = \mathcal{M}_{< k_\nu}(\mathcal F)$. Using the tail bound
for the hypergeometric distribution defined by choosing independently
at random and without replacement $d$ computation rounds out of a
total of $n$ rounds, $k_\nu$ of which at most are affected by a
deviation in $\mathcal{F}$, one finds that, for $z_1 = (k_\nu/n + \chi) d$ with
$\chi > 0$,
\begin{equation}
\Pr[Z_1 > \left(\frac{k_\nu}{n} + \chi\right) d] \leq \exp(-2\chi^2 d).
\end{equation}
Additionally, once $Z_1$ is fixed, $Z_2$ is binomially distributed
with probability $c$. Therefore, using tail bound for this
distribution, one has for $\xi > 0$
\begin{equation}
\Pr[Z_2 \geq (c + \xi)\left(1 - \frac{k_\nu}{n} - \chi\right) d | Z_1 = \left(\frac{k_\nu}{n}+\chi\right) d] \leq \exp(- 2\left(1-\frac{k_\nu}{n} - \chi\right)d \xi^2).
\end{equation}

In the bounds above, we can set
\begin{equation}
\label{eq:corr_amp_cond}
\frac{d}{2} - \left(\frac{k_\nu}{n} + \chi\right) d = (c + \xi)\left(1 - \frac{k_\nu}{n} - \chi\right) d.
\end{equation}
This equation has solutions for $\chi, \xi > 0$ when $k_\nu/n < (1 - 2c)/(2 - 2c)$. 
Using these inequalities, we obtain that
\begin{equation}
\Pr[Z_1 + Z_2 \geq \frac{d}{2}] \leq \exp(-2\chi^2 d) + \exp(- 2\left(1- \frac{k_\nu}{n} - \chi\right)d \xi^2).
\end{equation}
Using Equation \ref{eq:corr_amp_cond} to express $\xi$ as a function of $\chi, k_\nu$ we get
\begin{equation}
\xi = \frac{2\left(1 - \frac{k_\nu}{n} - \chi\right)(1 - c) - 1}{2\left(1 - \frac{k_\nu}{n} - \chi\right)}.
\end{equation}
Substituting $\xi$ into the bound above we get
\begin{equation}
\Pr[Z_1 + Z_2 \geq \frac{d}{2}] \leq \exp(-2\chi^2 d) + \exp(- 2 \frac{\left(2\left(1 - \frac{k_\nu}{n} - \chi\right)(1 - c) - \frac{1}{2}\right)^2}{\left(1 - \frac{k_\nu}{n} - \chi\right)}d).
\end{equation}
Minimising over $0 < \chi < \frac{1 - 2c}{2 - 2c} - \frac{k_\nu}{n}$ yields the desired result.
\end{proof}

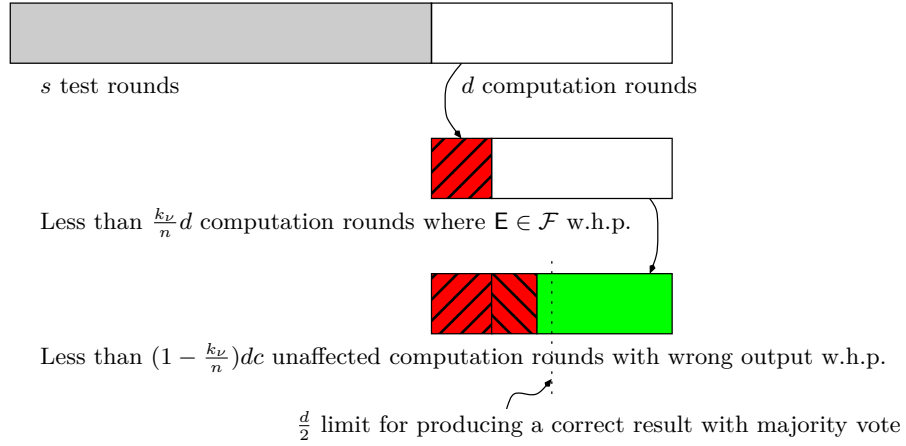
\begin{figure}[tb]
  \centering
  \begin{picture}(0,0)\includegraphics{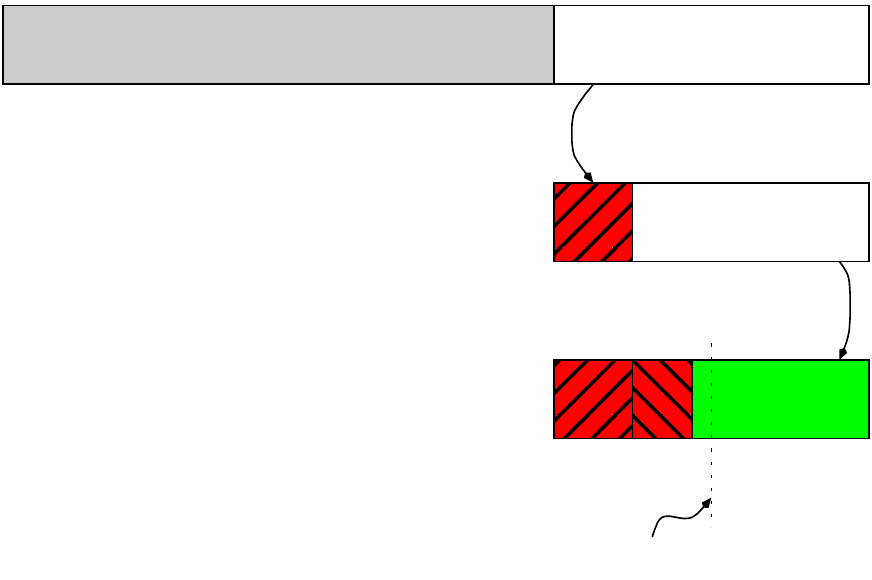}\end{picture}\setlength{\unitlength}{4144sp}\begingroup\makeatletter\ifx\SetFigFont\undefined \gdef\SetFigFont#1#2#3#4#5{\reset@font\fontsize{#1}{#2pt}\fontfamily{#3}\fontseries{#4}\fontshape{#5}\selectfont}\fi\endgroup \begin{picture}(3984,2637)(-11,-1786)
\put(181,299){\makebox(0,0)[lb]{\smash{{\SetFigFont{9}{10.8}{\rmdefault}{\mddefault}{\updefault}{\color[rgb]{0,0,0}$s$ test rounds}}}}}
\put(2701,299){\makebox(0,0)[lb]{\smash{{\SetFigFont{9}{10.8}{\rmdefault}{\mddefault}{\updefault}{\color[rgb]{0,0,0}$d$ computation rounds}}}}}
\put(181,-511){\makebox(0,0)[lb]{\smash{{\SetFigFont{9}{10.8}{\rmdefault}{\mddefault}{\updefault}{\color[rgb]{0,0,0}Less than $\frac{k_\nu}{n}d$ computation rounds where $\cptp E \in \mathcal F$ w.h.p.}}}}}
\put(181,-1321){\makebox(0,0)[lb]{\smash{{\SetFigFont{9}{10.8}{\rmdefault}{\mddefault}{\updefault}{\color[rgb]{0,0,0}Less than $(1-\frac{k_\nu}{n})dc$ unaffected computation rounds with wrong output w.h.p.}}}}}
\put(1711,-1726){\makebox(0,0)[lb]{\smash{{\SetFigFont{9}{10.8}{\rmdefault}{\mddefault}{\updefault}{\color[rgb]{0,0,0}$\frac{d}{2}$ limit for producing a correct result with majority vote}}}}}
\end{picture}   \caption{Here, we consider the effect of a deviation $\cptp E \in
   \mathcal{M}_{\geq n - k_\nu}(\mathcal E_\nu)$ on computation rounds, together
    with the inherent probability of failure on the ability to still
    produce a correct output after the majority vote.  W.h.p, no more
    than $\frac{k_\nu}{n}d$ computation rounds are affected by a deviation
    that is not in $\mathcal E_\nu$ (clear area on the second line). The
    affected rounds are the red right-hatched zone on the second
    line. On the remaining non affected rounds, with high probability,
    no more than a fraction $c$ will be outputting the wrong result
    (red left-hatched). Hence, the total number of wrong output will
    be below $d/2$ w.h.p. when $k_\nu/n < (1 - 2c)/(2 - 2c)$.}
  \label{fig:correctness_amplification}
\end{figure}

\subsection{Putting the Pieces Together}

To conclude this section, we combine Theorems \ref{thm:boost} and
\ref{thm:exp-cor} to obtain simultaneous negligibility for detection,
insensitivity and correctness errors
for sets which are compatible with the requirements of Theorems \ref{thm:verif} and \ref{thm:robust_verif}.
Recall that $\mathfrak{C} = \mathsf{BQP} \cap \mathsf{MBQC}_{G, \preceq_{\sch P}}$.

\begin{theorem}[Amplified Security and Robustness]
\label{thm:comb}
  Let $\sch P$ be a trappified scheme on graph $G$ with measurement order $\preceq_{\sch P}$.
  For $d, s \in \mathbb{N}$, $n = d + s$ and $w \in [s]$, let $\tilde{\sch P}$ be the trappified scheme obtained through the compiler of Definition~\ref{def:rvbqc_bqp_compiler} applied to $\sch P$ with parameters $n, d, s, w$.
  
\noindent Let $\mathcal E_\epsilon, \mathcal E_\delta, \mathcal{E}_\nu$ be subsets of Pauli deviations such that:
  \begin{itemize}
  \item $\mathcal{E}_\nu \subseteq \{\cptp E \in \mathcal{P}_V \mid \forall \cptp C \in \mathfrak{C}, \forall T \in \sch P, \Cdev{\cptp E} = \Cdev{\cptp \Id}\}$;
  \item $\mathcal{P}_V \setminus \mathcal E_\epsilon \subset \mathcal{E}_\nu$;
  \item $\sch P$ $\epsilon$-detects $\mathcal E_\epsilon$, is $\delta$-insensitive to $\mathcal E_\delta$ and perfectly insensitive to $\Id$;
  \item The honest Server's noise is modelled by sampling for each computation or test round an error $\cptp E \in \mathcal E_\delta$ with probability $p_\delta$ and $\cptp E = \Id$ with probability $1 - p_\delta$. 
  \end{itemize}
  Let $k_\epsilon, k_\delta, k_\nu \in \mathbb{N}$ and $c$ is the bounded error of $\mathsf{BQP}$. If the following conditions are satisfied:
  \begin{itemize}
  \item $\frac{w}{s(1 - \epsilon)} < \frac{k_\epsilon}{n} \leq \frac{k_\nu}{n} < \frac{1 - 2c}{2 - 2c}$;
  \item $p_\delta < \frac{k_\delta}{n} < \frac{w}{s\delta}$;
  \item $k_\delta \leq k_\nu$;
  \end{itemize}
  then Protocol \ref{proto:dev_detect} for computing CPTP maps $\cptp C$ in $\mathfrak C$ using $\tilde{\sch P}$ is $p_\delta' + \delta' + \nu'$-correct and $\max(\epsilon', \nu')$-secure in the Abstract Cryptography framework for
  \begin{align}
  	\epsilon' &= \min_{\chi \in \left[ 0, \frac{k_\epsilon}{n} - \frac{w}{s(1-\epsilon)} \right]} \exp \left( -2 \chi^2 s \right) + \exp \left( -2 \frac{\left(\left( \frac{k_\epsilon}{n} - \chi \right)(1-\epsilon) - \frac{w}{s}\right)^2}{\left( \frac{k_\epsilon}{n} - \chi \right)} s\right),\\
  	\delta' &= \min_{\chi \in \left[ 0, \frac{w}{s\delta} - \frac{k_\delta}{n} \right]} \exp \left( -2 \chi^2 s \right) + \exp \left( -2 \frac{\left(\left( \frac{k_\delta}{n} + \chi \right)\delta - \frac{w}{s}\right)^2}{\left( \frac{k_\delta}{n} + \chi \right)} s\right),\\
  	\nu' &= \min_{\chi \in \left[ 0, \frac{1 - 2c}{2 - 2c} - \frac{k_\nu}{n} \right]} \exp(-2\chi^2 d) + \exp(- 2\frac{\left(\left(1 - \frac{k_\nu}{n} - \chi\right)(1 - c) - \frac{1}{2}\right)^2}{\left(1 - \frac{k_\nu}{n} - \chi\right)}d),\\
  	p_\delta' &= \exp\left(- 2\left(p_\delta - \frac{k_\delta}{n}\right)^2 n\right).
  \end{align}
\end{theorem}

Hence, as long as $p_\delta\delta < \frac{1 - 2c}{2 - 2c}(1 - \epsilon)$, it is possible to choose parameters $s, d, w$ which scale linearly with $n$ and yield correctness and security bounds which are negligible in $n$.

\begin{proof}
On one hand we have the following prerequisites for applying Theorems~\ref{thm:boost} and \ref{thm:exp-cor}:
\begin{enumerate}
\item $\sch P$ $\epsilon$-detects $\mathcal{E}_\epsilon$, is $\delta$-insensitive to $\mathcal{E}_\delta$ and perfectly insensitive to at least $\{\Id\}$;
\item $\mathcal{E}_\nu \subset \{\cptp E \in \mathcal{P}_V | \forall \cptp C \in \mathfrak{C}, \Cdev{\cptp E} = \Cdev{\cptp \Id}\}$;
\item $k_\epsilon > nw/(s(1-\epsilon))$, $k_\delta < nw/(s\delta)$ and $k_\nu < \frac{1 - 2c}{2 - 2c}n$, for $c$ the $\mathsf{BQP}$ error.
\end{enumerate}
We have assumed that these three conditions are satisfied, therefore the compiled scheme $\tilde{\sch P}$:
\begin{enumerate}
\item $\epsilon'$-detects $\mathcal{E}_\epsilon' = \mathcal{M}_{\geq k_\epsilon}(\mathcal E_\epsilon)$;
\item is $\delta'$-insensitive to $\mathcal{E}_\delta' = \mathcal{H}_{< k_\delta}(\mathcal E_\delta)$;
\item is $\nu'$-correct for computations in $\mathfrak{C}$ on $\mathcal{E}_\nu' = \mathcal{M}_{\geq n - k_\nu}(\mathcal E_\nu)$.
\end{enumerate}

On the other hand, the application of Theorems \ref{thm:verif} and \ref{thm:robust_verif} require the following conditions on $\tilde{\sch P}$:
\begin{enumerate}
\item The embedding is proper, as specified in Definition \ref{def:prop-embed};
\item $\Id^{\otimes n} \in \mathcal E_\delta'$ and $\Id^{\otimes n} \in \mathcal{E}_\nu'$, where $\Id^{\otimes}$ acts as $\Id$ on all $n$ graphs;
\item $\mathcal P_{V^n} \setminus \mathcal E_\epsilon' \subset \mathcal E_\nu'$;
\item $\mathcal{E}_\delta' \subset \mathcal{E}_\nu'$;
\item The global noise model over all $n$ graphs of the Server samples an error $\cptp E \in \mathcal E_\delta'$ with probability $(1-p_\delta')$ and $\cptp E \in \mathcal P_{V^n} \setminus \mathcal E_\delta'$ with probability $p_\delta'$.
\end{enumerate}
If these conditions are satisfied, then Protocol~\ref{proto:dev_detect} using $\tilde{\sch P}$ is $p_\delta' + \delta' + \nu'$-correct and $\max(\epsilon', \nu')$-secure in the Abstract Cryptography framework. We now show that these four conditions are satisfied in order.

Since the graph common to all trappified canvases in $\tilde{\sch P}$ are has no edges between computation and test rounds, no information from one round can leak into the others and the embedding is necessarily proper.

We have $\mathcal{H}_{< k_1}(\mathcal F) \subset \mathcal{H}_{< k_2}(\mathcal F)$ if $k_1 < k_2$, therefore since $0 < k_\delta$:
\begin{equation}
\{\Id^{\otimes n}\} = \mathcal{H}_{< 1}(\mathcal E_\delta) \subset \mathcal{H}_{< k_\delta}(\mathcal E_\delta) = \mathcal{E}_\delta'.
\end{equation}

By definition of $\mathcal{E}_\nu$ we have that $\Id \in \mathcal{E}_\nu$. Then since $\mathcal{E}_\nu' = \mathcal{M}_{\geq n - k_\nu}(\mathcal E_\nu)$, we have as well that $\Id^{\otimes n} \in \mathcal{E}_\nu'$.

Then we must show that $\mathcal P_{V^n} \setminus \mathcal E_\epsilon' \subset \mathcal E_\nu'$. We have that
\begin{equation}
\mathcal P_{V^n} \setminus \mathcal E_\epsilon' = \mathcal P_{V^n} \setminus \mathcal{M}_{\geq k_\epsilon}(\mathcal E_\epsilon) = \mathcal{M}_{< k_\epsilon}(\mathcal E_\epsilon) = \mathcal{M}_{\geq n - k_\epsilon}(\mathcal{P}_V \setminus \mathcal E_\epsilon).
\end{equation}
We have $\mathcal{M}_{\geq k_1}(\mathcal F) \subset \mathcal{M}_{\geq k_2}(\mathcal F)$ if $k_1 \geq k_2$. Since we assumed that $k_\epsilon \leq  k_\nu$ then
\begin{equation}
\mathcal P_{V^n} \setminus \mathcal E_\epsilon' \subset \mathcal{M}_{\geq n - k_\nu}(\mathcal{P}_V \setminus \mathcal E_\epsilon).
\end{equation}
Finally, we have assumed that $\mathcal{P}_V \setminus \mathcal E_\epsilon \subset \mathcal E_\nu$ and we have that $\mathcal{M}_{\geq k}(\mathcal F) \subset \mathcal{M}_{\geq k}(\mathcal G)$ if $\mathcal{F} \subset \mathcal{G}$, meaning that
\begin{equation}
\mathcal P_{V^n} \setminus \mathcal E_\epsilon' \subset \mathcal{M}_{\geq n - k_\nu}(\mathcal E_\nu) = \mathcal E_\nu'.
\end{equation}

We then show that $\mathcal{E}_\delta' \subset \mathcal{E}_\nu'$. Since $k_\delta \leq k_\nu$ and $\Id \in \mathcal E_\nu$, we have
\begin{equation}
\mathcal{E}_\delta' = \mathcal{H}_{< k_\delta}(\mathcal E_\delta) \subset \mathcal{M}_{\geq n - k_\delta}(\{\Id\}) \subset \mathcal{M}_{\geq n - k_\nu}(\{\Id\}) \subset \mathcal{M}_{\geq n - k_\nu}(\mathcal E_\nu) = \mathcal E_\nu'.
\end{equation}

Finally, we have to bound the probability $p_\delta'$ that the honest-but-noisy Server's error is not in $\mathcal{E}_\delta'$. Let $X$ be a random variable counting the number of rounds in which the error comes from $\mathcal E_\delta$. We then have $p_\delta' = \Pr[X \geq k_\delta]$. We know that the honest Server's noise is sampled independently each round from $\mathcal E_\delta$ with probability $p_\delta$ and is equal to $\Id$ with probability $1 - p_\delta$. $X$ then follows a binomial distribution of parameter $(n, p_\delta)$. Since we assumed that $k_\delta > p_\delta n$, we can bound $p_\delta'$ with Hoeffding's inequality for binomial distributions:
\begin{equation}
p_\delta' = \Pr[X \geq k_\delta] \leq \exp\left(- 2 \left(p_\delta - \frac{k_\delta}{n}\right)^2 n\right).
\end{equation}
This in turn completes the proof.
\end{proof}

\section{New Optimised Trappified Schemes from Stabiliser Testing}
\label{sec:new_traps}

In this section we demonstrate how the various tools and
techniques introduced earlier can be combined to design trappified
schemes that provide efficient and robust verifiability. To achieve
this, we use Remark~\ref{rem:tr-can-sch} and Lemma \ref{lem:compos-sch} to
construct a trappified scheme $\sch T$ based on stabiliser testing 
with a constant detection error.
Here we again focus on classical-input
classical-output computations.
Theorem~\ref{thm:comb} show that it is 
sufficient in this case to focus on designing test rounds, with the 
compiler from Definition \ref{def:rvbqc_bqp_compiler} then 
boosting the detection, insensitivity and correctness.

In the process, we show a close correspondence between prepare-and-send
protocols derived from \cite{FK17:unconditionally}, and protocols based on stabiliser tests
following \cite{M13:interactive}. This broadens noticeably the possibilities for designing new types of 
trappified patterns beyond those which are used by existing prepare-and-send protocols. 
It also allows to transfer existing protocols based on stabiliser testing from the 
non-communicating multi-server setting to the prepare-and-send model, thus lowering 
the assumptions of these protocols and making them more readily implementable and practical.
We show in later subsections how to use the compiler results together with these 
new possibilities to optimise the current state-of-the-art protocol of \cite{LMKO21:verifying}.

\subsection{Trappified Schemes from Subset Stabiliser Testing}

Given $G=(V,E)$ and a partial order $\preceq_G$ on $V$, the first step for constructing a verification protocol for computations on $G$ is to detect deviations from the server.
To this end, we recall that any action from the server can be always be viewed as first performing the unitary part of Protocol~\ref{proto:ubqc} followed by a pure deviation that is independent from the computation delegated to the server (see Section~\ref{sec:verif}).
To be constructive and build traps that can be easily computed and checked by the client, we impose in this section that the outcomes of trappified canvases are deterministic and that they accept with probability $1$ for honest executions of the protocol. This means that the condition that all tests will be perfectly insensitive to $\Id$, as required for Theorem~\ref{thm:comb}.

We first focus on the simplest case of deterministic functions,
where the decision algorithm $\tau$ for the trappified canvas
is such that $\tau(t) = t_i$ where $t_i$ is
measurement outcome of qubit $i$. In other words the test
round accepts if the outcome $t_i = 0$, which corresponds to obtaining
outcome $\ket 0$ for qubit $i$, while all other measurements
outcomes $t_j$ for $j \neq i$ are ignored.\footnote{Recall that throughout
the paper, our convention is to view rotated $\{\ket{\pm_\theta}\}$
measurements as $\Z$ rotations followed by a Hadamard gate and a
measurement in the computational basis.}

For the outcome of the trappified canvas to be deterministic, qubit $i$ must be equal to
$\ket 0$ in absence of deviations before the computational basis
measurement. In other words, the state of $i$ is an eigenstate
of $\Z_i$. By commuting $\Z_i$
towards the initialisation of the qubits -- through the Hadamard gate and the entangling
operations defined by the graph $G$, 
we conclude that determinism and acceptance of
deviation-less test rounds implies that the initial state of the
qubits before running the protocol is an eigenstate of
$\X_i \prod_{j \in N_G(i)} \Z_j = \cptp S^{\mathtt{i}}$.\footnote{We use the upper index to avoid confusion between the operator applied to qubit $i$ and the stabiliser associated to qubit $i$.}

The following lemma explains how to prepare a single-qubit tensor
product state stabilised by such given Pauli operator.
\begin{lemma}[Tensor Product Preparation of a State in a Stabiliser Subspace]\label{lem:prep}
 Let $\cptp P$ be an element of the Pauli group over $N$ qubits, such
 that $\cptp P^2 \neq -\Id$. Then, there exists $\ket \psi
 = \bigotimes_{i=1}^N \ket{\psi_i}$ such that $\ket \psi = \cptp
 P \ket \psi$, and $\forall
 i, \ket{\psi_i} \in \{\ket{0}, \ket{+}, \ket{+_{\pi/2}}\}$.
\end{lemma}

\begin{proof}
 Without loss of generality, one can write $\cptp P =
 s \bigotimes_i \cptp P(i)$ with $s= \pm 1$ and where $\cptp
 P(i) \in \{\Id, \X, \Y, \Z\}$ is the restriction of $\cptp P$ to
 qubit $i$. Then by construction, $\cptp P \in \langle S \rangle$,
 where $\langle \mathcal S \rangle$ denotes the multiplicative group generated
 by the set $\mathcal S = \{s\cptp P(i_0) \bigotimes_{j\neq
 i_0} \Id\} \cup \{\cptp P(i) \bigotimes_{j\neq i} \Id \}_{i\neq
 i_0}$, where $i_0$ is the smallest index $i$ for which
 $\cptp P(i) \neq \Id$. Now, consider the state that is obtained by taking
 the tensor product of single qubit states that are the common $+1$
 eigenstates of the operators in set $\mathcal S$. The above shows that it is
 a $+1$ eigenstate of all operators in $\langle \mathcal S \rangle$, and in
 particular of $\cptp P$, which concludes the proof as eigenstates of
 single-qubit Pauli operators are precisely the desired set.
\end{proof}
One can further note that the above lemma also holds for a set $\mathcal R$ of
Pauli operators if,
\begin{equation}
\label{eq:overlap}
\forall \cptp P, \cptp Q \in \mathcal R, \ \forall i \in V, \ \cptp P(i) = \cptp Q(i) \mbox{ or } \cptp P(i) = \Id \mbox{ or } \cptp Q(i) = \Id.
\end{equation}

Now take $\mathcal R$ a set of Pauli operators generating the stabiliser group
of $\ket G$, and $\{\mathcal R^{(k)}\}_j$ a collection of subsets of $\mathcal R$ that
such that each $\mathcal R^{(k)}$ satisfies the condition of
Equation~\ref{eq:overlap} and $\cup_k \mathcal R^{(k)} = \mathcal R$ -- note that $\mathcal R$
need not be a minimal set of generators. We then construct a set of
trappified canvases $T^{(k)}$ which have $V$ as their input set and
for which all qubits are measured in the $\X$ basis. They only differ
in the prepared input states, each being prescribed by
Lemma~\ref{lem:prep} for the stabilisers in $\mathcal R^{(k)}$ -- that is qubits
are prepared in an $\X$, $\Y$ or $\Z$ eigenstate each time one of the
Pauli operator in $\mathcal R^{(k)}$ is respectively $\X$, $\Y$ or $\Z$ for
this qubit, and chosen arbitrarily to be $\X$ eigenstates elsewhere. As above, the
computation defined by the pattern where all qubits are measured in the $\X$ basis amounts to measuring the stabiliser generators $\cptp S^{\mathtt{i}}$. The
output distribution $\pd T^{(k)}$ can be computed given the prepared
input state for $T^{(k)}$ using elementary properties of stabiliser
states. But for our purposes, it is sufficient to construct the
decision function $\tau^{(k)}$. This can be done by noting that for
all $\cptp P \in \mathcal R^{(k)}$, there is a unique binary vector $\{p_i\}_i$ such
that $\cptp P = \prod_i \mathcal S_i^{p_i}$. This, in turn, implies that $\pd
T^{(k)}$ is such that $\bigoplus_i p_i t_i = 0$ where $t_i$ is the
outcome of the measurement of the $i$-th qubit in the $\X$
basis. Therefore, we define
\begin{equation}
\label{eq:general-tau}
\tau^{(k)}(t) = \bigwedge_{\cptp P \in \mathcal R^{(k)}} \left(\bigoplus_i p_i t_i = 0\right),
\end{equation}
which reconstructs the measurement outcomes of stabilisers in
$\mathcal R^{(k)}$ from the measurements outcomes of operators
$\cptp S^{\mathtt{i}}$. The function $\tau^{(k)}(t)$ will accept whenever the measurement outcomes of
all stabilisers in $\mathcal R^{(k)}$ are zero. We denote by $\mathcal
E_\epsilon^{(k)}$ the set of Pauli deviations that are perfectly detected by
$T^{(k)}$ and $\mathcal E_\delta^{(k)} = \mathcal P_V \setminus \mathcal
E_\epsilon^{(k)}$ the set of deviations to which $T^{(k)}$ is perfectly
insensitive.

Now, using Remark~\ref{rem:tr-can-sch} and Lemma \ref{lem:compos-sch}, the
trappified canvases $T^{(k)}$ can be composed with equal probability
$p$ to obtain a trappified scheme $\sch T$. We then consider the sets
of all Pauli deviations $\mathcal E_\epsilon = \bigcup_k \mathcal E_\epsilon^{(k)}$
and $\mathcal E_\delta = \bigcup_k \mathcal E_\delta^{(k)} = \mathcal P_V$. We
conclude that the scheme $\sch T$ then $(1-p)$-detects $\mathcal E_\epsilon$
and is $(1-p)$-insensitive to $\mathcal P_V$. Note that these values are upper-bounds,
with equality being achieved if there is no overlap in the set of errors which
each canvas can detect.

The scheme $\sch T$ therefore detects all possibly harmful deviations with finite
probability, and is partly insensitive to all deviations -- i.e.~both
harmless and harmful -- that can affect computations in $\mathfrak C$.

\subsubsection{A Linear Programming Problem for Trap Optimisation}
\label{app:lin-prog}
At first glance, the main goal to optimise 
such schemes seems to be to lower as much as possible the number of subsets of stabilisers $\mathcal R^{(k)}$ which
cannot be tested at the same time. Each such subset of stabilisers needs a 
different canvas $T^{(k)}$ to test for it, and the probability $p$ increases with a lower number of canvases.
An increase in $p$ automatically decreases the detection and insensitivity errors. These in turn
appear in the exponential bounds from Theorem \ref{thm:boost}, meaning that even a slight decrease
greatly influences the total security for a given number of repetitions, or equivalently the number of 
repetitions required to achieve a given security level.

However this is the case only if each test detects a set of errors disjoint from those detected by the other sets. Another way to increase the probability of detection
is to increase the coverage of each canvas by increasing the number of stabiliser errors which each can detect.
In this case, the sets can be made to overlap and the detection probability can be lowered below the upper-bound of $1-p$. We explore both approaches in the next two subsections. We now give a general process for systematising this optimisation with different constraints.

In particular situations, it might be useful to have more granular control of the design and error-detecting capabilities of the test rounds.
For instance, because of hardware constraints or ease of implementation, it might be favourable to restrict the set of tests one is willing to perform to only a subset of the tests resulting from generalised traps. As one example, one might desire to avoid the preparation of dummy states and therefore restrict the set of feasible tests to those requiring the preparation of quantum states in the $\X-\Y$-plane only. It might also not be necessary for the employed tests to detect all possible Pauli errors because of inherent robustness of the target computation.

In such cases, we can expect better error-detection rates if we (i) allow for more types of tests, or (ii) remove deviations from the set of errors that are required to be detected. To this end, we present a linear programming formulation of the search for more efficient tests in Problem~\ref{prob:optimization_distribution}.

\begin{problem}
\caption{Optimisation of the Distribution of Tests}\label{prob:optimization_distribution}
	\textbf{Given}
	\begin{itemize}
		\item a set of errors $\mathcal{E}$ to be detected,
		\item a set of feasible tests $\mathcal{H}$,
		\item a relation between tests and errors describing whether a test detects an error, $R : \mathcal{H} \times \mathcal{E} \to \bin$,
	\end{itemize}
	\textbf{find} an optimal distribution $p: \mathcal{H} \to [0,1]$ \textbf{maximising} the detection rate $\epsilon \in [0,1]$ \textbf{subject to} the following conditions:
	\begin{itemize}
		\item $p$ describes a probability distribution, i.e. $\sum_{H\in\mathcal{H}} p(H) \leq 1$,
		\item all concerned errors are detected at least with the target detection rate, i.e.
			\begin{align}
				\forall E \in \mathcal{E} : \quad \sum_{\substack{H\in\mathcal{H}\\R(H,E)=1}} p(H) \geq \epsilon.
			\end{align}
	\end{itemize}
\end{problem}

\begin{remark}
	While efficient algorithms exist to find solutions to such real-valued constrained linear problems, in this case the number of constraints grows linearly with the number of errors that need to be detected, and therefore generally exponentially in the size of the graph.
\end{remark}

\begin{remark}
	Solutions to the dual problem of Problem~\ref{prob:optimization_distribution} are distributions of deviations applied to the test rounds. An optimal solution to the dual gives therefore an optimal attack, i.e. a distribution of deviations that achieves a minimal detection rate with the tests at hand.
\end{remark}

\subsection{Standard Traps}

The simplest application of Lemma \ref{lem:prep} is to prepare qubit $i_0$ as an
eigenstate of $\X$, while its neighbours in the graph
are prepared as an eigenstate of $\Z$.
This setup can detect all deviations which do not
commute with the $\Z_{i_0}$ measurements of $i_0$.
Here, the reader familiar with the line of work
following~\cite{FK17:unconditionally} note that we have
recovered their single-qubit traps: single qubits prepared in the $\X-\Y$ plane
and surrounded by dummy $\ket 0$ or $\ket 1$ qubits.

Additionally, within each test round, it is possible to include
several such atomic traps as long as their initial states can be
prepared simultaneously -- i.e.~they can at most overlap on qubits
that need to be prepared as eigenstates of $\Z$. More precisely, take
$H$ to be an independent set of vertices from
$G$ (see Definition~\ref{def:indep_set}).
We define the set of stabilisers associated to $H$ as 
$\mathcal{R}_H = \{\cptp S^{\mathtt{i}}\}_{i \in H}$. Such sets naturally follow the 
condition of Equation \ref{eq:overlap} since $H$ is an independent set and 
therefore if $i \neq j$, $\cptp S^{\mathtt{i}}(j) = \cptp S_j(i) = \Id$ and 
both stabilisers are equal to either $\Z$ or $\Id$ for all qubits different from $i$ or $j$.
This is the extreme case where all stabilisers in $\mathcal{R}_H$ have
a single component when decomposed in the generator set $\{\cptp S^{\mathtt{i}}\}$.

Following the same line of argument as above, in absence of deviation, 
the state of qubit $i$ must be $\ket 0$ for
all $i\in H$ before the measurement, or equivalently, is an eigenstate
of $\Z_i$. Commuting these operators towards the initialisation of
the qubits shows that the qubits in $H$ must be prepared in the state
$\ket +$, and $\ket 0$ for qubits in $N_G(H)$. These qubits form the 
input set $I_T$ of the trappified canvas $T_H$ associated to the 
independent set $H$.
Other qubits can be prepared in any allowed state.
Its output locations $O_T$ are the independent
set $H$.

Using the formula from Equation \ref{eq:general-tau} for set $\mathcal{R}_H$,
we get $\tau(t) = \bigwedge_{i\in H} t_i$ for the decision algorithm. That is, the trappified canvas
accepts whenever all outcomes $\Z$ measurements for qubits $i\in
H$ are $0$.

A trappified canvas $T_H$ generated in this way depends only on the choice of
independent set $H$. Such trappified canvases will be called \emph{standard trap} in the remaining of this work.

Let $\{H^{(j)}\}_j$ be a set of independent sets. 
Since $\mathcal{R}_{H^{(j)}}$ contains all stabilisers $\cptp S^{\mathtt{i}}$ for $i \in H^{(j)}$,
the sets $\mathcal{R}_{H^{(j)}}$ cover the generating set of stabiliser
$\{\cptp S^{\mathtt{i}}\}_{i \in V}$ entirely if and only if
each qubit $i \in V$ is in at least one of the
independent sets $H^{(j)}$. Then one can conclude that all $\X$ and $\Y$
deviations have a non-zero probability of being detected, while $\Id$
and $\Z$ deviations are never detected, but are harmless 
for classical output computations.

\subsubsection{Optimising Standard Traps.}
The background in graph theory and graph colourings necessary for this section can be found in~\ref{sec:graph_colourings}.

The crucial parameter to optimise is the detection probability of individual test rounds with respect to $\X$ deviations.
In other words, the performance of the scheme will vary depending on the choice of probability distribution over the independent set $\mathcal I(G)$ and the detection capability of each individual test round.

A test round, and therefore its corresponding trappified canvas, will detect a Pauli error if and only if at least one of the $\ket +$-states is hit by a local $\X$ or $\Y$ deviation.

\begin{lemma}[Detection Rate]
  Let $G=(V,E)$ be an undirected graph. Let $\pd D$ be a probability distribution over $\mathcal I(G)$, giving rise to the trappified scheme $\sch P$ where every element of $\mathcal{I}(G)$ describes one trappified canvas. We define the \emph{detection rate} of $\pd D$ over $G$ as
  \begin{align}
    p_\text{det}(\pd D) = 1 - \epsilon(\pd D) = \min_{\substack{M \subseteq V \\ M \neq \emptyset}} \;\Pr_{H \sim \pd D} \left[ M \cap H  \neq \emptyset \right].
  \end{align}
  Then $\sch P$ $\epsilon(\pd D)$-detects the error set $\mathcal{E}_\epsilon = \{\cptp I,\X,\Y,\Z \}^{\otimes V} \setminus \{ \cptp I,\Z \}^{\otimes V}$.
\end{lemma}
\begin{proof}
	The trappified canvas induced by the independent set $H \in \mathcal{I}(G)$ detects an error $\cptp E$ if and only if $M \cap H \neq \emptyset$, where $M$ is the set of all vertices on which $\cptp E$ reduces to the Pauli-$\X$ or $\Y$. The claim is then implied by Lemma~\ref{lem:compos-sch}.
\end{proof}
In the definition above, $H$ corresponds to a choice of test round, while $M$ is the set of qubits that are affected by to-be-detected $\X$ and $\Y$ deviations.

To obtain the lowest overhead, the distribution $\pd D$ should be chosen such that it maximises the detection probability $1 - \epsilon(\pd D)$ for a given graph $G$. 
The following characterisation of the detection rate is going to be useful to determine upper bounds on $p_\text{det}$.

\begin{remark}\label{remark:pdet_char}
  For any graph $G$ and any distribution $\pd D$ over $\mathcal I(G)$ it holds that
  \begin{align}
    p_\text{det}(\pd D) = \min_{\pd M} \;\Pr_{\substack{M \sim \pd M\\ H \sim D}} \left[ M \cap H \neq \emptyset \right],
  \end{align}
  where the minimum ranges over distributions $\pd M$ over $\wp(V) \setminus \{\emptyset\}$.
\end{remark}

It can be shown that the best achievable detection rate by standard traps for a graph $G$ lie in the interval $\left[ \frac{1}{\chi(G)} , \frac{1}{\omega(G)} \right]$, where $\chi(G)$ and $\omega(G)$ are respectively the chromatic number and the clique number of $G$. The protocol of \cite{LMKO21:verifying} in particular is designed with security bounds depending on the chromatic number of the underlying graph.
Note that the two graph invariants $\chi(G)$ and $\omega(G)$ are dual in the sense that they are integer solutions to dual linear programs and the gap between these two values can be large (see Lemma \ref{lem:gap}).
It turns out that both bounds can be improved to depend on the solutions of the relaxations of the respective linear programs.
This closes the integrality gap between the chromatic number and the clique number.

\begin{lemma}\label{lemma:lower_bound_fractional_chromatic}
  For every (non-null) graph $G$ there exists a distribution $\pd D$ over $\mathcal I(G)$ such that $p_\text{det}(\pd D) \geq \frac{1}{\chi_f(G)}$, where $\chi_f(G)$ is the fractional chromatic number of $G$ (see Definition~\ref{def:fractional_chromatic_number}).
\end{lemma}
\begin{proof}
  Let $\pd D$ be a distribution over $\mathcal I(G)$ such that for all $v \in V$ it holds that
$\Pr_{H \sim \pd D} \left[ v \in H \right] \geq \frac{1}{k}.$
For all $M \subseteq V, M \neq \emptyset$, then
$\Pr_{H \sim \pd D} \left[ M \cap H \neq \emptyset \right] \geq \frac{1}{k}$
and therefore
$p_\text{det}(\pd D) \geq \frac{1}{k}.$
By Lemma~\ref{lemma:fractional_chromatic_number_probabilistic}, we can find such a distribution $\pd D$ for any $k \geq \chi_f(G)$.
\end{proof}

We can also improve the upper bound using fractional cliques.

\begin{lemma} \label{lemma:upper_bound_fractional_clique}
  For every (non-null) graph $G$ and every distribution $\pd D$ over $\mathcal I(G)$ it holds that $p_\text{det}(\pd D) \leq \frac{1}{\omega_f(G)}$, where $\omega_f(G)$ is the fractional clique number of $G$ (see Definition~\ref{def:fractional_clique_number}).
\end{lemma}
\begin{proof}
  This statement is a direct consequence of Lemma~\ref{lemma:fractional_clique_number_probabilistic}.
\end{proof}

As a consequence, this shows that the protocol described in~\cite{LMKO21:verifying}, which is the current state-of-the-art, can sometimes be improved by constructing additional test rounds that would allow to have a probability of detection greater than the reported $1/\chi(G)$. In fact, this proves that the best possible detection rate by standard traps is equal to $1/\chi_f(G)$ since $\chi_f(G) = \omega_f(G)$ by Lemma~\ref{lemma:fractional_equality}. This is achieved precisely by choosing the set of possible tests to be a fractional colouring of the graph.

\begin{example}\label{example:pentagon_trap_distribution}
	Let $G=(V,E)$ be the cycle graph on $5$ nodes with $V = \{0,1,2,3,4\}$. An optimal proper $3$-colouring of $G$ is given by $\left(\{0,2\},\{1,3\},\{4\}\right)$, which gives rise to a standard trap with detection rate $1/3$.
	However, this may be further improved using Lemma~\ref{lemma:lower_bound_fractional_chromatic} and the fact that $\chi_f(G) = 5/2$. A standard trap with the optimal detection rate of $2/5$ is given by the uniform distribution over the set $\left\{ \{0,2\}, \{1,3\}, \{2,4\}, \{0,3\}, \{1,4\} \right\}$.
\end{example}

\begin{figure}
	\centering
	\begin{subfigure}[b]{0.36\textwidth}
		\centering
		\includegraphics[height=1.4cm]{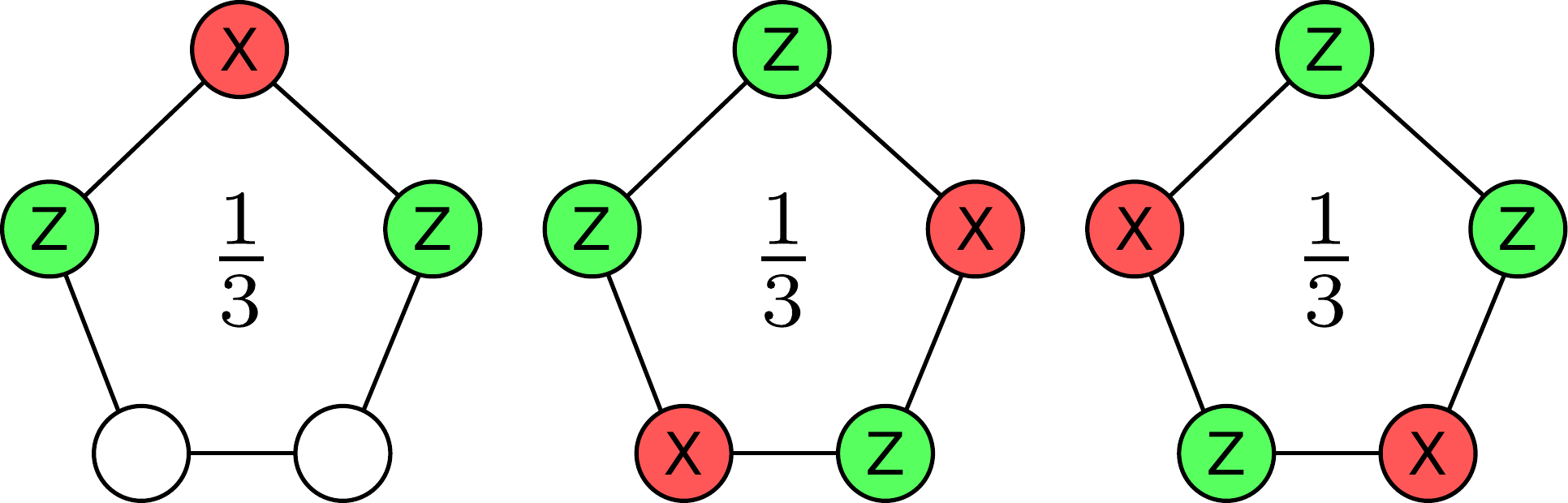}
		\caption{Trap distribution based on an optimal colouring of $G$.}
	\end{subfigure}
	\hfill
	\begin{subfigure}[b]{0.6\textwidth}
		\centering
		\includegraphics[height=1.4cm]{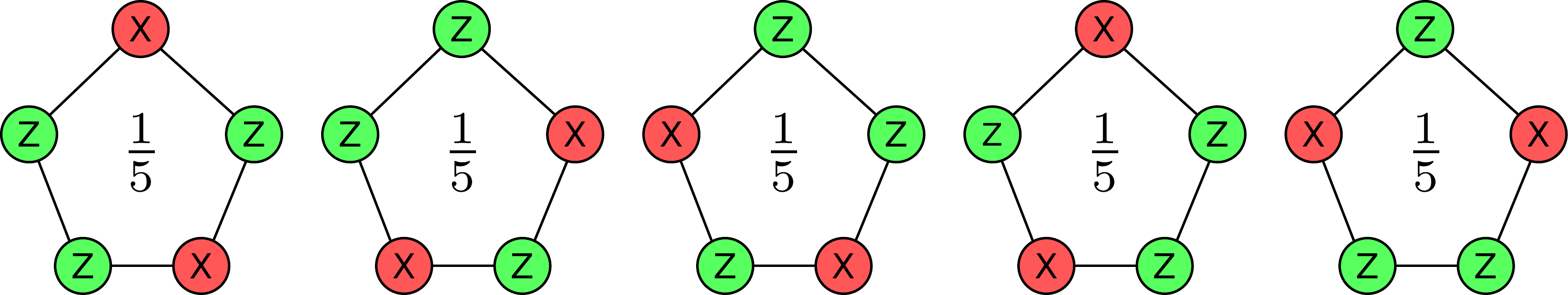}
		\caption{Trap distribution based on an optimal fractional colouring of $G$.}
	\end{subfigure}
	\caption{Traps on the cycle graph $G$ with $5$ nodes from Example~\ref{example:pentagon_trap_distribution}.}
	\label{fig:pentagon_trap_distribution}
\end{figure}

Yet, this leaves a dependency of the protocol's efficiency on graph invariants, meaning that depending on the chosen computation, the protocol could perform poorly.
The next section shows how to overcome this obstacle, as long as the client is willing to use more generalised traps.

\subsection{General Traps}
\label{subsec:gen-trap}

Above, the trappified canvases we obtained are a consequence of determinism,
insensitivity to harmless deviations and a restriction on the subsets $H$,
constrained to be independent. To
construct \emph{general traps}, we simply remove this last requirement and 
define instead $\mathcal{R}_H = \{\prod_{i \in H} \cptp S^{\mathtt{i}}\}$.
Using Equation \ref{eq:general-tau}, $\tau$ is then the parity of measurement outcomes
for qubits from $H$, i.e.~$\tau(t) = \bigoplus_{i\in H}
t_i$. This means that to accept the execution of such trappified canvas, the state of
the qubits $i \in H$ needs to be in the $+1$ eigenspace of the operator
$\prod_{i \in H} \Z_i$. This is the other extreme case since there is only a 
single stabiliser in the set $\mathcal{R}_H$.

Commuting this operator to the
initialisation imposes to prepare a $+1$ eigenstate of
$\prod_{i \in H_{even}}\X_i \prod_{j\in
H_{odd}} \Y_j \prod_{k \in N_G^{odd}(H)} \Z_k$, where $H_{even}$
(resp. $H_{odd}$) are the qubits of even (resp. odd) degree within
$H$, and $k \in N_G^{odd}(H)$ means $k$ is in the odd neighbourhood of
$H$. Again, applying Lemma~\ref{lem:prep} allows us to find in the
eigenspace of this operator a state that can be obtained as a tensor
product of single-qubit states, simply by looking at the individual
Paulis from the operator $\prod_{i \in H} \cptp S^{\mathtt{i}}$.

It is easy to see that this trappified canvas detects all deviations that
anti-commute with $\prod_{i \in H} \Z_i$, that is deviations that
have an odd number of $\X$ or $\Y$ for qubits in $H$.
Varying the sets $H$ allows to construct a trappified scheme which
detects all possible deviations containing any number of $\X$ or $\Y$ with a constant probability.

\subsubsection{Optimising General Traps.}
General traps are based on test rounds defined by a set $H \subseteq V$ of qubit locations.
It accepts whenever the parity of outcomes of $\Z$-measurements on the qubits of $H$ is even.
Here the testing set $H$ can be chosen freely and does not need to be independent as in the construction of standard traps.

\begin{lemma}[General Stabiliser-Based Trappified Scheme]
\label{lem:gen-trap-sch}
Let $\sch P$ be the trappified scheme defined by sampling uniformly at random a non-empty set $H \subseteq V$ and preparing the trappified canvas associated to $\mathcal{R}_H = \{\prod_{i \in H} \cptp S^{\mathtt{i}}\}$. Then $\sch P$ $1/2$-detects the error set $\mathcal{E}_\epsilon = \{\cptp I,\X,\Y\}^{\otimes V} \setminus \{ \cptp I^{\otimes V} \}$.
\end{lemma}
\begin{proof}
Looking at a given deviation $\cptp E$, we conclude that a test-round defined by $H$ detects $\cptp E$ if and only if $|E \cap H|$ is odd -- here $E$ denotes the set of qubits where $\cptp E$ is equal to $\X$ or $\Y$.
If $H$ is sampled uniformly at random from $\wp(G)$, then $\Pr_{H\sim \pd U(\wp(G))}\left[ |E \cap H| \equiv 1 \mod 2 \right] = 1/2$, and this is valid for any $\cptp E\neq \cptp I$.
\end{proof}

As a conclusion, we obtain that the probability of detection for this scheme is equal to $1/2$, which is independent of the graph $G$, and generally will beat the upper bound obtained in the previous section through standard traps.

\section{Discussion and Future Work}
\label{sec:concl}
We uncovered a profound correspondence between error-detection and verification that applies and unifies all previous trap-based blind verification schemes in the prepare-and-send MBQC model, which covers the majority of proposed protocols from the literature.  In addition, all results mentioned here also apply to receive-and-measure MBQC protocols via the recent equivalence result from \cite{WEP22:equivalence}.
On the theoretical side, it provides a direct and generic composable security proof of these protocols in the AC framework, which also gives the first direct and explicit proof of composability of the original VBQC protocol~\cite{FK17:unconditionally}.
We also formally showed that error-correction is required if one hopes to have negligible correctness and security errors with polynomial overhead when comparing unprotected and unverified computations to their secure counterparts.
On a practical side, this correspondence can be used to increases the tools available to design, prove the composable security, and optimise the performance of new protocols.
To exemplify these new possibilities, we described new protocols that improve the overhead of state-of-the-art verification protocols, thus making them more appealing for experimental realisation and possibly for integration into future quantum computing platforms.

The uncovered connection between error-detection and verification raises new questions such as the extent to which it is possible to infer from the failed traps what the server has been performing.
Additionally, Theorem~\ref{thm:encoding} implies that some form of error-correction is necessary to obtain exponential correctness.
Yet, the protocol of~\cite{LMKO21:verifying} and the result presented here shows that sometimes classical error-correction is enough, thereby raising the question of understanding what are the optimal error-correction schemes for given classes of computations that are to be verified.

Finally, we the link between error detection and verification can be further developed and yield new trappified schemes with not only more efficient implementations but also additional capabilities. A recent results~\cite{KKLM+23:quantum} presents a novel trappified scheme which does not require the client to prepare dummy qubits. This further not only reduces the capabilities required of the client to perform secure delegated quantum computations, but also allows them to construct a highly optimised implementation of quantum secure multi-party computation where mutually distrustful clients delegate a joint computation to a quantum server. Their protocol only requires each client to send exactly as many qubits as in one instance of the single-client protocol, there is no additional overhead. This result demonstrates the power of the construction presented here, in particular the trappified pattern compiler and generalised stabiliser traps.

\bibliographystyle{iopart-num}
\section*{References}
\bibliography{../qubib/qubib.bib}

\appendix

\section{Additional Preliminaries on Flow}
\label{app:flow}
In MBQC, it can be shown that measurements in the $\X - \Y$ plane on non-output qubits will always yield $0$ or $1$ with probability $1/2$. We consider the outcome $0$ as the ``correct'' outcome. In that case, obtaining $1$ is equivalent to applying a $\Z$ operation right before the measurement and obtaining $0$. The goal of the flow function is to describe how to remove the effect of this additional $\Z$ operation and thus recover the same computation as if the measurement outcome had been $0$.

This is done by relying on the stabilisers of the graph state $\ket{G}$. These are unitary operators $\cptp S$ such that $\cptp S\ket{G} = \ket{G}$, their set is a group generated by the Pauli operators $\cptp S^{\mathtt{i}} = X_i \prod_{j \in N_G(i)} \Z_j$. The vertex $f(i)$ is then chosen such that the operator $\cptp S^{\mathtt{f(i)}}$ acts as $\Z$ at vertex $i$ -- this is the one applied by obtaining measurement outcome $1$. This means that $f(i)$ must be a neighbour of $i$, i.e.~$(i, f(i)) \in E$.

Furthermore, in order to cancel out the effect of this $\Z$, we need to be able to apply the rest of the Paulis in the operator $\cptp S^{\mathtt{f(i)}}$. This implies first that $f(i)$ should not be measured by the time $i$ is measured, which imposes the condition $i \preceq f(i)$ on the order in which vertices are measured. Finally, it also implies that all neighbours of $f(i)$ should not be measured before $i$, which adds the conditions $i \preceq j$ for all $j \in N_G(f(i)) \setminus \{i\}$.

To summarise, the flow is an injective function from non-output vertices to non-input vertices $f : V \setminus O \rightarrow V \setminus I$.  Together with the ordering $\preceq$, it must satisfy the following conditions:
\begin{enumerate}
\item $(i, f(i)) \in E$;
\item $i \preceq f(i)$;
\item $\forall j \in N_G(f(i)) \setminus \{i\}, i \preceq j$.
\end{enumerate}

The existence of a flow and associated ordering is a sufficient condition for pattern determinism, meaning that MBQC computations with flow always implement the same unitary transformation from the input to the output layers regardless of the measurement outcomes. A generalisation called the g-flow which uses multiple stabilisers $\cptp S_i$ for each vertex provides a sufficient and necessary condition for determinism. 
\section{Graph Colourings}
\label{sec:graph_colourings}
In this section, we introduce graph colourings and recall some known related results that are useful to our theory.

\begin{definition}[Independent Set]\label{def:indep_set}
  Let $G=(V,E)$ be a graph. Then a set of vertices $t \subseteq V$ is called an \emph{independent set} of $G$ if
  \begin{align}
    \forall v_1, v_2 \in t: \{v_1, v_2 \} \not\in E.
  \end{align}
  The size of the largest independent set of $G$ is called the \emph{independence number} of $G$ and denoted by $\alpha(G)$. The set of all independent sets of $G$ is denoted $I(G)$.
\end{definition}

\begin{definition}[Graph Colouring]
  Let $G=(V,E)$ be a graph. Then a collection of $k$ pairwise disjoint independent sets $H_1, \dots, H_k \subseteq V$ such that $\bigcup_{j=1}^k H_j = V$ is called a \emph{(proper) $k$-colouring} of $G$. The smallest number $k\in\mathbb{N}_0$ such that $G$ admits a $k$-colouring is called the chromatic number of $G$ and denoted by $\chi(G)$.
\end{definition}

\begin{definition}[Clique]
  Let $G=(V,E)$ be a graph.
  Then a complete subgraph $C \subseteq V$ of size $k$ is called a $k$-\emph{clique} of $G$. The largest number $k\in\mathbb{N}_0$ such that $G$ admits a $k$-clique is called the clique number of $G$ and denoted by $\omega(G)$.
\end{definition}

\begin{lemma} \label{lem:gap}
  For any graph $G$ it holds that $\omega(G) \leq \chi(G)$.
  For any $n\in\mathbb{N}$, there exists a graph $G_n$ such that $\chi(G_n) - \omega(G_n) \geq n$.
\end{lemma}

\begin{definition}[Fractional Graph Colouring]\label{def:fractional_chromatic_number}
  Let $G=(V,E)$ be a graph. For $b \in \mathbb{N}$, a collection of independent sets $H_1, \dots, H_k \subseteq V$, such that for all $v \in V : |\{ 1 \leq j \leq k \; | \; v \in H_j \}| = b$, is called a \emph{$k$:$b$-colouring} of $G$.
  The smallest number $k \in \mathbb{N}_0$ such that $G$ admits a \emph{$k$:$b$-colouring} is called the $b$-fold chromatic number of $G$ and denoted by $\chi_b(G)$.
  Since $\chi_b(G)$ is subadditive we can define the \emph{fractional chromatic number} of $G$ as
  \begin{align}
    \chi_f(G) = \lim_{b \to \infty} \frac{\chi_b(G)}{b} = \inf_{b \in \mathbb{N}} \frac{\chi_b(G)}{b}.
  \end{align}
  Note that $k$:$1$-colourings are $k$-colourings and therefore $\chi_1(G) = \chi(G)$ which in turn implies that for all $b \in \mathbb{N}$ it holds that
  \begin{align}
    \chi_f(G) \leq \chi_b(G) \leq \chi(G).
  \end{align}
\end{definition}

\begin{lemma}\label{lemma:fractional_chromatic_number_probabilistic} Let $G=(V,E)$ be a graph.
  Then $\chi_f(G)$ equals the smallest number $k \in \mathbb{R}^+_0$ such that there exists a probability distribution $\pd D$ over the independent sets $\mathcal I(G)$ such that for all $v \in V$ it holds that
  \begin{align}
    \Pr_{H \gets \pd D} \left[ v \in t \right] \geq \frac{1}{k}.
  \end{align}
\end{lemma}

\begin{definition}[Fractional Clique]\label{def:fractional_clique_number}
  Let $G=(V,E)$ be a graph. For $b \in \mathbb{N}$, a function $f : V \to \mathbb{N}_0$, such that for all $H \in \mathcal I(G) : \sum_{v \in H} f(v) \leq b$ and $\sum_{v \in V} = k$, is called a \emph{$k$:$b$-clique} of $G$.
  The biggest number $k \in \mathbb{N}_0$ such that $G$ admits a \emph{$k$:$b$-clique} is called the $b$-fold clique number of $G$ and denoted by $\omega_b(G)$.
  Since $\chi_b(G)$ is superadditive we can define the \emph{fractional clique number} of $G$ as
  \begin{align}
    \omega_f(G) = \lim_{b \to \infty} \frac{\omega_b(G)}{b} = \sup_{b \in \mathbb{N}} \frac{\omega_b(G)}{b}.
  \end{align}
  Note that $k$:$1$-cliques are $k$-cliques and therefore $\omega_1(G) = \omega(G)$ which in turn implies that for all $b \in \mathbb{N}$ it holds that
  \begin{align}
    \omega(G) \leq \omega_b(G) \leq\omega_f(G).
  \end{align}
\end{definition}

\begin{lemma}\label{lemma:fractional_clique_number_probabilistic} Let $G=(V,E)$ be a graph.
  Then $\omega_f(G)$ equals the biggest number $k \in \mathbb{R}^+_0$ such that there exists a probability distribution $\pd D$ over the vertices $V$ such that for all $H \in \mathcal I(G)$ it holds that
  \begin{align}
    \Pr_{v \gets \pd D} \left[ v \in H \right] \leq \frac{1}{k}.
	\end{align}
\end{lemma}

Both the fractional clique number $\omega_f(G)$ and the fractional chromatic number $\chi_f(G)$ are rational-valued solutions to dual linear programs.
By the strong duality theorem, the two numbers must be equal.

\begin{lemma}\label{lemma:fractional_equality} For any graph $G$ it holds that $\omega_f(G) = \chi_f(G)$.
\end{lemma}
 
\section{General Parallel Repetition}
\label{app:gen-comp}
We here show an alternative method for performing the same decomposition as in Section~\ref{sec:rvbqc}, by focusing solely on the error-detection amplification of classical input computations. We then recover the results from that section as a consequence of this generic amplification. We start as before by defining a compiler taking as input a trappified scheme and running it several times in parallel before thresholding over the outcomes of the individual decision functions.

\begin{definition}[Parallel Repetition Compiler]\label{def:par_compiler}
Let $(\sch P, \preceq_G, \pd P, E_{\mathfrak{C}})$ be trappified scheme over a graph $G$ for computation class $\mathfrak{C}$ with classical inputs, and let $n \in \mathbb{N}$ and $w \in [n-1]$. We define the Parallel Repetition Compiler that turns $\sch P$ into a trappified scheme $\sch P_{\parallel n}$ on $G^n$ for computation class $\mathfrak{C}$ as follows:

\begin{itemize}
	\item The set of trappified canvases is defined as $\{T_{\parallel n}\} = \sch P_{\parallel n} = \sch P^n$, the distribution $\pd P_{\parallel n}$ samples $n$ times independently from $\pd P$;
	\item For each trappified canvas $T'$ defined above and an output $t = (t_{j})_{j \in n}$, we have:
	\begin{equation}
		\tau'(t) = 
0 \text{ if } \sum_{j = 1}^n \tau_{j}(t_j) < w, \ \text{and } 1 \text{ otherwise;}
\end{equation}
	\item The partial ordering of vertices of $G^n$ in $\sch P_{\parallel n}$ is given by the ordering $\preceq_G$ on every copy of $G$;
	\item Let $\cptp C \in \mathfrak{C}$. Given a trappified canvas $T_{\parallel n} = \{T_j\}_{j \in [n]}$, the embedding algorithm $E_\mathfrak{C, \parallel n}$ applies $E_\mathfrak{C}$ to embed $\cptp C$ in each $T_j$.
\end{itemize}
\end{definition}

The next theorem relates the parameters above to the detection and insensitivity of the compiled scheme.

\begin{theorem}[Exponential Detection and Insensitivity from Parallel Repetitions]\label{thm:par-boost}
  Let $\sch P$ be a trappified scheme on graph $G$ which $\epsilon$-detects the error set $\mathcal{E}_\epsilon$, is $\delta$-insensitive to $\mathcal{E}_\delta$ and perfectly insensitive to $\{\Id\}$.
  For $n \in \mathbb{N}$ and $w \in [n-1]$, let $\sch P_{\parallel n}$ be the trappified scheme resulting from the compilation defined in Definition~\ref{def:par_compiler}.
  
  Let $k_\epsilon > w/(1-\epsilon)$ and $k_\delta < w/\delta$. Then, $\sch P_{\parallel n}$ $\epsilon_{\parallel n}$-detects $\mathcal{M}_{\geq k_\epsilon}(\mathcal{E}_\epsilon)$ and is $\delta_{\parallel n}$-insensitive to $\mathcal{H}_{\leq k_\delta}(\mathcal{E}_\delta)$ where:
  
  \begin{align}
  	\epsilon_{\parallel n} &= \exp \left( -2 \frac{(k_\epsilon(1-\epsilon) - w)^2}{k_\epsilon} \right),\\
  	\delta_{\parallel n} &= \exp \left( -2 \frac{(k_\delta\delta - w)^2}{k_\delta} \right).
  \end{align}
\end{theorem}

\begin{proof}

We denote $Y$ a random variable counting the number of trappified canvases whose decision function rejects.

Let $\cptp E \in \mathcal{M}_{\geq k_\epsilon}(\mathcal{E}_\epsilon)$. We can lower-bound $Y$ in the usual stochastic order by a $(k_\epsilon, 1 - \epsilon)$-binomially distributed random variable $\tilde{Y}$. Since $k_\epsilon(1-\epsilon) > w$, Hoeffding's inequality yields directly that:

\begin{align}
	\Pr[Y < w] \leq \Pr[\tilde{Y} < w] \leq \exp \left( -2 \frac{(k_\epsilon (1 - \epsilon) - w)^2}{k_\epsilon} \right).
\end{align}

Similarly, let $\cptp E \in \mathcal{H}_{\leq k_\delta}(\mathcal{E}_\delta)$. Due to the perfect insensitivity of $\sch P$ to $\Id$, we can now upper-bound $Y$ in the usual stochastic order by a $(k_\delta, \delta)$-binomially distributed random variable $\tilde{Y}$. For $k_\delta\delta < w$, Hoeffding's inequality yields directly that:

\begin{align}
	\Pr[Y \geq w] \leq \Pr[\tilde{Y} \geq w] \leq \exp \left( -2 \frac{(k_\delta\delta - w)^2}{k_\delta} \right).
\end{align}

\end{proof}

We can also obtain a similar boosting result for correctness if we restrict the computations to $\mathsf{BQP}$. In that case, the correctness of a canvas is a bound on the probability that the classical result is not correct. If we bound the probability that more than a certain number of computations fail and this number is chosen such that the decoder is capable of correcting this amount of failure with overwhelming probability, then the overall failure probability will be negligible. This is captured by the following theorem.

\begin{theorem}[Exponential Correctness for $\mathsf{BQP}$ from Parallel Repetitions]\label{thm:par-boost-cor}
  Let $\sch P$ be a trappified scheme on graph $G$ which is $\nu$-correct for $\mathfrak{C} = \mathsf{BQP} \cap \mathsf{MBQC}_{G, \preceq_{\sch P}}$ on an error set $\mathcal{E}_\nu \subseteq \{\cptp E \in \mathcal{P}_V \mid \forall \cptp C \in \mathfrak{C}, \forall T \in \sch P, \Cdev{\cptp E} = \Cdev{\cptp \Id}\}$.
  
  For $n \in \mathbb{N}$ and $w \in [n-1]$, let $\sch P_{\parallel n}$ be the trappified scheme resulting from the compilation defined in Definition~\ref{def:par_compiler}.
  Assume there exists a decoding algorithm $\Deco^{\parallel n}$ which succeeds if there are more than $f_{\parallel n}$ correct results among the $n$ parallel runs.
  
  Let $k_\nu < \frac{1 - \nu - f_{\parallel n}/n}{1 - \nu}  n$. Then, $\sch P_{\parallel n}$ is $\nu_{\parallel n}$-correct on $\mathcal{M}_{\geq n - k_\nu}(\mathcal{E}_\nu)$ where:
  
  \begin{align}
  	\nu_{\parallel n} = \exp \left( -2 \frac{((n - k_\nu)(1 - \nu) - f_{\parallel n})^2}{n - k_\nu} \right).
  \end{align}
\end{theorem}

\begin{proof}
In the case of $\mathsf{BQP}$ computations, $\nu_{\parallel n}$ can be seen as a bound on the probability that the computation's classical result is incorrect. The failure can come from two places: either the number of correct computations is below $f_{\parallel n}$, or it is above and the decoder fails. This second option, by assumption, happens with probability at most $p_{\parallel n}$. We now bound the probability of first option.

Let $Z$ be a random variable counting the number of correct runs and let $\cptp E \in \mathcal{M}_{\geq n - k_\nu}(\mathcal{E}_\nu)$. We can lower-bound $Z$ in the usual stochastic order by a $(n - k_\nu, 1 - \nu)$-binomially distributed random variable $\tilde{Z}$. Then, for $(n - k_\nu)(1 - \nu) > f_{\parallel n}$, Hoeffding's inequality yields directly that:

\begin{align}
	\Pr[Z < f_{\parallel n}] \leq \Pr[\tilde{Z} < f_{\parallel n}] \leq \exp \left( -2 \frac{((n - k_\nu)(1 - \nu) - f_{\parallel n})^2}{n - k_\nu} \right).
\end{align} 

\end{proof}

\paragraph{Test and Computation Separation from Parallel Repetitions.}
We can now recover the case where some runs contain tests only while others consist only of the client's computation. This will be based on the following remark

\begin{remark}[Pure Computation]
\label{rem:pure_comp}
A trappified scheme $\sch P$ may also only contain a single trappified canvas on graph $G = (V, E)$ such that $V_T = \emptyset$. This is the opposite case from Remark \ref{rem:pure_traps} above in the sense that all vertices serve to embed a computation of interest and none are devoted to detecting deviations. The decision function always accepts, therefore the detection and insensitivity are $\epsilon = 1$ and $\delta = 0$ respectively for any set. Let $c$ be the bounded-error of computation on $\mathfrak{C} = \mathsf{BQP} \cap \mathsf{MBQC}_{G, \preceq_{\sch P}}$, then $\sch P$ is $c$-correct on $\mathcal{E}_\nu \subseteq \{\cptp E \in \mathcal{P}_V \mid \forall \cptp C \in \mathfrak{C}, \Cdev{\cptp E} = \Cdev{\cptp \Id}\}$, i.e.~the set of harmless deviations.

\end{remark}

We then use Remarks \ref{rem:pure_traps} and \ref{rem:pure_comp} to define trappified schemes $\sch P_C$ and $\sch P_T$ on a graph $G$. $\sch P_C$ contains a single empty trappified canvas (with no trap) which can then be used to embed any computation on graph $G$, with $1$-detection and $0$-insensitivity to all Paulis. On the other hand, $\sch P_T$ only contains pure traps with no space for embedding any computation, which $\epsilon$-detects a set of errors $\mathcal{E}_\epsilon$, is $\delta$-insensitive to $\mathcal{E}_\delta$ and perfectly insensitive to $\{\Id\}$.

Then, Lemma \ref{lem:compos-sch} allows us to compose these two schemes via a probabilistic mixture noted $\sch P_M$. For parameters $d,s\in \mathbb{N}$ and $n=d+s$, $\sch P_M$ chooses schemes $\sch P_C$ and $\sch P_T$ with probabilities $d/n$ and $s/n$ respectively. The parameters for $\sch P_M$ are $\epsilon_M = (d + s\epsilon)/n = 1 - (1-\epsilon)s/n$, $\delta_M = s\delta/n$ and $\nu_M = (s + dc)/n = 1 - (1-c)d/n$. It is also perfectly insensitive to $\{\Id\}$.

For the correctness, we set $f_{\parallel n} = d/2$. Instead of fixing the total number $n$, we sample a new canvas from $\sch P_M$ so long as the number of computations is not equal to $d$. The decoder outputs the majority outcome over the $d$ computation runs, which satisfies the correctness requirement from Theorem \ref{thm:par-boost-cor}.

Then the parallel repetition of Theorems \ref{thm:par-boost} and \ref{thm:par-boost-cor} can be applied to $\sch P_M$ with parameter $w < s$ to yield $\sch P_{\parallel n}$ with the following boosted values:
\begin{align}
	\epsilon_{\parallel n} &= \exp \left( -2 \left(\frac{k_\epsilon (1-\epsilon)}{n} - \frac{w}{s}\right)^2 \frac{s^2}{k_\epsilon} \right),\\
	\delta_{\parallel n} &= \exp \left( -2 \left(\frac{k_\delta \delta}{n} - \frac{w}{s}\right)^2 \frac{s^2}{k_\delta} \right),\\
	\nu_{\parallel n} &= \exp \left( -2 \left(\left(1 - \frac{k_\nu}{n}\right)(1 - c) - \frac{1}{2}\right)^2 \frac{d^2}{n - k_\nu} \right),
\end{align}
for values $k_\epsilon > \frac{w}{(1-\epsilon)s}n$, $k_\delta < \frac{w}{s\delta}n$ and $k_\nu < \frac{1 - 2c}{2 - 2c}n$.

Notice that the bounds on $k_\epsilon, k_\delta, k_\nu$ are identical to those from Theorems \ref{thm:boost} and \ref{thm:exp-cor}, hence they can be recombined as in Theorem \ref{thm:comb}. The values for $\epsilon_{\parallel n}, \delta_{\parallel n}, \nu_{\parallel n}$ obtained here are also simpler than those from Theorems~\ref{thm:boost} since they do not require an optimisation over the parameter $\chi$, while still being exponential. This parallel repetition compiler has the advantage of also working in cases where each trappified pattern in the scheme can also embed computations, making it more general than the one in Definition~\ref{def:rvbqc_bqp_compiler}.  
\end{document}